\documentclass[11pt]{article}

\usepackage{amssymb}
\usepackage{amsfonts}
\usepackage{amsmath}
\usepackage{amssymb}
\usepackage{color}
\usepackage{graphicx}
\usepackage{caption}
\usepackage{subcaption}
\usepackage{multicol}
\usepackage[svgnames]{xcolor}
\usepackage{listings}
 \usepackage[utf8x]{inputenc} 
\newtheorem{theorem}{Theorem}

\newtheorem{definition}[theorem]{Definition}
\newtheorem{example}[theorem]{Example}

\newtheorem{remark}[theorem]{Remark}

\newenvironment{proof}[1][Proof]{\noindent \textbf{#1.} }{\  \rule{0.5em}{0.5em}}
\begin{document}

\title{\textbf{New robust statistical procedures for polytomous logistic
regression models}}
\author{Castilla, E.$^{1}$; Ghosh, A$^{2}$; Mart\'{\i}n, N.$^{1}$ and Pardo, L.$%
^{1}$ \\
$^{1}${\small Complutense University of Madrid, 28040 Madrid, Spain} \\
$^{2}${\small Indian Statistical Institute, Kolkata, India}}
\date{\today }
\maketitle

\begin{abstract}
This paper derives a new family of estimators, namely the minimum density power divergence estimators, 
as a robust generalization of the maximum likelihood estimator for the polytomous logistic regression model. 
Based on these estimators, a family of Wald-type test statistics for linear hypotheses is introduced. 
Robustness properties of both the proposed estimators and the test statistics are theoretically studied 
through the classical influence function analysis. 
Appropriate real life  examples are presented to justify the requirement of  
suitable robust statistical procedures in place of the likelihood based inference for  
the polytomous logistic regression model. The validity of the theoretical results 
established in the paper are further confirmed empirically through suitable simulation studies.  
Finally, an approach for the data-driven selection of the robustness tuning parameter is  
proposed with empirical justifications.
\end{abstract}

\noindent \textbf{MSC}{\small : 62F35, 62F05 }

\noindent \textbf{Keywords}{\small : Influence function; Minimum density power divergence estimators; Polytomous logistic regression; Robustness;  Wald-type test statistics.}

\section{Introduction\label{S1}}

The polytomous logistic regression model (PLRM) is widely used in health and life sciences 
for analyzing nominal qualitative response variables
(e.g., Daniels and Gatsonis, 1997; Blizzard and Hosmer, 2007; Bull, Lewinger  and Lee, 2007; 
Dreassi, 2007; Biesheuvel et al., 2008; Bertens et al., 2015; Dey, Raheem and Lu, 2016; 
Ke, Fu and Zhang, 2016, and the references therein). 
Such examples occur frequently in medical studies where  
disease symptoms  may be classified as absent, mild or severe,
the invasiveness of a tumor may be classified as in situ, locally invasive, or metastatic, etc.
The qualitative response models specify the multinomial distribution for such a response variable with
individual category probabilities being modeled as a function of suitable explanatory variables. 
One such popular model is the PLRM, where the logit function is used 
to link the category probabilities with the explanatory variables.

Mathematically, let us assume that the nominal
outcome variable $\tilde{Y}$ has $d+1$ categories $C_{1},...,C_{d+1}$ and we
observe $\tilde{Y}$ together with $k$ explanatory variables with given values $x_{h}$, $%
h=1,...,k$. In addition, assume that $\boldsymbol{\beta }_{j}^{T}=\left(
\beta _{0j},\beta _{1j},...,\beta _{kj}\right) ,$ $j=1,...,d$, is a vector
of unknown parameters and $\boldsymbol{\beta }_{d+1}$ is a $(k+1)$%
-dimensional vector of zeros; i.e., the last category $C_{d+1}$ has been
chosen as the baseline category. Since the full parameter vector $%
\boldsymbol{\beta }^{T}=(\boldsymbol{\beta }_{1}^{T},...,\boldsymbol{\beta }%
_{d}^{T})$ is $\nu $-dimensional with $\nu =d(k+1)$, the parameter space is 
$\Theta =%
\mathbb{R}
^{d(k+1)}$. Let $\pi _{j}\left( \boldsymbol{x},\boldsymbol{\beta }\right) =P(%
\tilde{Y}\in C_{j}\mid \boldsymbol{x},\boldsymbol{\beta }{)}$ denote the
probability that $\tilde{Y}$ belongs to the category $C_{j}$ for $%
j=1,...,d+1,$ when the vector of explanatory variable takes the value $%
\boldsymbol{x}^{T}=(x_{0},x_{1},\ldots ,x_{k})$, 
with $x_{0}=1$ being associated with the intercept ${\beta }_{0j}$. 
Then, the PLRM is given by 
\begin{equation}
\pi _{j}\left( \boldsymbol{x},\boldsymbol{\beta }\right) =\exp (\boldsymbol{%
	x}^{T}\boldsymbol{\beta }_{j})\left/ \left\{ 1+\sum_{h=1}^{d}\exp (
\boldsymbol{x}^{T}\boldsymbol{\beta }_{h})\right\} ,\right. \text{ }%
j=1,...,d+1.  \label{1}
\end{equation}%

Now assume that we have observed the data on $N$ individuals having responses $\tilde{\boldsymbol{y}}_i$ 
with associated covariate values (including intercept)  $\boldsymbol{x}_{i}\in 
\mathbb{R}
^{k+1}$, $i=1,...,N$, respectively. 
For each individual, let us introduce the corresponding tabulated response 
$\boldsymbol{y}_{i}=\left( y_{i1},...,y_{i,d+1}\right) ^{T}$ with $y_{ir}$ $%
=1 $ and $y_{is}=0$ for $s\in \left\{ 1,...,d+1\right\} -\left\{ r\right\} $
if $\tilde{\boldsymbol{y}}_i\in C_{r}$. The most common estimator of $%
\boldsymbol{\beta }$ under the PLRM is the maximum likelihood estimator (MLE),
which is obtained by maximizing the loglikelihood function, $\log 
\mathcal{L}%
\left( \boldsymbol{\beta }\right) \equiv
\sum_{i=1}^{N}\sum_{j=1}^{d+1}y_{ij}\log \pi _{j}\left( \boldsymbol{x}_{i},\boldsymbol{\beta }\right)$. 
One can then develop all the subsequent inference procedures based on 
the MLE $\widehat{\boldsymbol{\beta }}$ of $\boldsymbol{\beta }$.
Although the MLE has optimal asymptotic properties in relation to the efficiency for most cases, 
its serious lack of robustness against the outlying observations  is also a well-known problem. 
However, in any practical dataset it is quite natural to have outlying observations 
which can lead to incorrect inference for the likelihood based approach 
and can be very dangerous specially in applications like medical sciences. 
The above formulation of the PLRM is not exclusive only for distinct covariate values; 
it can also be applied, with the same notation, if  multiple responses are observed for the same covariate values. 
Let us begin with the following motivating example. 

\begin{example}[Mammography experience data]
	The mammography experience data, which assess factors associated with
	women's knowledge, attitude and behavior towards mammography, was
	introduced in Hosmer and Lemeshow (2000); it is a subset of the original study
	by the University of Massachusetts Medical School and recently studied by Mart\'{\i}n (2015). 
	It involves $N=412$ individuals, $k=8$ explanatory variables and a nominal response with $d+1=3$
	categories (studied in detail in Section \ref{S5}). 
	Here, all individuals do not have distinct covariate values so that their plots (e.g., Figure \ref{fig:motivating_example})
	only distinguish $125$ indices in its x-axis, which corresponds to the grouped observations 
	for 125 distinct covariates values available in the data. 
	Following Mart\'{\i}n (2015), the grouped observations associated with seven such distinct indices 
	can be considered as outliers. A ``good"  robust statistical inference procedure should 
	not get highly affected by the presence  of the outliers. 
	So, we compute the MLE of $\boldsymbol{\beta }$ under the PLRM for the full dataset 
	and also for the outliers deleted dataset and plot the corresponding (estimated) 
	category probabilities for each available distinct covariate values. 
	The left panel of Figure \ref{fig:motivating_example} presents these category probabilities for the second category, 
	which clearly indicates the significant variation of the MLE in the presence or absence of the outliers. 
	In addition, the mean deviations of the estimated probabilities with respect to the relative frequencies, 
	shown in Figure \ref{fig:motivating_example} (right), are always seen to be lower for the outlier deleted estimators.
	This introductory example clearly illustrates the non-robust nature of the MLE 
	and motivates us to look for a robust inference procedure that will	generate correct results with
	high efficiency even without removing the outlying observations. \hfill {$\square $}
	
	\begin{figure}[h]
		\includegraphics[width=8.8cm, height=5.8cm]{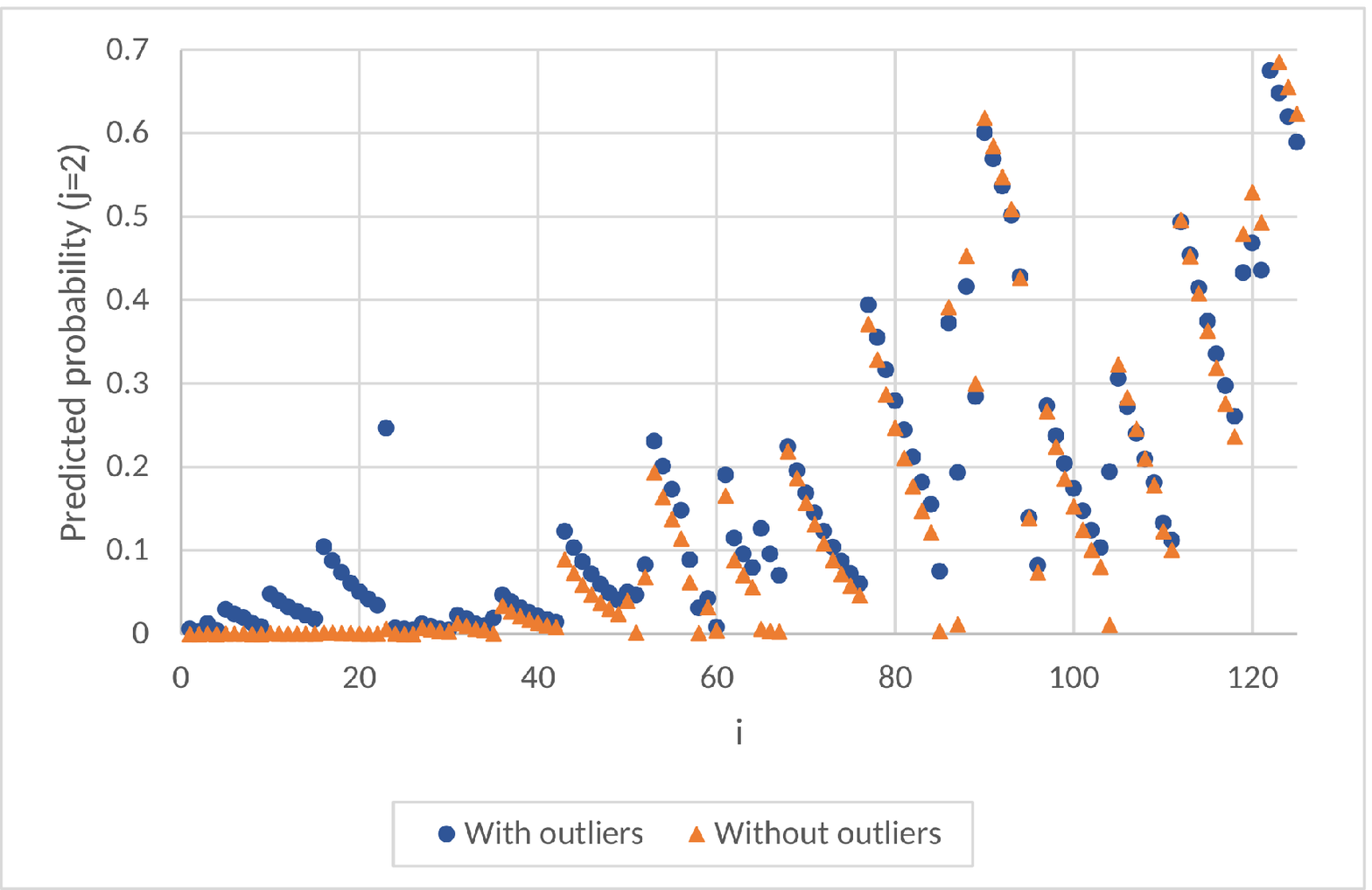} %
		\includegraphics[width=7.2cm, height=5.8cm]{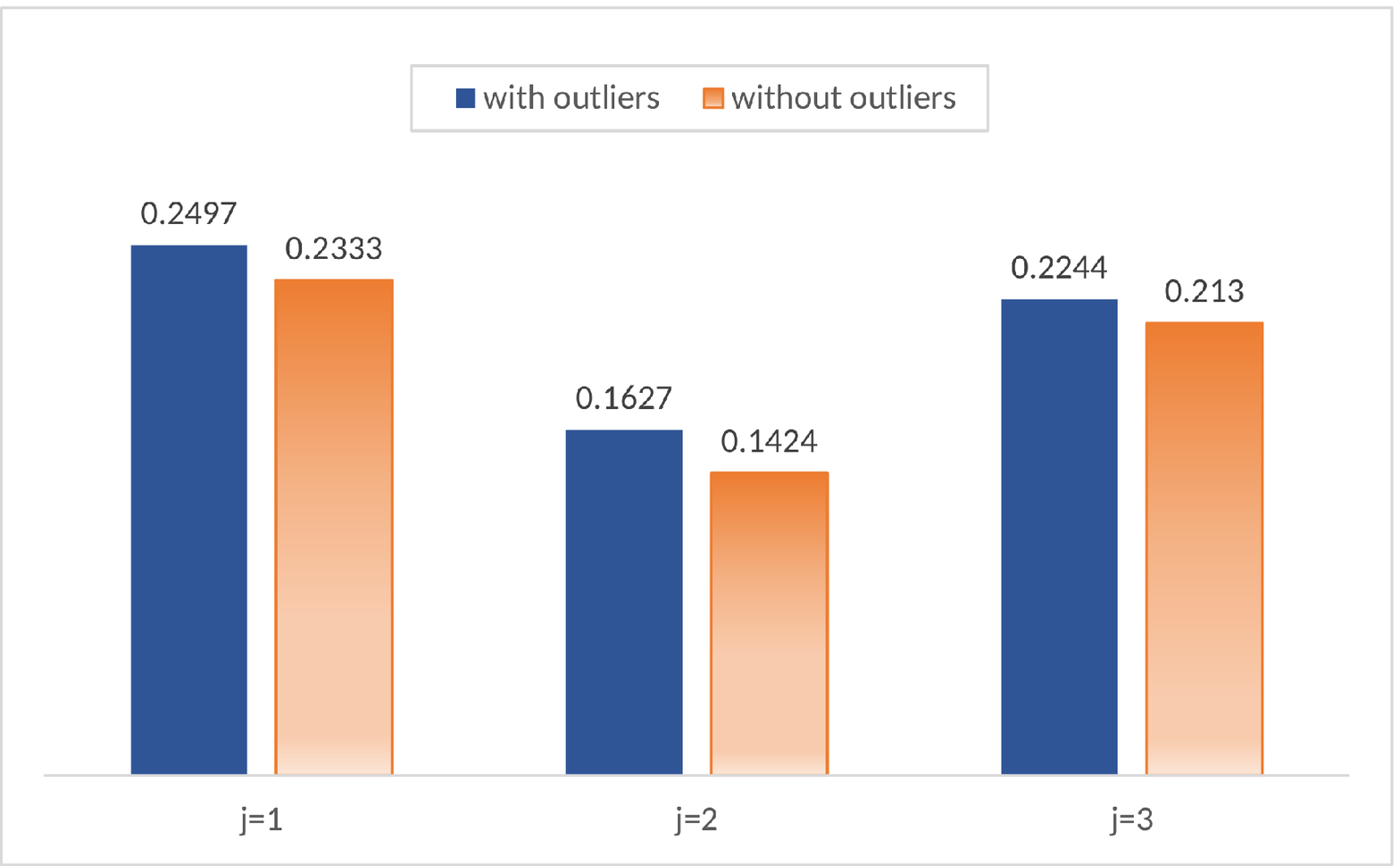}
		\caption{MLE based estimates of (a) the  response probabilities in category $j=2$ (left)
			and (b) the expected mean deviations of the estimated probabilities with respect to
			the corresponding relative frequencies for each category $j=1,2,3$ (right)
			in the Mammography experience data}
		\label{fig:motivating_example}
	\end{figure}
\end{example}

The term \textquotedblleft outlier\textquotedblright needs a clarification when referring Martín (2015). In logistic regression diagnostics, the Cook's distance is a tool for identifying influential observations associated with specific explanatory variables. In the current paper, the term  \textquotedblleft outlier\textquotedblright matches this notion of \textquotedblleft influency \textquotedblright but not directly the term \textquotedblleft outlying for having a large value of the residual\textquotedblright, as in Martín (2015).

There exist some alternative estimation procedures for the PLRM other
than the MLE, although they are used less often in practice; these include the estimators
proposed by Begg and Gray (1984) and its modification by Rom and Cohen (1995). 
The generalized method of moments (GMM) has also been considered (Hayashi, 2000)  
which has been shown to be consistent and fully efficient under suitable conditions.  
Gupta et al. (2006) considered the PLRM with only categorical covariates and discussed the
family of the minimum phi-divergence estimators (M$\phi $Es) that contains
the MLE as a particular case. The M$\phi $Es are BAN estimators and have
high efficiency for moderate and small sample sizes. However, the important
issue of robustness against outliers  was ignored in all these cited references. 
Ronchetti and Trojani (2001) pointed out the  non-robustness of the GMM estimators. 
Recently, Wang (2014) gave a robust modification of the GMM estimator, 
whereas earlier Victoria-Feser and Ronchetti (1997) had presented a robust estimator 
for grouped data with categorical covariates only.

In this paper, we present an alternative simple robust generalization of the
MLE for the general PLRM (\ref{1}) by using the density power divergence (DPD) measure of Basu et al.~(1998). 
The DPD based inference has become very popular in recent time specially because of its strong robustness
properties against outliers; see, e.g.,  Basu et al.\ (2016, 2017a,b), Ghosh et al.\ (2016), among others. 
In this paper, we first develop the estimator of $\boldsymbol{\beta }$ in the PLRM 
by minimizing a suitably defined DPD measure and derive its asymptotic distribution in Section \ref{S2}. 
In Section \ref{S3} the problem of testing linear hypotheses in the PLRM is considered using the newly proposed estimators. 
The robustness of both the proposed estimator and the test of hypothesis are studied theoretically 
through the influence function analysis in Section \ref{S4}. 
Section \ref{S5} presents some real data examples and Section \ref{S6} is devoted to  a simulation study. 
A method for the data-driven selection of the robustness tuning parameter is described in Section \ref{S7}. 
The paper ends with brief concluding remarks in Section \ref{S8}. 
Proofs of the results and further technical details are given in the online supplementary material.

\section{Minimum density power divergence estimator for the PLRM \label{S2}}

The MLE $\widehat{\boldsymbol{\beta }}$ of the parameter $\boldsymbol{%
	\beta }$ under the PLRM in (\ref{1}) can be equivalently defined as the
minimizer of the Kullback-Leibler divergence (KLD) between the probability vectors 
$\widehat{\boldsymbol{p}}=\frac{1}{N}(y_{11},...,y_{1,d+1},y_{21},...,$ $y_{2,d+1},...,y_{N1},...,y_{N,d+1})^{T}$ and $\boldsymbol{p}(\boldsymbol{\beta })=\frac{1}{N}(\boldsymbol{\pi }%
_{1}(\boldsymbol{\beta })^{T},...,\boldsymbol{\pi }_{N}(\boldsymbol{\beta }%
)^{T})^{T}$, where $\boldsymbol{\pi }_{i}\left( \boldsymbol{\beta }\right)
^{T}=(\pi _{i1}(\boldsymbol{\beta }),\ldots ,\pi _{i,d+1}(\boldsymbol{\beta }%
))$, with $\pi _{ij}(\boldsymbol{\beta })=\pi _{j}\left( \boldsymbol{x}_{i},%
\boldsymbol{\beta }\right) $. This follows from the expression of the
KLD between $\widehat{\boldsymbol{p}}$ and $%
\boldsymbol{p}\left( \boldsymbol{\beta }\right) $ given by 
\begin{equation}
D_{KL}\left( \widehat{\boldsymbol{p}},\boldsymbol{p}\left( \boldsymbol{\beta 
}\right) \right) =\sum_{i=1}^{N}\sum_{j=1}^{d+1}\dfrac{y_{ij}}{N}\log \frac{%
	y_{ij}}{\pi _{ij}\left( \boldsymbol{\beta }\right) }=c-\dfrac{1}{N}%
\sum_{i=1}^{N}\sum_{j=1}^{d+1}y_{ij}\log \pi _{ij}\left( \boldsymbol{\beta }%
\right) ,  \label{kull}
\end{equation}%
where $c$ does not depend on $\boldsymbol{\beta }$ and hence $\widehat{%
	\boldsymbol{\beta }}=\arg \min_{\boldsymbol{\beta }\mathbf{\in \Theta }%
}D_{KL}\left( \widehat{\boldsymbol{p}},\boldsymbol{p}\left( \boldsymbol{%
	\beta }\right) \right) $.

The KLD is a particular case of the general DPD measure $D_{\lambda}(\cdot, \cdot)$ 
between $\widehat{\boldsymbol{p}}$ and $\boldsymbol{p}\left(\boldsymbol{\beta }\right) $ 
\begin{equation*}
D_{\lambda }\left( \widehat{\boldsymbol{p}},\boldsymbol{p}\left( \boldsymbol{%
	\beta }\right) \right) =\frac{1}{N^{\lambda +1}}\sum\limits_{i=1}^{N}\left\{
\sum\limits_{j=1}^{d+1}\pi _{ij}^{\lambda +1}\left( \boldsymbol{\beta }%
\right) -\frac{\lambda +1}{\lambda }\sum\limits_{j=1}^{d+1}y_{il}\pi
_{ij}^{\lambda }\left( \boldsymbol{\beta }\right) +\frac{1}{\lambda }\right\}.
\end{equation*}%
Since the term $\frac{1}{\lambda }$ does not have any role in the minimization of 
$D_{\lambda }\left( \widehat{\boldsymbol{p}},\boldsymbol{p}\left( 
\boldsymbol{\beta }\right) \right) $ with respect to $\boldsymbol{\beta }$,
it is sufficient to consider 
$d_{\lambda }\left( \boldsymbol{y}_{i},\boldsymbol{\pi }%
_{i}(\boldsymbol{\beta })\right) =\sum_{j=1}^{d+1}\left\{ \pi _{ij}\left( 
\boldsymbol{\beta }\right) -\tfrac{\lambda +1}{\lambda }y_{ij}\right\} \pi
_{ij}\left( \boldsymbol{\beta }\right) ^{\lambda }$, 
and minimize
\begin{equation}
d_{\lambda }\left( \widehat{\boldsymbol{p}},\boldsymbol{p}\left( \boldsymbol{%
	\beta }\right) \right) =\frac{1}{N^{\lambda +1}}\sum\limits_{i=1}^{N}d_{%
	\lambda }\left( \boldsymbol{y}_{i},\boldsymbol{\pi }_{i}(\boldsymbol{\beta }%
)\right).  
\label{4.6b}
\end{equation}%

\begin{definition}
	The minimum DPD estimator (MDPDE) of $\boldsymbol{\beta }$ with tuning parameter $\lambda $ 
	in the PLRM (\ref{1}) is given by $\widehat{\boldsymbol{\beta }}_{\lambda }=\displaystyle\arg \min_{\boldsymbol{%
			\beta \in }\Theta }d_{\lambda }\left( \widehat{\boldsymbol{p}},\boldsymbol{p}%
	\left( \boldsymbol{\beta }\right) \right) $, where $d_{\lambda }\left( 
	\widehat{\boldsymbol{p}},\boldsymbol{p}\left( \boldsymbol{\beta }\right)
	\right) $ is given in (\ref{4.6b}).
\end{definition}

The DPD at $\lambda =0$ is defined by the limit of the expression in (\ref{4.6b}) as $\lambda \rightarrow 0$, 
which coincides with the KLD (\ref{kull}) up to an additive constant. 
Therefore, the MDPDE at $\lambda =0$ is nothing but the MLE. 
See Appendix for the detailed procedure to obtain the MDPDE.

Note that the random variables $\boldsymbol{Y}_{i}$ associated with the tabulated response 
$\boldsymbol{y}_{i}$,  given the covariate value $\boldsymbol{x}_{i}$ under the PLRM, 
are independent but non-homogeneous. So, we can apply the general theory from Ghosh and Basu (2013) 
to study the properties of the MDPDE under the PLRM. 
In fact, the minimization of the intuitive objective function (\ref{4.6b}) here 
is equivalent to the minimization of the average DPD measure
between the observed data and the model probability mass functions over each distributions (indexed by $i$), 
which is  proposed in Ghosh and Basu (2013) for the general non-homogeneous set-up. 
Hence, based on their general theory, we obtain the asymptotic
properties of our MDPDE under the PLRM which is presented in the following theorem.

\begin{theorem}
	\label{THM:Asymp_MDPDE} Consider the PLRM (\ref{1}) with the true parameter
	value being $\boldsymbol{\beta }_{0}$. Under Assumptions (A1)--(A7) of Ghosh
	and Basu (2013), there exists a consistent MDPDE $\widehat{%
		\boldsymbol{\beta }}_{\lambda }$ of $\boldsymbol{\beta}$ and $\sqrt{N}(\widehat{%
		\boldsymbol{\beta }}_{\lambda }-\boldsymbol{\beta }_{0})\overset{%
		\mathcal{L}%
	}{\underset{N\rightarrow \infty }{\rightarrow }}%
	\mathcal{N}%
	\left( \boldsymbol{0}_{d(k+1)},\boldsymbol{J}_{\lambda }^{-1}\left( 
	\boldsymbol{\beta }_{0}\right) \boldsymbol{V}_{\lambda }\left( \boldsymbol{%
		\beta }_{0}\right) \boldsymbol{J}_{\lambda }^{-1}\left( \boldsymbol{\beta }%
	_{0}\right) \right) $, where $\boldsymbol{J}_{\lambda }\left( \boldsymbol{%
		\beta }\right) =\lim\limits_{N\rightarrow \infty }\boldsymbol{\Psi }%
	_{N,\lambda }\left( \boldsymbol{\beta }\right) $ and $\boldsymbol{V}%
	_{\lambda }\left( \boldsymbol{\beta }\right) =\lim\limits_{N\rightarrow
		\infty }\boldsymbol{\Omega }_{N,\lambda }\left( \boldsymbol{\beta }\right) $
	with%
	\begin{equation}
	\boldsymbol{\Psi }_{N,\lambda }\left( \boldsymbol{\beta }\right) =\frac{1}{N}%
	\sum\limits_{i=1}^{N}\boldsymbol{\Delta }(\boldsymbol{\pi }_i^{\ast }\left( 
	\boldsymbol{\beta }\right) )\mathrm{diag}^{\lambda -1}\{\boldsymbol{\pi }_i^{\ast }\left( \boldsymbol{\beta }\right) \}
	\boldsymbol{\Delta }(\boldsymbol{\pi }_i^{\ast }\left( \boldsymbol{\beta }\right) )\otimes \boldsymbol{x}_{i}\boldsymbol{x}%
	_{i}^{T},  \label{EQ:Psi}
	\end{equation}%
	\begin{equation}
	\boldsymbol{\Omega }_{N,\lambda }\left( \boldsymbol{\beta }\right) =\frac{1%
	}{N}\sum\limits_{i=1}^{N}\boldsymbol{\Delta }(\boldsymbol{\pi }_i^{\ast
	}\left( \boldsymbol{\beta }\right) )\mathrm{diag}%
	^{\lambda -1}\{\boldsymbol{\pi }_i^{\ast }\left(
	\boldsymbol{\beta }\right) \}\boldsymbol{\Delta }(\boldsymbol{\pi }_i^{\ast
	}\left( \boldsymbol{\beta }\right) )
	\mathrm{diag}^{\lambda -1}\{\boldsymbol{\pi }_i^{\ast }\left( 
	\boldsymbol{\beta }\right) \}\boldsymbol{\Delta }(%
	\boldsymbol{\pi }_i^{\ast }\left( \boldsymbol{\beta }%
	\right) )\otimes \boldsymbol{x}_{i}\boldsymbol{x}_{i}^{T}.
	\label{EQ:Omega}
	\end{equation}%
	Here, $\boldsymbol{\Delta }(\boldsymbol{p})=\mathrm{%
		diag}(\boldsymbol{p})-\boldsymbol{pp}^{T}$, the vector with superscript $\ast $ denotes its subvector with the last element removed 
	and the vector with any other superscript takes the corresponding power for all its components.
\end{theorem}

\section{Wald-type test statistics for testing linear hypotheses \label{S3}}

Most  testing problems for $\boldsymbol{\beta }$ in the PLRM belong to the class of 
linear hypotheses given by
\begin{equation}
H_{0}:\boldsymbol{L}^{T}\boldsymbol{\beta }=\boldsymbol{h}\text{ against }%
H_{1}:\boldsymbol{L}^{T}\boldsymbol{\beta }\neq \boldsymbol{h},  \label{3.1}
\end{equation}%
where $\boldsymbol{L}$ is $d(k+1)\times r$ full rank matrix with $r\leq
d(k+1)$ and $\boldsymbol{h}$ an $r$-vector. Two important particular cases
are the testing problems for $H_{0}:\boldsymbol{\beta }=\boldsymbol{\beta }_{0}$ or $%
H_{0}:\boldsymbol{\beta }_{(s)}=\boldsymbol{0}$ against their respective
omnibus alternatives, where $\boldsymbol{\beta }_{(s)}$ is a subvector of $%
\boldsymbol{\beta }$. 

\begin{definition}
	Let $\widehat{\boldsymbol{\beta }}_{\lambda }$ be the MDPDE of $\boldsymbol{\beta }$ in the PLRM (\ref{1})
	and denote $\boldsymbol{M}_{N,\lambda }(\widehat{%
		\boldsymbol{\beta }}_{\lambda })=$ $\boldsymbol{\Psi }_{N,\lambda }^{-1}(%
	\widehat{\boldsymbol{\beta }}_{\lambda })\boldsymbol{\Omega }_{N,\lambda }(%
	\widehat{\boldsymbol{\beta }}_{\lambda })$ $\boldsymbol{\Psi }_{N,\lambda
	}^{-1}(\widehat{\boldsymbol{\beta }}_{\lambda })$.  
	Then, the family of Wald-type test statistics for testing the null hypothesis given in (\ref{3.1}) is defined as 
	\begin{equation}
	W_{N}(\widehat{\boldsymbol{\beta }}_{\lambda })=N(\boldsymbol{L}^{T}\widehat{%
		\boldsymbol{\beta }}_{\lambda }-\boldsymbol{h})^{T}\left\{\boldsymbol{L}^{T}%
	\boldsymbol{M}_{N,\lambda }(\widehat{\boldsymbol{\beta }}_{\lambda })%
	\boldsymbol{L}\right\}^{-1}(\boldsymbol{L}^{T}\widehat{\boldsymbol{\beta }}%
	_{\lambda }-\boldsymbol{h}).  \label{3.2}
	\end{equation}
\end{definition}

In particular,  since $\widehat{\boldsymbol{\beta }}_{\lambda =0}=\widehat{\boldsymbol{\beta }}$, 
the MLE of $\boldsymbol{\beta }$, and 
$\boldsymbol{M}_{N,\lambda =0}(\widehat{\beta })=\boldsymbol{\Psi }_{N,\lambda=0 }^{-1}(\widehat{\boldsymbol{\beta }})$ 
with 
\begin{equation*}
\boldsymbol{\Psi }_{N,\lambda=0 }\left( \boldsymbol{\beta }\right) =%
\boldsymbol{\Omega }_{N,\lambda=0 }\left( \boldsymbol{\beta }\right) =\frac{1}{%
	N}\sum\limits_{i=1}^{N}\boldsymbol{\Delta }(\boldsymbol{\pi }_i^{\ast }\left( 
\boldsymbol{\beta }\right) )\otimes \boldsymbol{x}_{i}%
\boldsymbol{x}_{i}^{T},
\end{equation*}%
being the Fisher information matrix, $W_{N}(\widehat{\boldsymbol{\beta }}_{\lambda=0 })$ 
becomes the classical Wald test statistic.

\begin{theorem}
	\label{THM:null_wald} The asymptotic distribution of the Wald-type test
	statistics, $W_{N}(\widehat{\boldsymbol{\beta }}_{\lambda })$, under the
	null hypothesis in (\ref{3.1}), is a chi-square distribution ($\chi_r^2$) 
	with $r$ degrees of freedom.
\end{theorem}

Based on Theorem \ref{THM:null_wald}, the null hypothesis in (\ref{3.1})
will be rejected if $W_{N}(\widehat{\boldsymbol{\beta }}_{\lambda })>\chi_{r,\alpha }^{2}$,
the upper $\alpha$-th quantile of $\chi_r^2$. 
See Appendix for some results about its  power function.

\section{Robustness Properties\label{S4}}


We first study the robustness of the proposed MDPDE of $\boldsymbol{\beta }$\ under the PLRM
(\ref{1}) through the influence function analysis. 
For any estimator defined in terms of a statistical functional $\boldsymbol{U}(G)$ 
in the set-up of independent and identically distributed (IID) data from the true distribution function $G$, 
its influence function is defined as 
$\mathcal{IF}\left( \boldsymbol{t},\boldsymbol{U},G\right) =\lim\limits_{\epsilon
	\downarrow 0}\frac{\boldsymbol{U}(G_{\epsilon })-\boldsymbol{U}(G)}{\epsilon }
=\left. \frac{\partial \boldsymbol{U}(G_{\epsilon })}{\partial \epsilon }\right\vert _{\epsilon =0}$, 
where $G_{\epsilon }=(1-\epsilon )G+\epsilon \wedge _{\boldsymbol{t}}$ 
with $\epsilon $ being the contamination proportion 
and $\wedge _{\boldsymbol{t}}$ being the degenerate distribution at the contamination point $\boldsymbol{t}$. 
Thus, the (first order) influence function (IF), as a function of $\boldsymbol{t}$, measures the
standardized asymptotic bias (in its first order approximation) caused by the
infinitesimal contamination at the point $\boldsymbol{t}$. The maximum of
this IF over $\boldsymbol{t}$ indicates the extent of bias due to
contamination and hence lower its value, the  more robust the estimator is.

For the present case of PLRM, given the value of $\boldsymbol{x}_{i}$, $\boldsymbol{y}%
_{i} $s are independent but not-identically distributed. The extended
definition of IF for such non-homogeneous cases has been proposed in Huber
(1983) and Ghosh and Basu (2013, 2016). Following the approach of the later, let $G_{i}$
denote the true distribution function of $\boldsymbol{y}_{i}$ having joint mass function $%
g_{i}$ and $F_{i}\left( \boldsymbol{\beta }\right)$ denote the distribution
function under the assumption of PLRM having joint mass function 
$f_{i}(\boldsymbol{y}_{i},\boldsymbol{\beta })=\boldsymbol{y}_{i}^{T}\boldsymbol{\pi }_{i}\left(\boldsymbol{\beta }\right) 
=\sum_{j=1}^{d+1}y_{ij}\pi _{ij}({\boldsymbol{\beta }})$.
Denote $\boldsymbol{G}=(G_{1},\cdots ,G_{N})^{T}$, 
$\boldsymbol{F}(\boldsymbol{\beta }) =\left(F_{1}(\boldsymbol{\beta }) ,\cdots ,F_{N}(\boldsymbol{\beta })\right)^{T}$ 
and  $\mathcal{Y}=\{\boldsymbol{e}_{j,d+1}\}_{j=1}^{d+1}$, 
where $\boldsymbol{e}_{j,d+1}$ is the $j$-th column of the identity matrix of order $d+1$. 
Note that, $\mathcal{Y}$ is the sample space of the response variable $\boldsymbol{Y}_{i}$. 
Then, the statistical
functional $\boldsymbol{U}_{\lambda }(\boldsymbol{G})$ corresponding to the
proposed MDPDE, $\widehat{\boldsymbol{\beta }}_{\lambda }$,  of $\boldsymbol{\beta }$
is  defined as the minimizer of 
\begin{equation*}
H(\boldsymbol{\beta })=\sum\limits_{i=1}^{N}\left\{ \sum_{\boldsymbol{y}_{i}%
	\boldsymbol{\in }%
	\mathcal{Y}%
}f_{i}^{\lambda +1}(\boldsymbol{y}_{i},\boldsymbol{\beta })-\frac{\lambda +1%
}{\lambda }\sum_{\boldsymbol{y}_{i}\boldsymbol{\in }%
	\mathcal{Y}%
}f_{i}^{\lambda }(\boldsymbol{y}_{i},\boldsymbol{\beta })g_{i}(\boldsymbol{y}%
_{i})\right\} ,
\end{equation*}%
whenever it exists. 
When the assumption of PLRM holds with true parameter $\boldsymbol{\beta }_{0}$, 
we have $g_{i}(\boldsymbol{y}_{i})=f_{i}(\boldsymbol{y}_{i},\boldsymbol{\beta }_{0})$, 
and thus $H(\boldsymbol{\beta })=N^{\lambda +1}d_{\lambda }\left(\boldsymbol{p}\left(\boldsymbol{\beta }_{0}\right),
\boldsymbol{p}\left(\boldsymbol{\beta }\right) \right)$; 
this is minimized at $\boldsymbol{\beta}=\boldsymbol{\beta}_0$ 
implying the Fisher consistency of the MDPDE functional $\boldsymbol{U}_{\lambda}(\boldsymbol{G})$ at the PLRM.
Note that, in such non-homogeneous settings, outliers can be either in any one or more index $i\in
\{1,...,N\}$ for $g_{i}$. Following the general results
from Ghosh and Basu (2013), one can derive the IF of the MDPDE at $%
\boldsymbol{F}({\boldsymbol{\beta }}_{0})$ as given by%
\begin{equation*}
\mathcal{IF}_{i_{0}}(\boldsymbol{t}_{i_{0}};\boldsymbol{U}_{\lambda },\boldsymbol{F}_{{\boldsymbol{\beta }}_{0}})
=\boldsymbol{\Psi }_{N,\lambda}^{-1}({\boldsymbol{\beta }}_{0})\frac{1}{N}
\left( \boldsymbol{I}_{d},\boldsymbol{0}_{d}\right) 
\boldsymbol{\Delta }(\boldsymbol{\pi }_{i_{0}}\left( \boldsymbol{\beta }\right)\mathrm{diag}^{\lambda -1}\{\boldsymbol{\pi }_{i_{0}}^{\ast }\left( {\boldsymbol{\beta}}_{0}\right) \}
\left\{\boldsymbol{t}_{i_{0}}^{\ast }-\boldsymbol{\pi }^{\ast}\left({\boldsymbol{\beta }}_{0}\right) \right\}
\otimes \boldsymbol{x}_{i_{0}},
\end{equation*}%
when there is contamination in only one specific index $i_0$  at the point $\boldsymbol{t}_{i_{0}}\in \mathcal{Y}$ 
and 
\begin{equation*}
\mathcal{IF}%
(\underline{\boldsymbol{t}};\boldsymbol{U}_{\lambda },\boldsymbol{F}({\boldsymbol{\beta }}_{0}))
=\boldsymbol{\Psi }_{N,\tau }^{-1}(\boldsymbol{\beta }_{0})\frac{1}{N}\sum_{i=1}^{N}
\left( \boldsymbol{I}_{d},\boldsymbol{0}_{d}\right) 
\boldsymbol{\Delta }(\boldsymbol{\pi }_{i}\left( \boldsymbol{\beta }\right)\mathrm{diag}^{\lambda -1}\{\boldsymbol{\pi }_i^{\ast }\left( {\boldsymbol{\beta}}_{0}\right) \}
\left\{\boldsymbol{t}_{i}^{\ast }-\boldsymbol{\pi }^{\ast}\left({\boldsymbol{\beta }}_{0}\right) \right\} 
\otimes \boldsymbol{x}_{i},
\end{equation*}%
when the contamination is in all $N$ distributions $G_{i}$ at the points $\boldsymbol{t}_{i}\in\mathcal{Y}$, respectively,
for $i=1,\ldots ,N$. Here $\underline{\boldsymbol{t}} = (\boldsymbol{t}_1, \ldots, \boldsymbol{t}_N)$.

Note that, these IFs are bounded in large $\boldsymbol{x}_{i}$s (leverage
points) for all $\lambda >0$, but unbounded at $\lambda =0$ (the MLE). This
implies that the proposed MDPDEs with $\lambda >0$ are robust against leverage points,
but the classical MLE is clearly non-robust. However, we cannot
directly infer about the robustness against outliers in the response
variable which are, in fact, the misspecification errors. This is because, in
such cases, $\boldsymbol{t}_{i}$ changes its indicative category only (does
not go to infinity) and the IFs are bounded in $\boldsymbol{t}_{i}$s for all 
$\lambda \geq 0$. But, it is well studied that the MLE at $\lambda =0$ is
highly non-robust against these misspecification errors. 
In the next two sections, we will empirically illustrate 
the strong robustness of our proposed MDPDEs with $\lambda >0$
also against such misspecification errors.

Next, we study the robustness of the proposed Wald-type test statistics.
The IF of a testing procedure, as introduced by Rousseeuw and Ronchetti (1979) for IID data, 
is also defined as in the case of estimation but with the statistical functional corresponding to the test statistics
and  it is studied under the null hypothesis. 
This concept has been extended to the non-homogeneous data,
which is the case here, by Ghosh and Basu (2017), Aerts and Haesbroeck
(2017) and Basu et al.~(2017b); the last one considered the general Wald-type test
statistics. The associated statistical functional for our Wald-type test statistics (\ref{3.2}) can be defined as (ignoring the multiplier $N$) 
\begin{equation}
T_{\lambda }\left( \boldsymbol{G}\right) 
=\left\{ \boldsymbol{L}^{T}\boldsymbol{U}_{\lambda }(\boldsymbol{G})-\boldsymbol{h}\right\}^{T}
\left\{\boldsymbol{L}^{T}\boldsymbol{M}_{N, \lambda }(\boldsymbol{U}_{\lambda }(\boldsymbol{\boldsymbol{G}}))\boldsymbol{L}\right\}^{-1}
\left\{ \boldsymbol{L}^{T}\boldsymbol{U}_{\lambda }(\boldsymbol{G})-\boldsymbol{h}\right\},  \label{EQ:WT_func}
\end{equation}%
where $\boldsymbol{U}_{\lambda }(\boldsymbol{G})$ is the  MDPDE functional. 
Again we can have contamination in either one or all distributions as before 
and the corresponding IFs can be obtained from the general theory of Basu et al.\ (2017b). 
Letting $\boldsymbol{\beta }_{0}$ be the true null parameter value under (\ref{3.1}), 
the (first order) IFs in either case of contamination become identically zero 
at $\boldsymbol{G}=\boldsymbol{F}({\boldsymbol{\beta }}_{0})$.

Thus the first order bias approximation is not much informative for the Wald-type test statistics  
and we need to study their second order bias approximation, 
quantified through the second order IFs which we denote by $IF^{(2)}$. 
It is defined as the second order derivative of the functional value at $G_{\varepsilon}$ with respect to the
contamination proportion $\varepsilon\rightarrow0^+$. 
Following Basu et al.~(2017b), we can show that,
at the null distribution $\boldsymbol{G}=\boldsymbol{F}({\boldsymbol{\beta }}_{0})$, 
\begin{align*} 
\mathcal{IF}^{(2)}(\boldsymbol{t}_{i_{0}};T_{\lambda },\boldsymbol{F}({\boldsymbol{\beta }}_{0})) 
& =2%
\mathcal{IF}(\boldsymbol{t}_{i_{0}};\boldsymbol{U}_{\lambda },\boldsymbol{F}({\boldsymbol{\beta }}_{0}))^{T}
\boldsymbol{L}^{T}\left\{ \boldsymbol{L\boldsymbol{M}}_{N,\lambda }(\boldsymbol{\boldsymbol{\beta }}_{0})
\boldsymbol{L}^{T}\right\}^{-1}\boldsymbol{L}
\mathcal{IF}(\boldsymbol{t}_{i_{0}};\boldsymbol{U}_{\lambda },\boldsymbol{F}({\boldsymbol{\beta }}_{0})), \\
\mathcal{IF}^{(2)}(\underline{\boldsymbol{t}};T_{\lambda },\boldsymbol{F}({\boldsymbol{\beta }}_{0})) 
& =2%
\mathcal{IF}(\underline{\boldsymbol{t}};\boldsymbol{U}_{\lambda },\boldsymbol{F}({\boldsymbol{\beta }}_{0}))^{T}
\boldsymbol{L}^{T}\left\{\boldsymbol{L\boldsymbol{\boldsymbol{M}}}_{N, \lambda }(\boldsymbol{\beta}_{0})
\boldsymbol{L}^{T}\right\}^{-1}\boldsymbol{L}
\mathcal{IF}(\underline{\boldsymbol{t}};\boldsymbol{U}_{\lambda },\boldsymbol{F}({\boldsymbol{\beta }}_{0})),
\end{align*}%
for contamination only in $i_{0}$-th distribution or in all distributions, respectively. 
Clearly, the boundedness of these IFs directly depend on that of the MDPDE. 
Hence, the proposed Wald-type test statistics with $\lambda >0$ are expected to be robust,
whereas the classical Wald test statistic at $\lambda =0$ is non-robust against infinitesimal contamination.

\section{Numerical Examples\label{S5}}

\subsection{Mammography experience data (Hosmer and Lemeshow)}

Let us reconsider our motivating example, the mammography experience data, 
to study the performance of our proposed robust procedures. 
As noted in Section \ref{S1}, 
this dataset involves $k=8$ explanatory variables with all of them being dummy variables, except one. 
Among them, three dummy variables are associated with four
categories of the variable SYMPT (`You do not need a mammogram unless you
develop symptoms': 1, strongly agree; 2, agree; 3, disagree; 4, strongly
disagree), the fourth dummy variable represents two categories of variable
HIST (`Mother or sister with a history of breast cancer': 1, no; 2, yes),
the fifth dummy variable corresponds to two categories of variable BSE (`Has
anyone taught you how to examine your own breasts?': 1, no; 2, yes) and other
two dummy variables are associated with three categories of the variable DETC
(`How likely is it that a mammogram could find a new case of breast
cancer?': 1, not likely; 2, somewhat likely; 3, very likely). The final explanatory variable is a quantitative variable representing the PB score (\textquotedblleft Perceived
benefit of mammography\textquotedblright : values between $5$ and $20$, with
the lowest value representing the highest benefit perception). The response
variable ME (Mammography experience) is a categorical factor with three
levels: \textquotedblleft Never\textquotedblright , \textquotedblleft Within
a Year\textquotedblright\ and \textquotedblleft Over a
Year\textquotedblright\ ($d=2$). Hence, $\widehat{\boldsymbol{\beta }}%
_{\lambda }^{T}=(\widehat{\boldsymbol{\beta }}_{1,\lambda }^{T},\widehat{%
	\boldsymbol{\beta }}_{2,\lambda }^{T})$, with 
$ 
\widehat{\boldsymbol{\beta }}_{j,\lambda }^{T}=(\widehat{\beta }_{0j,\lambda
},\widehat{\beta }_{1j,\lambda },\dots,%
\widehat{\beta }_{8j,\lambda }),\quad j=1,2,
$
and $\widehat{\beta }_{1j,\lambda }=\widehat{\beta }_{SYMPT1,j,\lambda }$, $%
\widehat{\beta }_{2j,\lambda }=\widehat{\beta }_{SYMPT2,j,\lambda }$, $%
\widehat{\beta }_{3j,\lambda }=\widehat{\beta }_{SYMPT3,j,\lambda }$, $%
\widehat{\beta }_{4j,\lambda }=\widehat{\beta }_{HIST,j,\lambda }$, $%
\widehat{\beta }_{5j,\lambda }=\widehat{\beta }_{BSE,j,\lambda }$, $\widehat{%
	\beta }_{6j,\lambda }=\widehat{\beta }_{DETC1,j,\lambda }$, $\widehat{\beta }%
_{7j,\lambda }=\widehat{\beta }_{DETC2,j,\lambda }$, $\widehat{\beta }%
_{8j,\lambda }=\widehat{\beta }_{PB,j,\lambda }$.

As suggested by Mart\'{\i}n (2015), the groups of observations associated
with covariate values $\boldsymbol{x}_{i}\ $for $i$ equal to $1$, $3$%
, $17$, $35$, $75$, $81$ and $102$ can be treated as outliers; the MDPDEs of 
$\boldsymbol{\beta }$ obtained with and without these outliers are presented in Table 1--2 in Appendix \ref{tables}. 
One important difference is observed in $\widehat{\beta }_{52,\lambda }=\widehat{\beta }_{BSE,2,\lambda }$; 
when outliers are deleted $\widehat{\beta }_{BSE,2,\lambda }>0$ 
whereas $\widehat{\beta }_{BSE,2,\lambda }<0$ for the full data. 
Although it is observed for all values of  $\lambda $, 
MDPDEs with moderate $\lambda>0 $ are quite near to zero in both the cases, 
and hence they may not be significantly different. 
A deeper study of the MDPDEs in this example is presented below.

\textit{Study of the Efficiency:} 
For each category of the response variable, we have calculated the estimated mean deviation
(EMD) of the predicted probabilities with respect to the relative
frequencies of the response variable under presence or absence of the outliers.
These are defined as $\widehat{md}_{j}(\widehat{\boldsymbol{\beta }}_{\lambda }^{\star })
=\dfrac{1}{N}\sum_{i=1}^{N}\left\vert \pi _{ij}(\widehat{\boldsymbol{\beta }}_{\lambda }^{\star })
-\widehat{p}_{ij}\right\vert $, $j=1,2,3$, with $\star \in \{pres,abs\}$ and are reported  in Table 3 in Appendix \ref{tables}. 
We can see that the MLE yields the minimum average EMD, defined as 
$\overline{\widehat{md}}(\widehat{\boldsymbol{\beta }}_{\lambda}^{\star })
=\frac{1}{d+1}\sum_{j=1}^{d+1}\widehat{md}_{j}(\widehat{\boldsymbol{\beta }}_{\lambda }^{\star })$, 
and hence leads to the highest efficiency in the absence of outliers. 
This efficiency decreases as $\lambda$ increases, but the loss is not very significant. 
The simulation study, presented in the next section, indicates the same behavior for pure data.

\textit{Study of Robustness:} 
From Table 3 in Appendix \ref{tables}, the average  EMDs in the presence of outliers decrease significantly as $\lambda $ increases from $0$. 
This illustrates the increasing robustness of our proposed MDPDEs with increasing $\lambda >0$.
In order to further examine the robustness of the MDPDEs, 
we compute the average mean deviations between the predicted probabilities 
obtained in presence and absence of the outliers, as given by
$\overline{\widehat{md}}(\widehat{\boldsymbol{\beta }}_{\lambda }^{pres},\widehat{\boldsymbol{\beta }}_{\lambda }^{abs})
=\frac{1}{d+1}\sum_{j=1}^{d+1}
\widehat{md}_{j}(\widehat{\boldsymbol{\beta }}_{\lambda }^{pres},\widehat{\boldsymbol{\beta }}_{\lambda }^{abs})$ 
with $\widehat{md}_{j}(\widehat{\boldsymbol{\beta }}_{\lambda }^{pres},\widehat{\boldsymbol{\beta }}_{\lambda }^{abs})
=\dfrac{1}{N}\sum_{i=1}^{N}\left\vert \pi _{ij}(\widehat{\boldsymbol{\beta }}_{\lambda }^{pres})
-\pi_{ij}(\widehat{\boldsymbol{\beta }}_{\lambda}^{abs})\right\vert$, $j=1,2,3$. 
Their values, as presented in Table 4 of Appendix \ref{tables}, clearly show that the MDPDE becomes more robust as $\lambda $ increases. 
This behavior is also illustrated in Figure \ref{fig:mamo_points}, 
where predicted category probabilities at each observed covariate values are shown 
for $\lambda =0$ (MLE, left) and $\lambda =0.7$ (right). 
The differences between blue and orange points, corresponding to the outlier deleted data and the full data,
respectively, are quite significant for the classical MLE but much more stable for the MDPDE at $\lambda =0.7$.

\begin{figure}[th]
	\begin{tabular}{ll}
		\includegraphics[width=8cm, height=5cm]{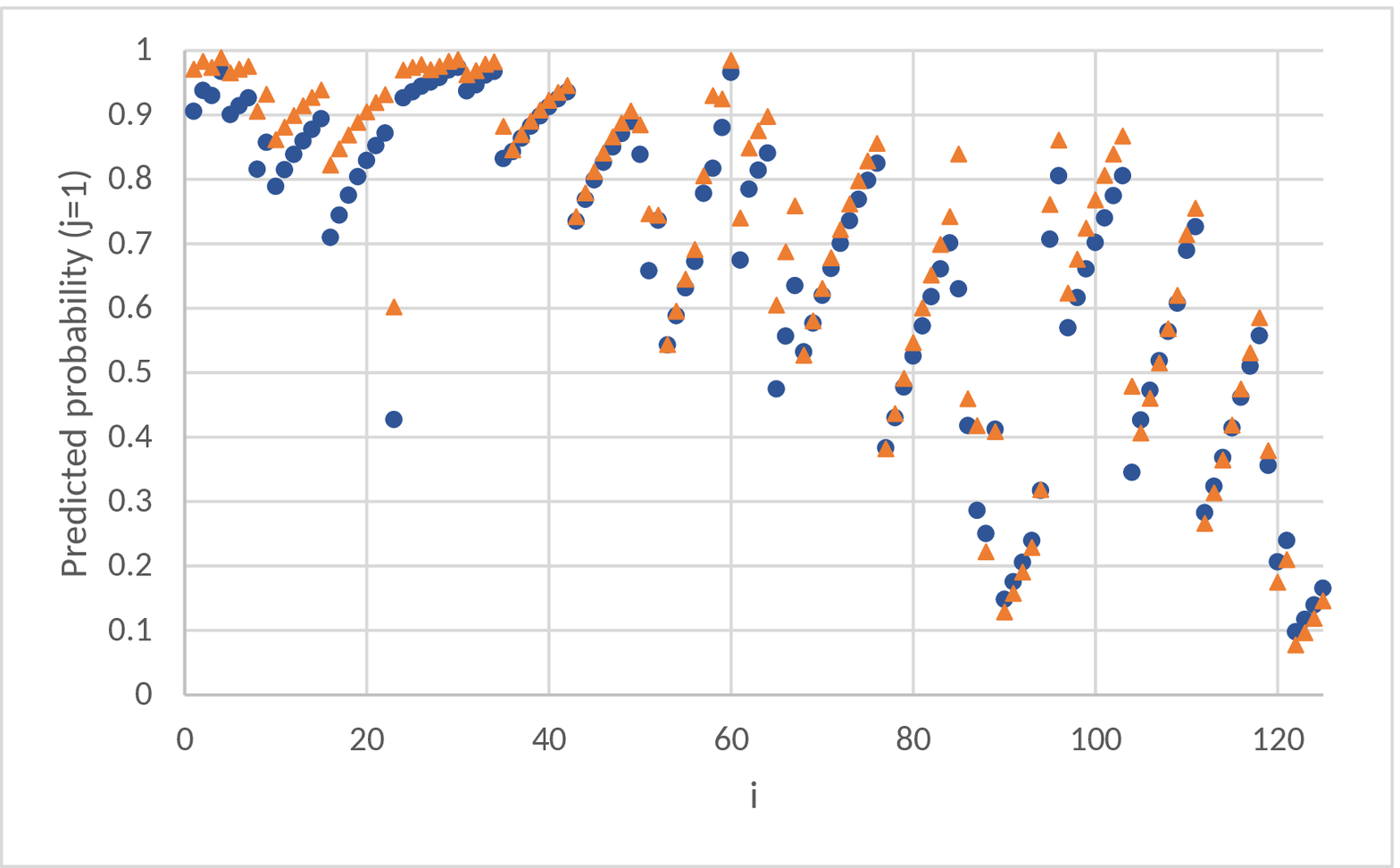}
		& \includegraphics[width=8cm, height=5cm]{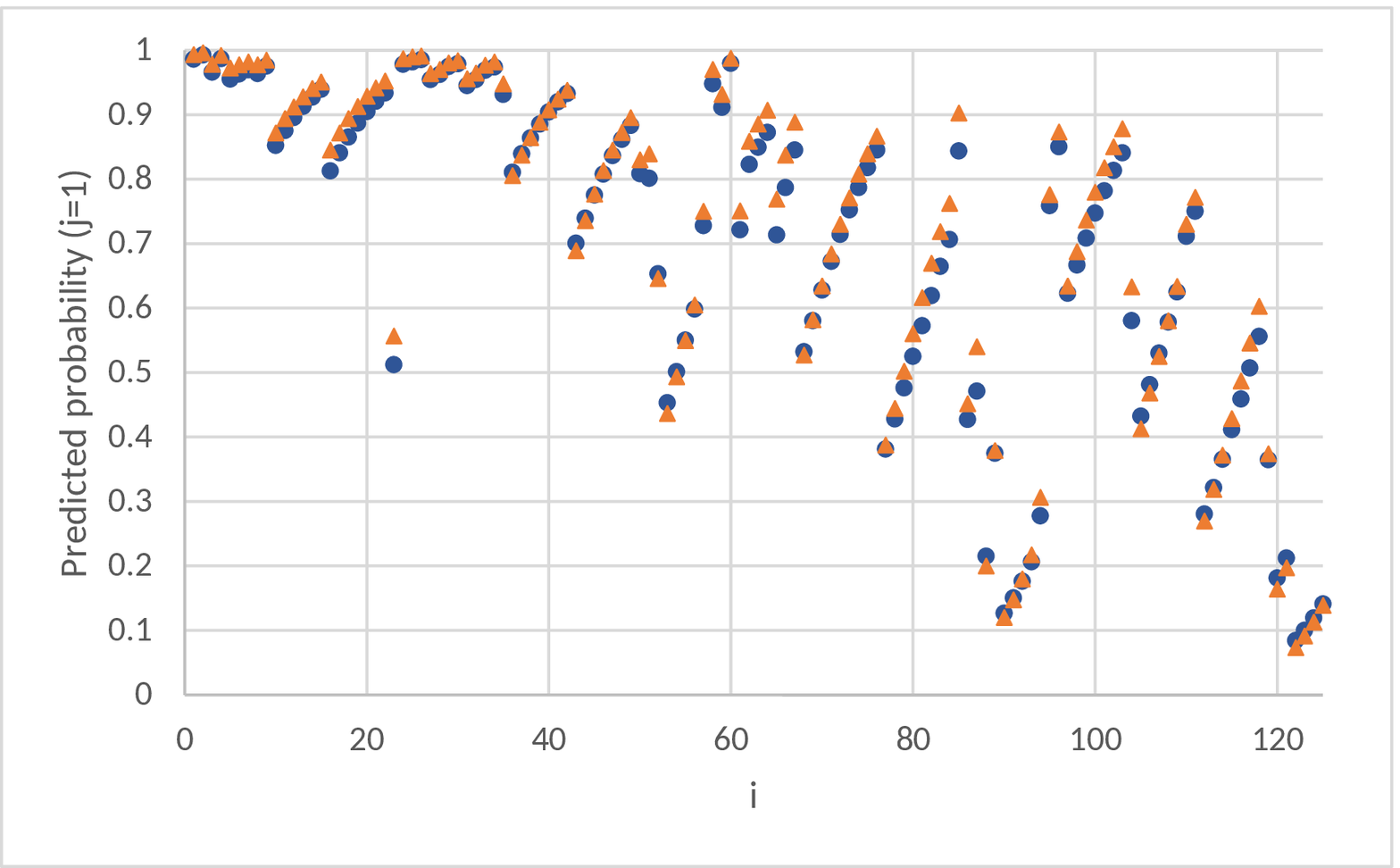}
		\\ 
		\includegraphics[width=8cm, height=5cm]{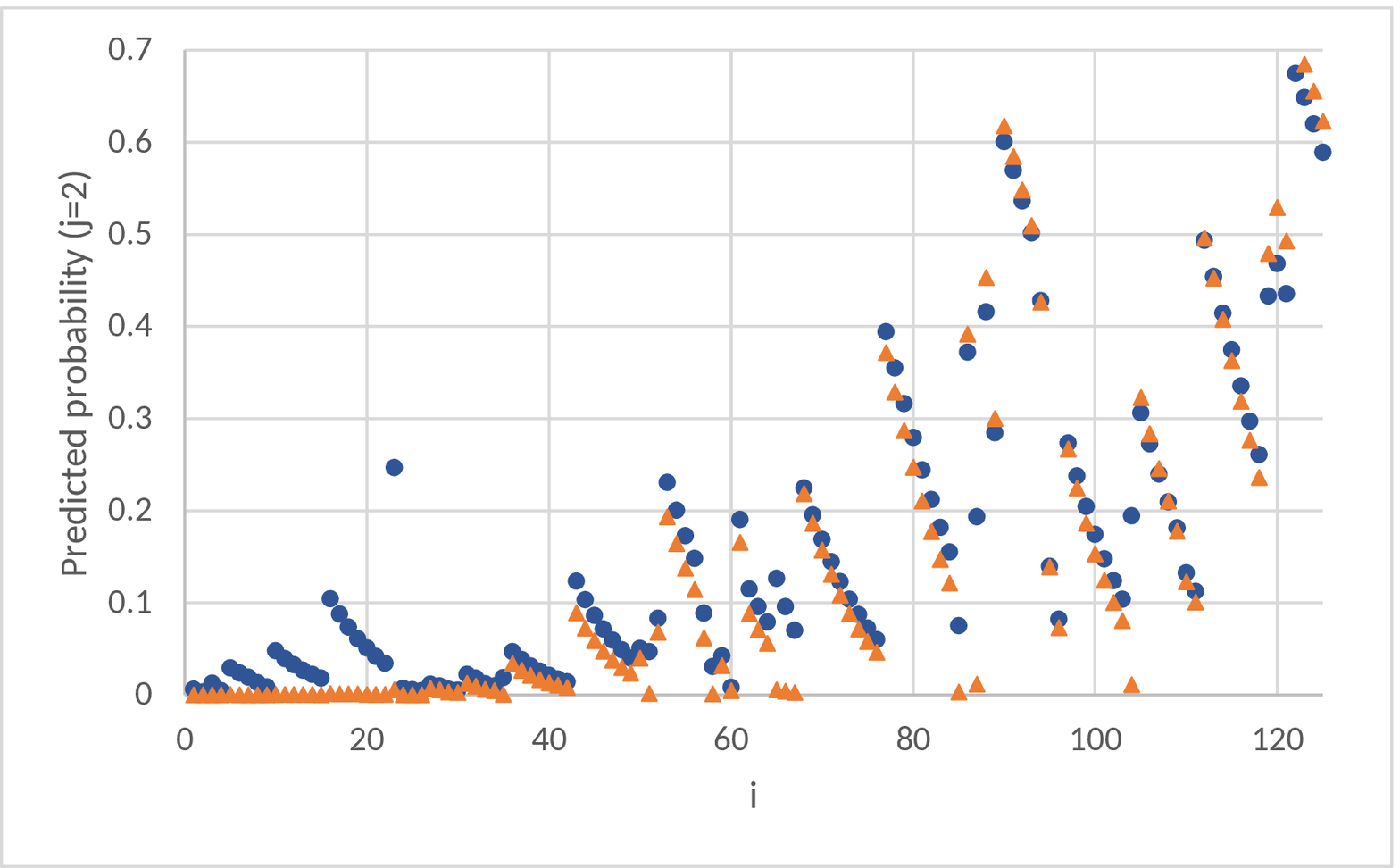}
		& \includegraphics[width=8cm, height=5cm]{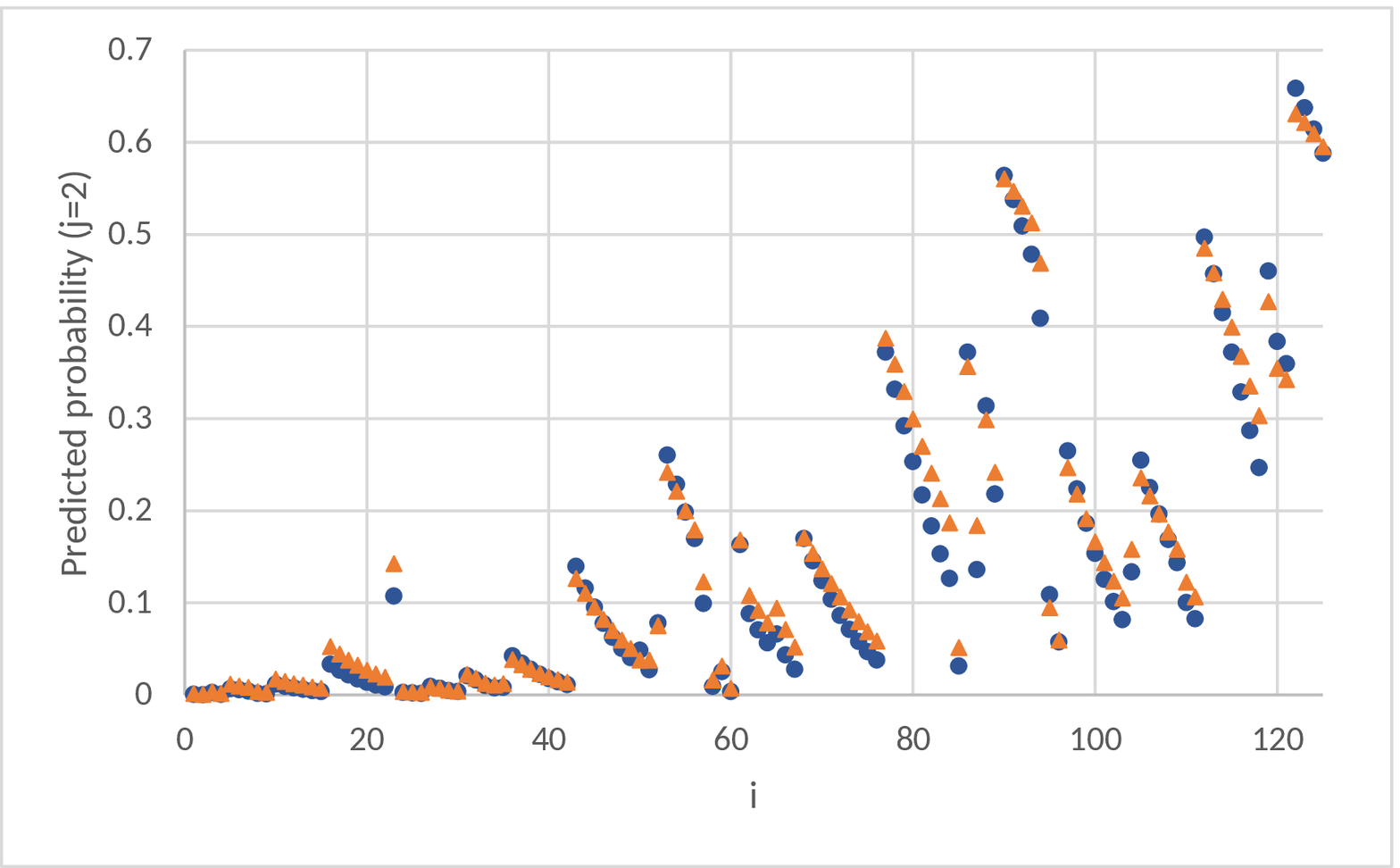}
		\\ 
		\includegraphics[width=8cm, height=5cm]{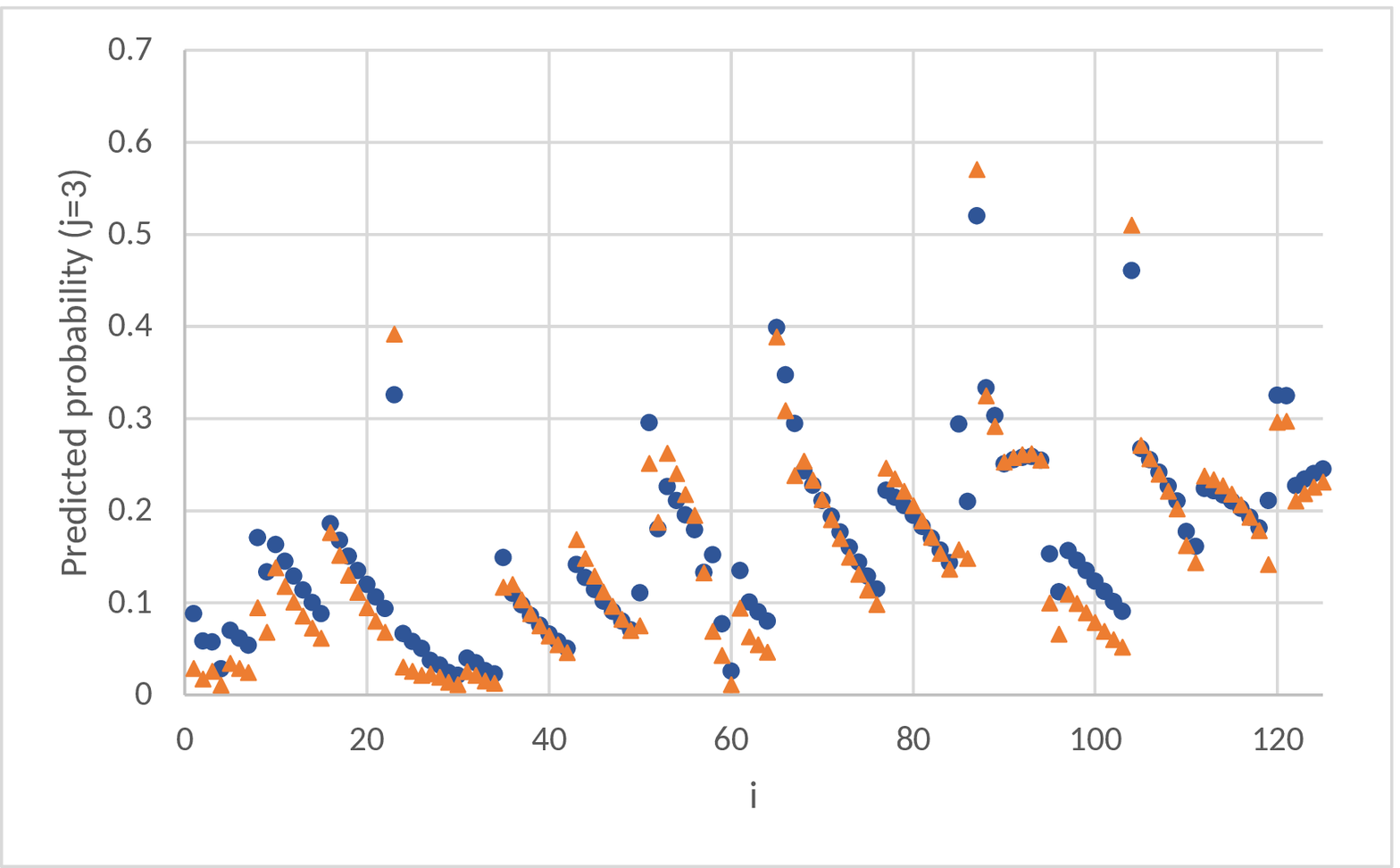}
		& \includegraphics[width=8cm, height=5cm]{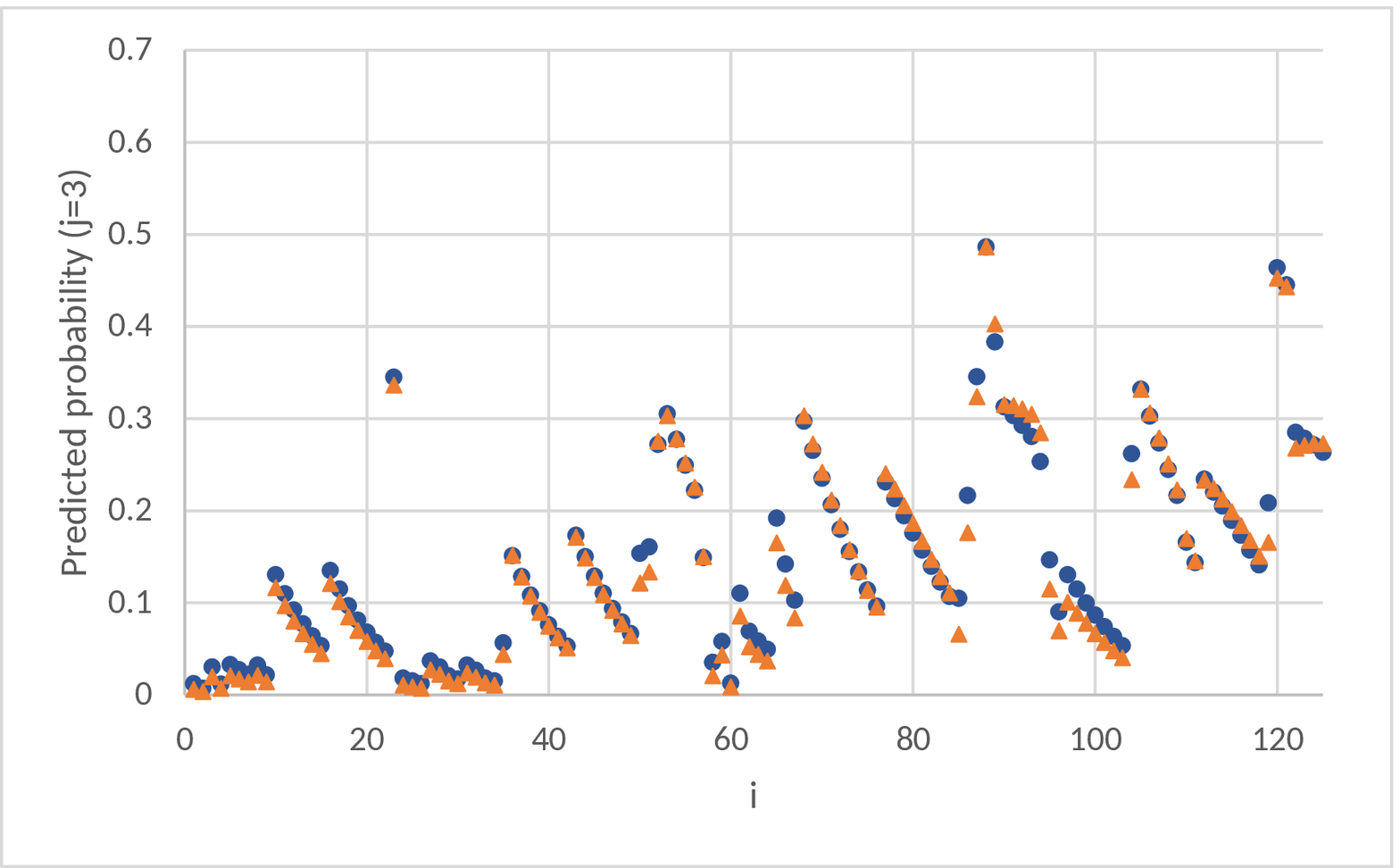}%
	\end{tabular}%
	\bigskip
	\caption{Predicted category probabilities of the response variable, 
		based on the MLE (left) and the MDPDE with $\protect\lambda =0.7$ (right), 
		for the three categories ($j$) of the Mammography experience data
		[Blue circles: Full data; Orange triangles: Outlier deleted data]. }
	\label{fig:mamo_points}
\end{figure}

As an illustration of the proposed Wald-type test for this dataset, 
we consider the problem of testing $H_{0}:\beta _{SYMPT1,1}=0$. 
The p-values obtained based on the proposed test are plotted over $\lambda$ in Figure \ref{fig:mamo_test} (left) 
for both the full and the outlier deleted data. 
Clearly  the test decision at the significance level $\alpha =0.05$ changes completely 
in the presence  of outliers for smaller $\lambda$ including the classical Wald test at $\lambda=0$.
But the inference becomes much stable for larger $\lambda\geq0.3$ implying strong robustness of our proposal.

\begin{figure}[th]
	\centering
	\includegraphics[width=8.1cm,
	height=5.5cm]{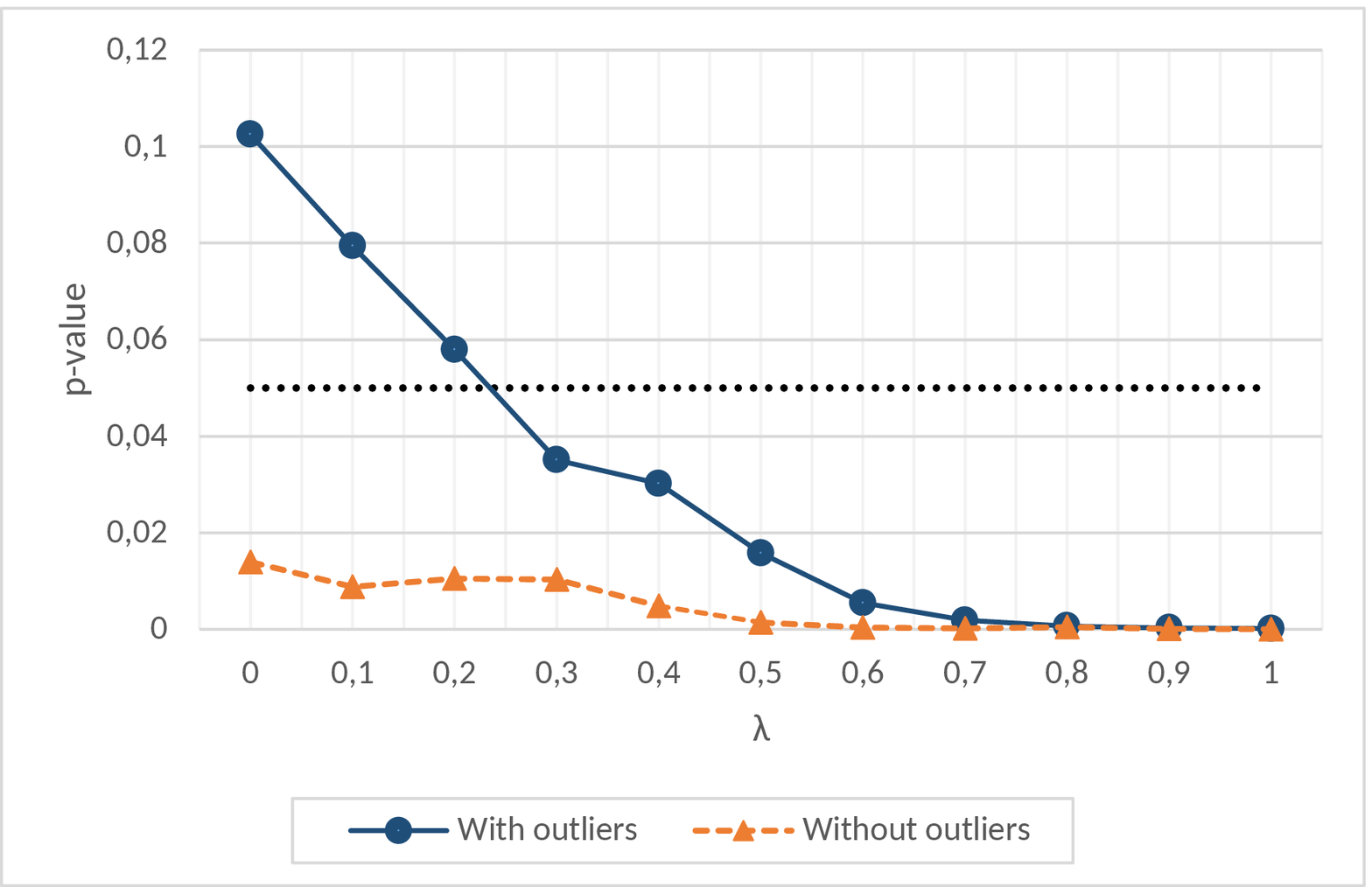} 
	\includegraphics[width=8.1cm,
	height=5.5cm]{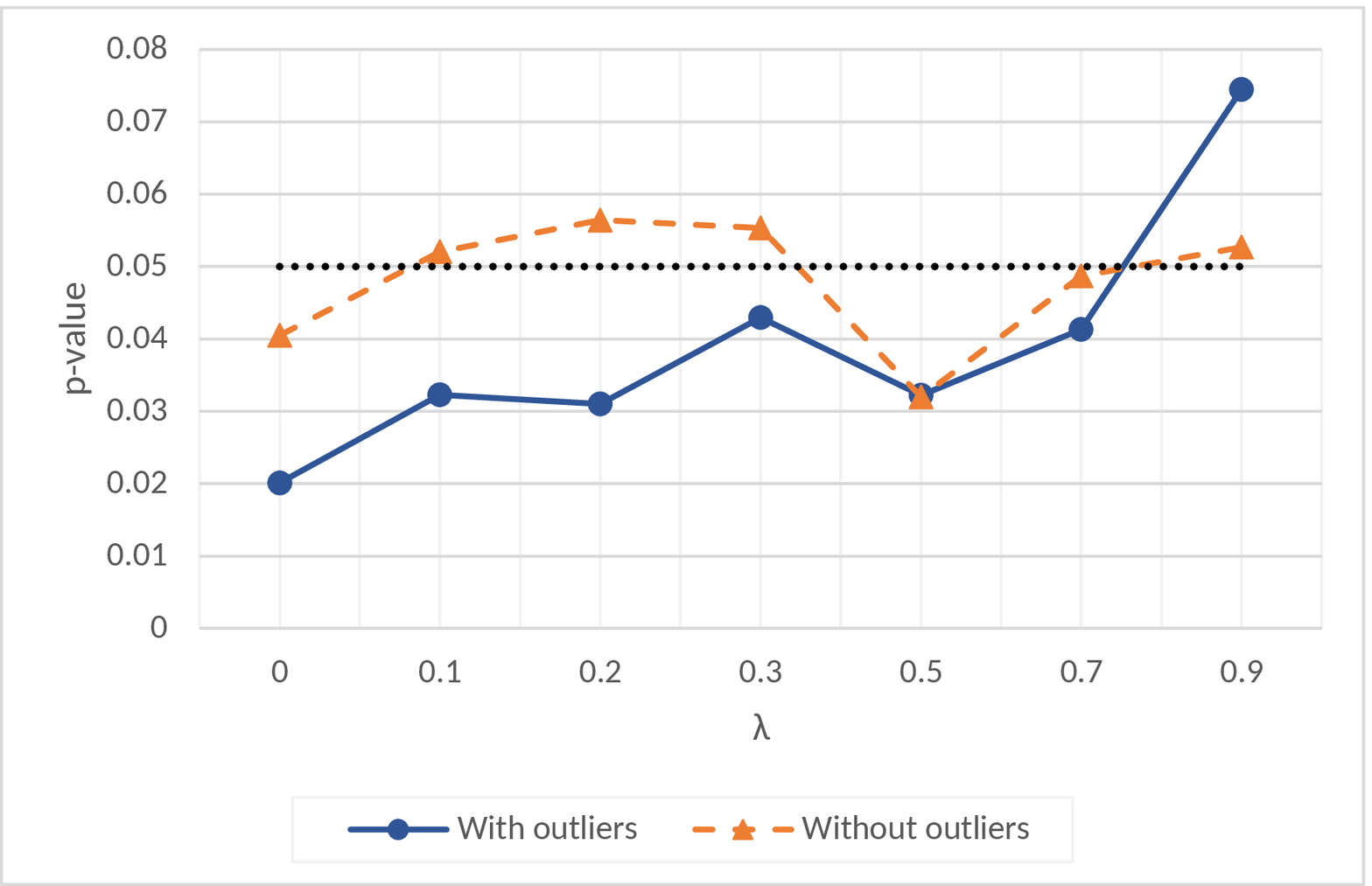}
	\caption{P-values of the proposed MDPDE based Wald-type tests for testing 
		(a) $H_{0}:\protect\beta_{SYMPT11}=0$ in the Mammography data (left) 
		and (b) $H_{0}:\protect\beta _{02}=0$ in the Liver enzyme data (right).}
	\label{fig:mamo_test}
\end{figure}

\subsection{Liver enzyme data (Plomteux)}

Plomteux (1980) showed that the four levels of hepatitis can be explained based
on three liver function tests: aspartate aminotransferase (AST), alanine
aminotransferase (ALT) and glutamate dehygrogenase (GIDH). The associated
dataset (Albert and Harris, 1987) consists of 218 patients with their
hepatitis level being categorized as 1 = acute viral hepatitis, 2 = persistent
chronic hepatitis, 3 = aggressive chronic hepatitis and 4 = post-necrotic cirrhosis ($d=3$). 
Their respective category frequencies are 57, 44, 40 and 77. 
This dataset has been studied in the literature by Mu\~{n}oz-Pichardo et al (2011) and Mart\'{\i}n (2015),
among others.

Here, we model these data with a PLRM where explanatory variables are taken to be 
the logarithms of the three liver function tests, namely 
$\boldsymbol{x}=(1,\mathrm{log}$ $AST,\mathrm{log}$ $ALT,\mathrm{log}$ $GLDH)^{T}$. 
As suggested in Mart\'{\i}n (2015), the observations associated with indices 93, 101, 108, 116, 131 and 136  
of the explanatory variables can be considered as outliers. 
Table 5 in Appendix \ref{tables} shows the MDPDEs of the model parameters in presence or absence of the outliers, 
while Table 6 in Appendix \ref{tables} shows the EMDs of the predicted probabilities with respect to the relative frequencies. 
In this case, all the MDPDEs present a more efficient behavior than MLE 
even after removing the aforementioned outliers; 
this indicates that perhaps there are still some masked outliers left in the data
which were unidentified by the previous studies. 
The advantage of the proposed MDPDEs under such cases is clearly evident from this analysis.

Table 7 in Appendix \ref{tables} shows the EMDs between the predicted probabilities obtained from the full data and the outlier deleted data. Again, moderate and large values of $\lambda>0$ yield lesser deviation than that for the
MLE, indicating their strong robustness against outliers. For testing $H_{0}:\beta _{02}=0$,
the p-values of the proposed Wald-type tests are plotted in Figure \ref{fig:mamo_test}. Note that, for $\lambda =0.5$ the
p-values coincide for both the cases with or without outliers.

\vspace{-0.05in}
\section{Simulation Study\label{S6}}

\subsection{Performance of the MDPDE\label{S6.1}}

We consider a nominal outcome variable $\boldsymbol{Y}$ with $d+1=3$
categories, depending on $k=2$ explanatory variables $x_{1}$ and $x_{2}$.
The true value of the parameter 
$\boldsymbol{\beta }=(\beta _{01},\beta_{11},\beta _{21},\beta _{02},\beta _{12},\beta _{22})^{T}$ 
is taken as  $\boldsymbol{\beta }_{0}=(0,-0.9,0.1,0.6,-1.2,0.8)^{T}$. 
We first simulate pure samples of size $N$ based on 
covariates $x_{i1}\sim\mathcal{N}(0,1)$, $x_{i2}\sim\mathcal{N}(0,1)$ and 
the multinomial responses 
$$\boldsymbol{y}_{i}\sim 
\mathcal{M}%
\left( 1;\pi _{i1}(\boldsymbol{\beta }_{0}),\pi _{i2}(\boldsymbol{\beta }_{0}),
1-\pi _{i1}(\boldsymbol{\beta }_{0})-\pi _{i2}(\boldsymbol{\beta }_{0})\right),
~~i=1,\dots ,N.
$$
Then, to study the robustness, we additionally change the last $p\%$ of responses according to
\begin{eqnarray*}
	\boldsymbol{y}_{i} &\sim &%
	\mathcal{M}%
	\left( 1;1-\pi _{i1}(\boldsymbol{\beta }_{0})-\pi _{i2}(\boldsymbol{\beta }%
	_{0}),\pi _{i1}(\boldsymbol{\beta }_{0}),\pi _{i2}(\boldsymbol{\beta }%
	_{0})\right) ,\quad i=\left[ \frac{p}{100}N\right] +1,\dots ,N.
\end{eqnarray*}%
Note that, although we have simulated the contaminated observations with a model-misspecification point-of-view, 
it indeed also covers the cases of category misspecification. This is because, the contaminated responses are
generated with permuted class probabilities, so that categories 1, 2, 3 in
the original data are now classified as category 2, 3, 1.

The mean square error (MSE) of the MDPDEs are computed based on $1000$ such simulated samples 
which are plotted in Figure \ref{fig:tests_partial} for different $N$, $\lambda$ and $p=5\%$ contaminations. 
In pure data, the MLE (at $\lambda=0$) presents the most efficient behavior
having the least MSE for each sample sizes, while MDPDEs
with larger $\lambda $ have slightly larger MSEs. For
contaminated data the behavior of the MDPDEs is almost the opposite; the
best behavior (least MSE) is obtained for moderate values of  $\lambda $. This becomes more clear with larger sample sizes.

\subsection{Performance of the MDPDE based Wald-type tests\label{S6.2}}

Let us now empirically study the robustness of the MDPDE based Wald-type
tests for the PLRM. The simulation is performed with the same model as in
Section \ref{S6.1}. We first study the observed level (measured as the
proportion of test statistics exceeding the corresponding chi-square
critical value) of the test under the true null hypothesis $H_{0}:\beta
_{02}=0.6$. The resulting p-values are plotted in Figure \ref{fig:tests_partial} for both the pure and the $5\%$ contaminated samples. In contaminated data, the level of the classical Wald test (at $\lambda=0$) as well as the
proposed Wald-type tests with small $\lambda $ break down, while the MDPDE
based Wald-type tests for moderate and large positive $\lambda $ yield greater stability in their levels.

To investigate the power robustness of these tests, we change the true data generating parameter value
to be $\beta _{02}=1.35$ and the resulting empirical powers are plotted in Figure \ref{fig:tests_partial}. Again, the classical Wald test (at $\lambda=0$)
presents the best behavior under pure data, while  the Wald-type
tests with larger $\lambda>0 $ lead to better stability in the contaminated samples.

\begin{figure}[th]
	\begin{tabular}{ll}
		\includegraphics[width=8.2cm,
		height=7.5cm]{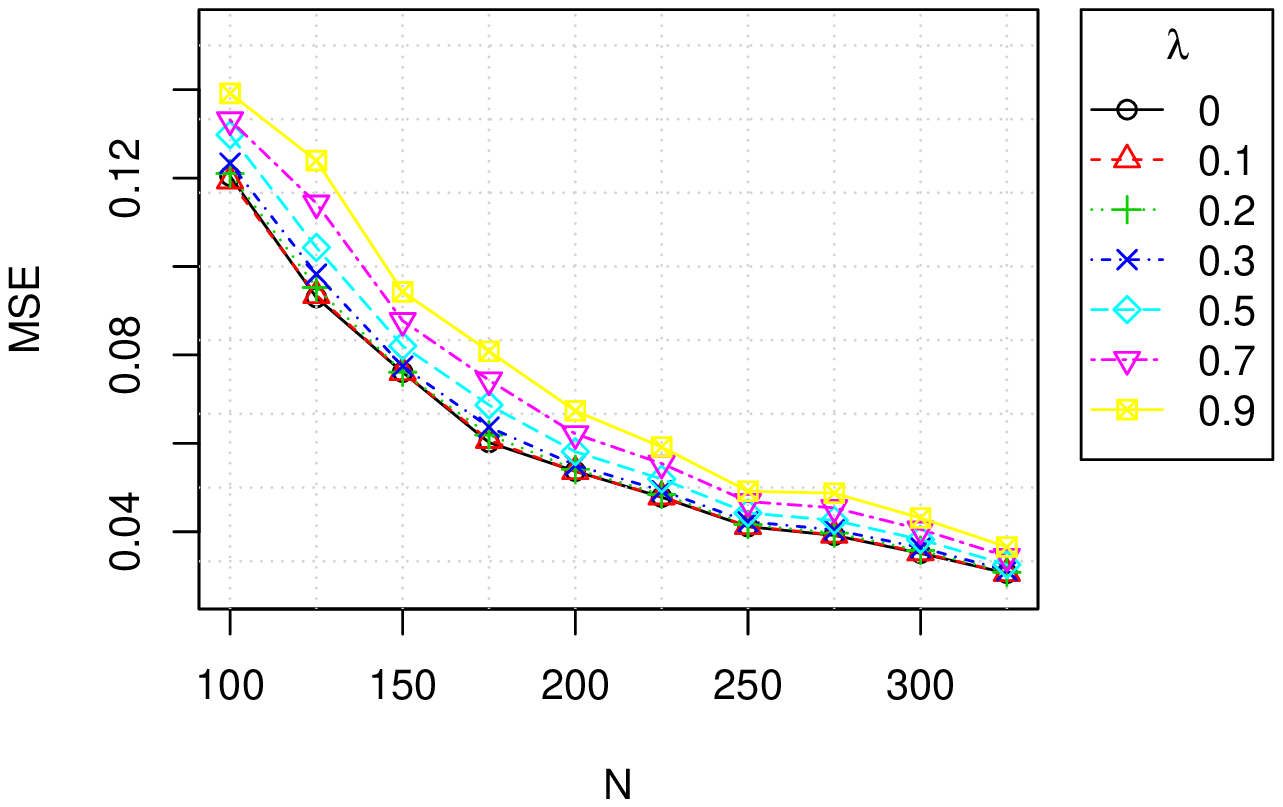} & 
		\includegraphics[width=8.2cm,
		height=7.5cm]{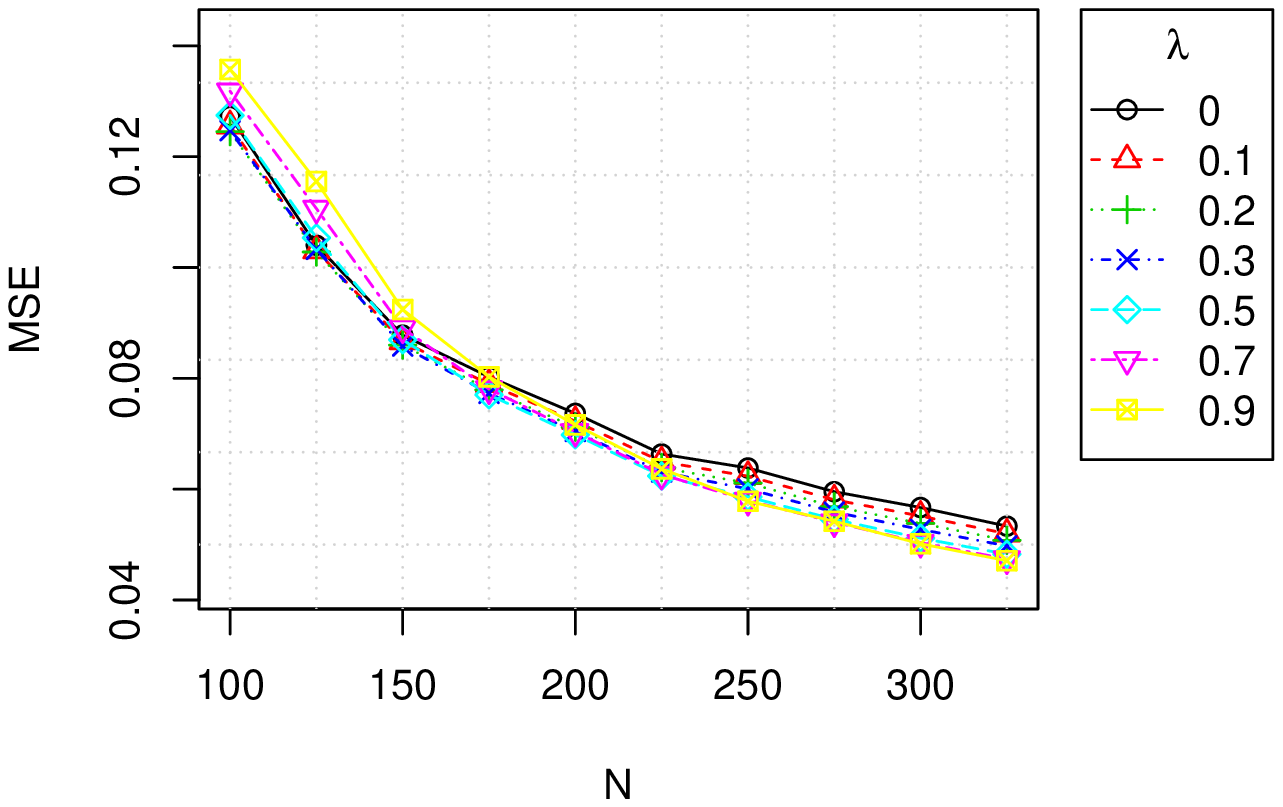} \\ 
		\includegraphics[width=8.3cm,
		height=7.5cm]{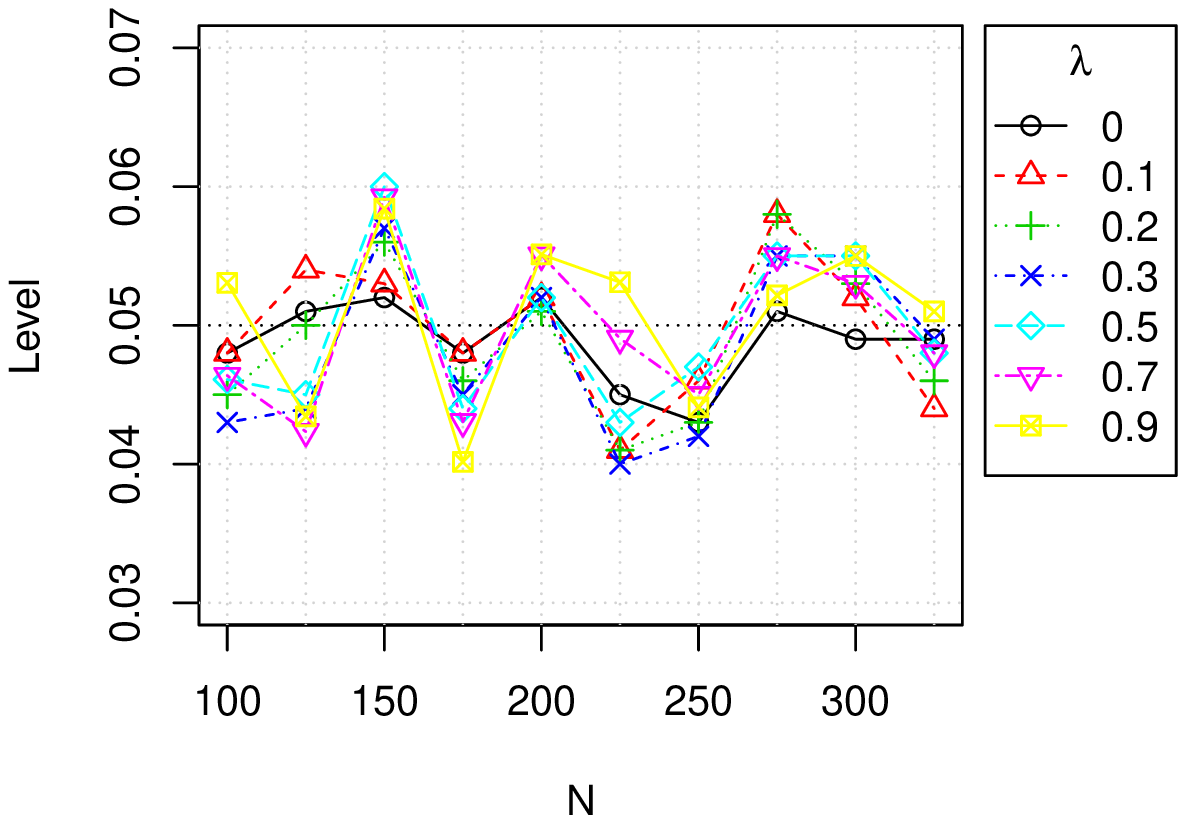} & 
		\includegraphics[width=8.3cm,
		height=7.5cm]{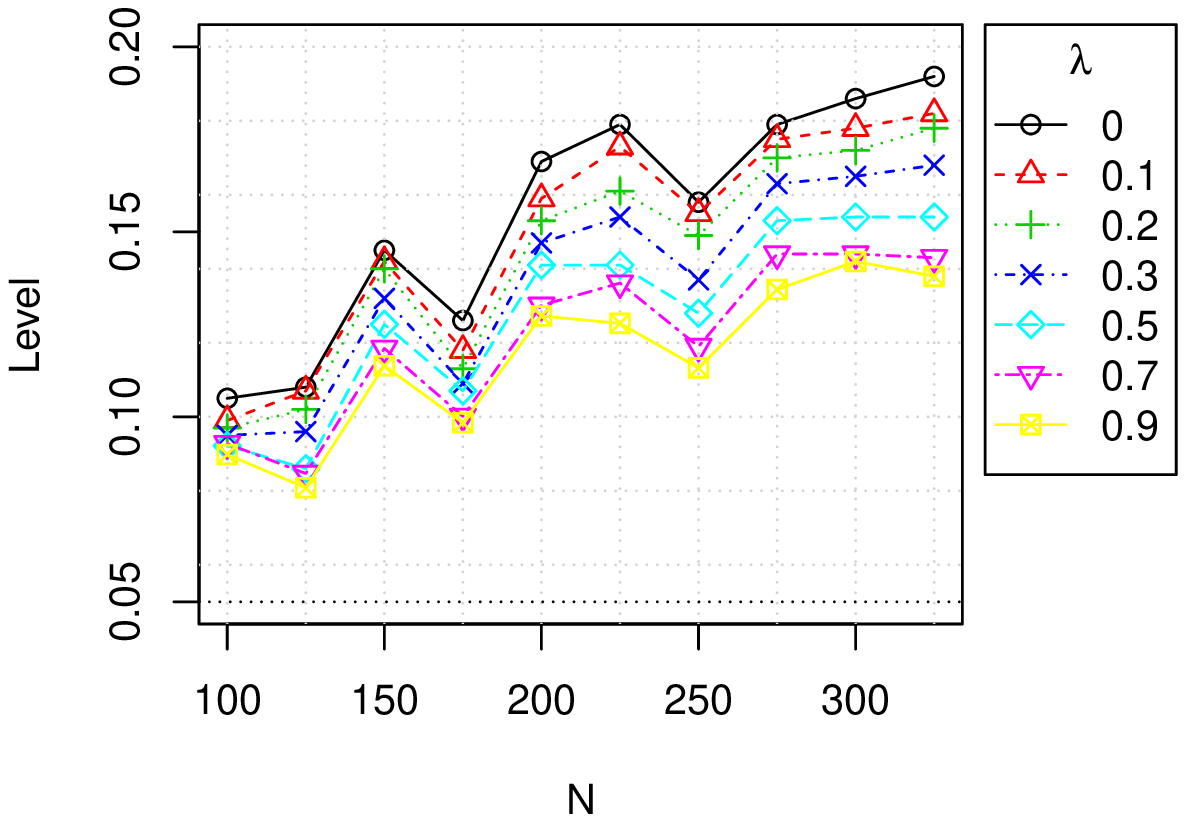} \\ 
		\includegraphics[width=8.3cm,
		height=7.5cm]{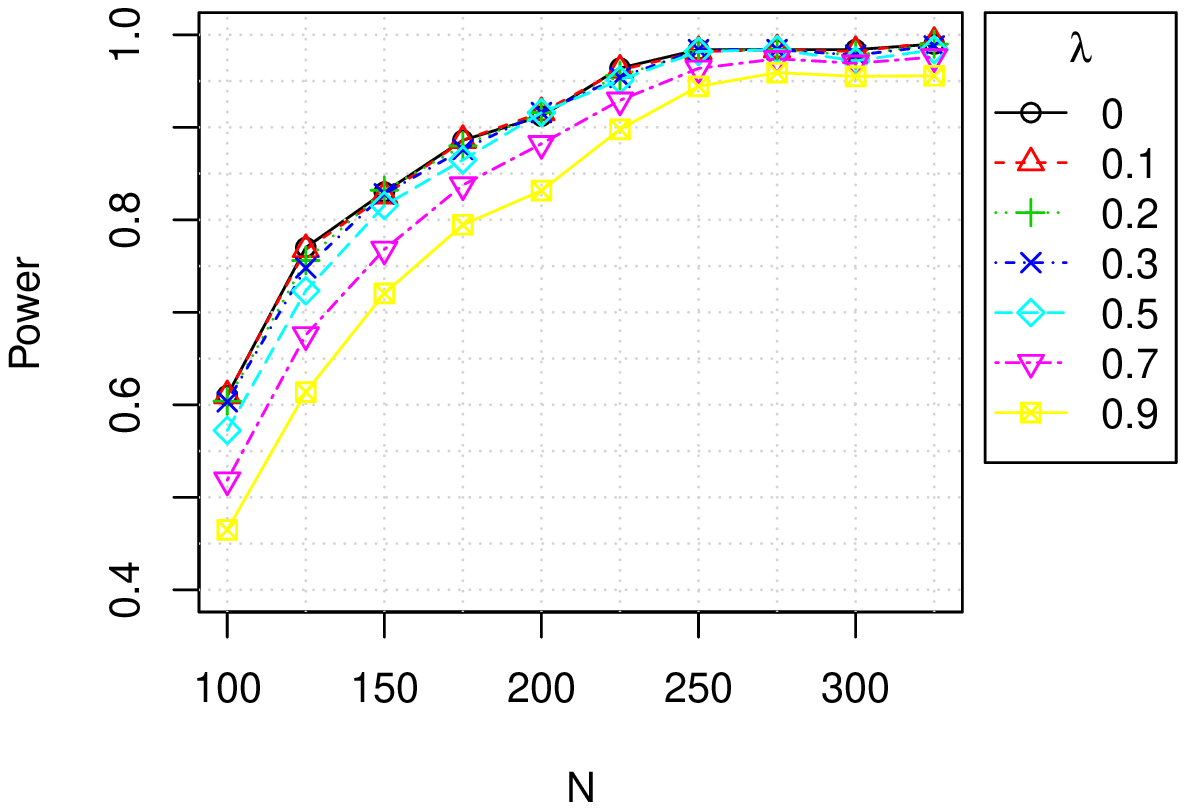} & 
		\includegraphics[width=8.3cm,
		height=7.5cm]{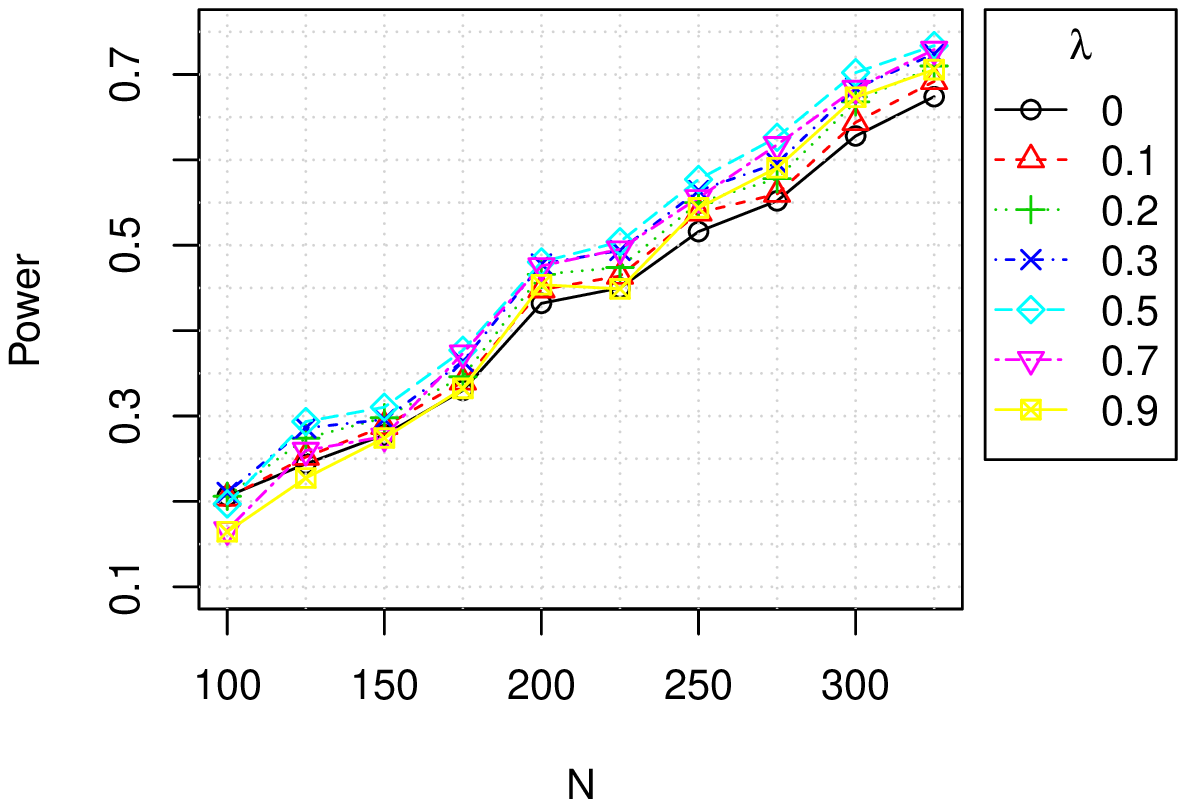}%
	\end{tabular}%
	\caption{MSEs (top panel) of the MDPDEs of  $\boldsymbol{\protect\beta }$ 
		and the simulated levels (middle panel) and powers (bottom panel)
		of the Wald-type tests under  the pure data	(left) and the contaminated data (right).}
	\label{fig:tests_partial}
\end{figure}

\section{On the choice of tuning parameter $\protect\lambda $\label{S7}}

Throughout the previous sections, we have noted that the robustness of both 
the proposed MDPDE and the associated Wald-type tests increase with increasing $\lambda$; 
but their pure data (asymptotic) efficiency and power decrease slightly. 
From our empirical analyses, it seems that a moderately large value of $\lambda$ 
is expected to provide the best trade-off for possibly contaminated data.
However, a data-driven choice of $\lambda $ would be more helpful in practice.

As noted in Section \ref{S4}, the robustness of the Wald-type test 
directly depends on that of the MDPDE used. 
A useful procedure of the data-based selection of $\lambda $ for the MDPDE was proposed by Warwick and Jones (2005) 
under IID data, which is recently extended for the non-homogeneous data 
by Ghosh and Basu (2013, 2015, 2016) and Basu et al. (2017a). 
We can adopt a similar approach to obtain a suitable data-driven $\lambda$ in our PLRM. 
In this approach, we minimize an estimate of the asymptotic MSE of the MDPDE  $\widehat{\boldsymbol{\beta }}_{\lambda }$, 
given by 
$
\mathrm{MSE}(\lambda )=\left( \boldsymbol{\beta }_{\lambda }-\boldsymbol{%
	\beta }^{\ast }\right) ^{T}\left( \boldsymbol{\beta }_{\lambda }-\boldsymbol{%
	\beta }^{\ast }\right) +\frac{1}{N}\mathrm{trace}\left\{\boldsymbol{J}%
_{\lambda }^{-1}\left( \boldsymbol{\beta }_{\lambda }\right) \boldsymbol{V}%
_{\lambda }\left( \boldsymbol{\beta }_{\lambda }\right) \boldsymbol{J}%
_{\lambda }^{-1}\left( \boldsymbol{\beta }_{\lambda }\right) \right\} 
$,
over $\lambda \in \lbrack 0,1]$, 
where $\boldsymbol{\beta }_{\lambda }$ is the asymptotic mean of  $\widehat{\boldsymbol{\beta }}_{\lambda }$ 
and $\boldsymbol{\beta }^{\ast }$ is the true target parameter value.
As pointed out by Basu et. al (2017a), the estimation of the variance component should not assume 
the model to be true for a better robustness trade-off. 
So, following the general formulation of Ghosh and Basu (2015),  
model robust estimates of $\boldsymbol{V}_{\lambda }$ and $\boldsymbol{J}_{\lambda}$  
can be obtained as 
$\widehat{\boldsymbol{V}}_{N,\lambda }=\boldsymbol{\Omega }_{N,\lambda }(\widehat{\boldsymbol{\beta }}_{\lambda })$ 
and 
\begin{align*}
\widehat{\boldsymbol{J}}_{N,\lambda }&=(\lambda +1)%
\boldsymbol{\Psi }_{N,\lambda }(\widehat{\boldsymbol{\beta }}_{\lambda })+%
\frac{1}{N}\sum_{i=1}^{N}\Delta \left(\boldsymbol{\pi }_{i}^{\ast }(%
\widehat{\boldsymbol{\beta }}_{\lambda })\right) \left\{ \sum_{j=1}^{d+1}\pi _{ij}( \widehat{\boldsymbol{\beta }}_{\lambda })
^{\lambda +1}\right\} \otimes 
\boldsymbol{x}_{i}\boldsymbol{x}_{i}^{T}\\
&-\frac{1}{N}\sum_{i=1}^{N}\left\{
\lambda \boldsymbol{u}_{i}(\boldsymbol{y}_{i}^{\ast },\widehat{\boldsymbol{\beta }}_{\lambda })\boldsymbol{u}_{i}(\boldsymbol{y}_{i}^{\ast },\widehat{\boldsymbol{\beta }}_{\lambda })^{T}+\Delta \left( \boldsymbol{\pi }_{i}^{\ast }(%
\widehat{\boldsymbol{\beta }}_{\lambda })\right) \otimes 
\boldsymbol{x}_{i}\boldsymbol{x}_{i}^{T}
\right\} f_{i}(\boldsymbol{y}_{i}^{\ast },\widehat{\boldsymbol{\beta }}_{\lambda })^{\lambda },
\end{align*}
where $\boldsymbol{u}_{i}$ and $f_{i}$ are given in Appendix.
Next, for the bias part, we can use the MDPDE $\widehat{\boldsymbol{\beta }}_{\lambda }$
to estimate  ${\boldsymbol{\beta }}_{\lambda }$ 
but there is no clear choice for estimation of $\boldsymbol{\beta }^{\ast }$;
Warwick and Jones (2005) suggested to estimate $\boldsymbol{\beta }^{\ast }$ 
by some appropriate pilot estimator $\boldsymbol{\beta }^{P}$. 
Note that, the overall performance of this procedure of selecting optimum $\lambda$ 
depends on the choice of $\boldsymbol{\beta }^{P}$, which we will explore through an empirical study.

Consider the same simulation study as in Section 6.1.
We now compute the optimal $\lambda$ value in each iteration 
following the proposed method with a given $\boldsymbol{\beta}^P$.
As potential choices of $\boldsymbol{\beta }^{P}$, 
we consider the MDPDEs with ``pilot" parameters $\lambda_p \in\{0,0.3,0.5,1 \}$. 
For example, when $\lambda_p=0.5$, we fix $\boldsymbol{\beta }^{P}=\widehat{\boldsymbol{\beta}}_{0.5}$
and minimize the estimated quantity
$
\widehat{\mathrm{MSE}}(\lambda )=
\left( \widehat{\boldsymbol{\beta }}_{\lambda }-\widehat{\boldsymbol{\beta}}_{0.5}\right) ^{T}
\left( \widehat{\boldsymbol{\beta }}_{\lambda }-\widehat{\boldsymbol{\beta}}_{0.5}\right) 
+\frac{1}{N}\mathrm{trace}\left(
\widehat{\boldsymbol{J}}_{N,\lambda }^{-1}\widehat{\boldsymbol{V}}_{N,\lambda }\widehat{\boldsymbol{J}}_{N\lambda }^{-1}
\right)$,
through a grid search over $\lambda\in [0,1]$, to obtain the optimum $\lambda$ value.
Note that, the bias term is not generally zero even though we are using MDPDEs as the pilot estimator. 
Figure \ref{fig:perf1} shows the empirical MSEs for the final MDPDEs with 
the resulting optimum $\lambda$ (in each iteration) for the pure and the contaminated data. 
Clearly, the best trade-off between the efficiency in pure data and the robustness under contaminated data 
is provided by the pilot choice $\lambda_p= 0.3$ and the corresponding MSEs are also satisfactorily small in both the cases. 
So, we suggest to use the pilot choice $\boldsymbol{\beta }^{P}=\widehat{\boldsymbol{\beta}}_{0.3}$
for the PLRM and the steps for the final algorithm are clearly mentioned below.

\smallskip
\fbox{\parbox{0.95\textwidth}{
		\noindent\textbf{\underline{Algorithm for the data-driven selection of optimum tuning parameter $\lambda$}~~~}
		\begin{itemize}
			\item \textbf{Aim:} Optimal fitting of the PLRM given a dataset
			\item \textbf{Fix:} Pilot estimate $\boldsymbol{\beta }^{P}=\widehat{\boldsymbol{\beta}}_{0.3}$.
			~~~~~~~~~~~~~~~~~~~~~~~~~~~~~~(Empirical suggestion)
			\item \textbf{For} each $\lambda$ in a grid of $[0,1]$, 
			\textbf{do} the following. ~~~~~~~~~(E.g., $\lambda$ \textbf{in} $0:0.01:1$)
			\begin{itemize}
				\item Compute the total estimated squared bias 
				$B(\lambda)=\left( \widehat{\boldsymbol{\beta }}_{\lambda }-\widehat{\boldsymbol{\beta}}_{0.3}\right) ^{T}
				\left( \widehat{\boldsymbol{\beta }}_{\lambda }-\widehat{\boldsymbol{\beta}}_{0.3}\right)$.
				\item Compute the total estimated variance $V(\lambda) =\frac{1}{N}\mathrm{trace}\left(
				\widehat{\boldsymbol{J}}_{N,\lambda }^{-1}\widehat{\boldsymbol{V}}_{N,\lambda }\widehat{\boldsymbol{J}}_{N\lambda }^{-1}
				\right)$.
				\item Compute the total estimated MSE as $\widehat{\mathrm{MSE}}(\lambda )=B(\lambda)+V(\lambda)$.
			\end{itemize}
			\item \textbf{Find:} Minimum of $\widehat{\mathrm{MSE}}(\lambda)$ and the corresponding $\lambda$. 
			Let $\lambda_{opt}=\arg\min \widehat{\mathrm{MSE}}(\lambda)$.
			\item \textbf{Return:} $\lambda_{opt}$ as the final optimum value of the tuning parameter.
			\item Compute $\widehat{\boldsymbol{\beta}}_{\lambda_{opt}}$ 
			as your final estimate with optimally chosen tuning parameter.
		\end{itemize}
}}

\begin{figure}[ht!]
	\begin{tabular}{ll}
		\includegraphics[width=8.3cm,
		height=7.5cm]{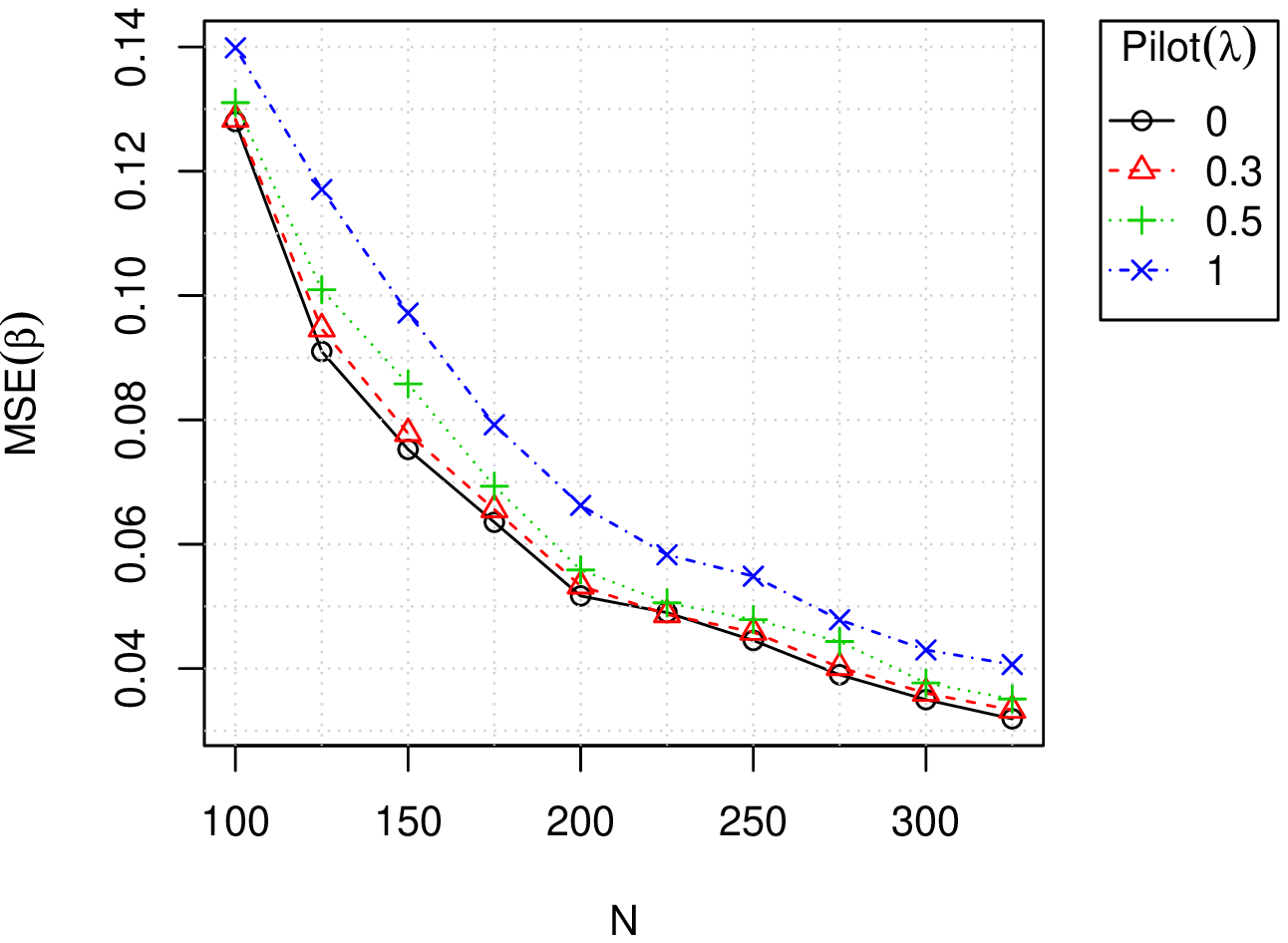} & 
		\includegraphics[width=8.3cm,
		height=7.5cm]{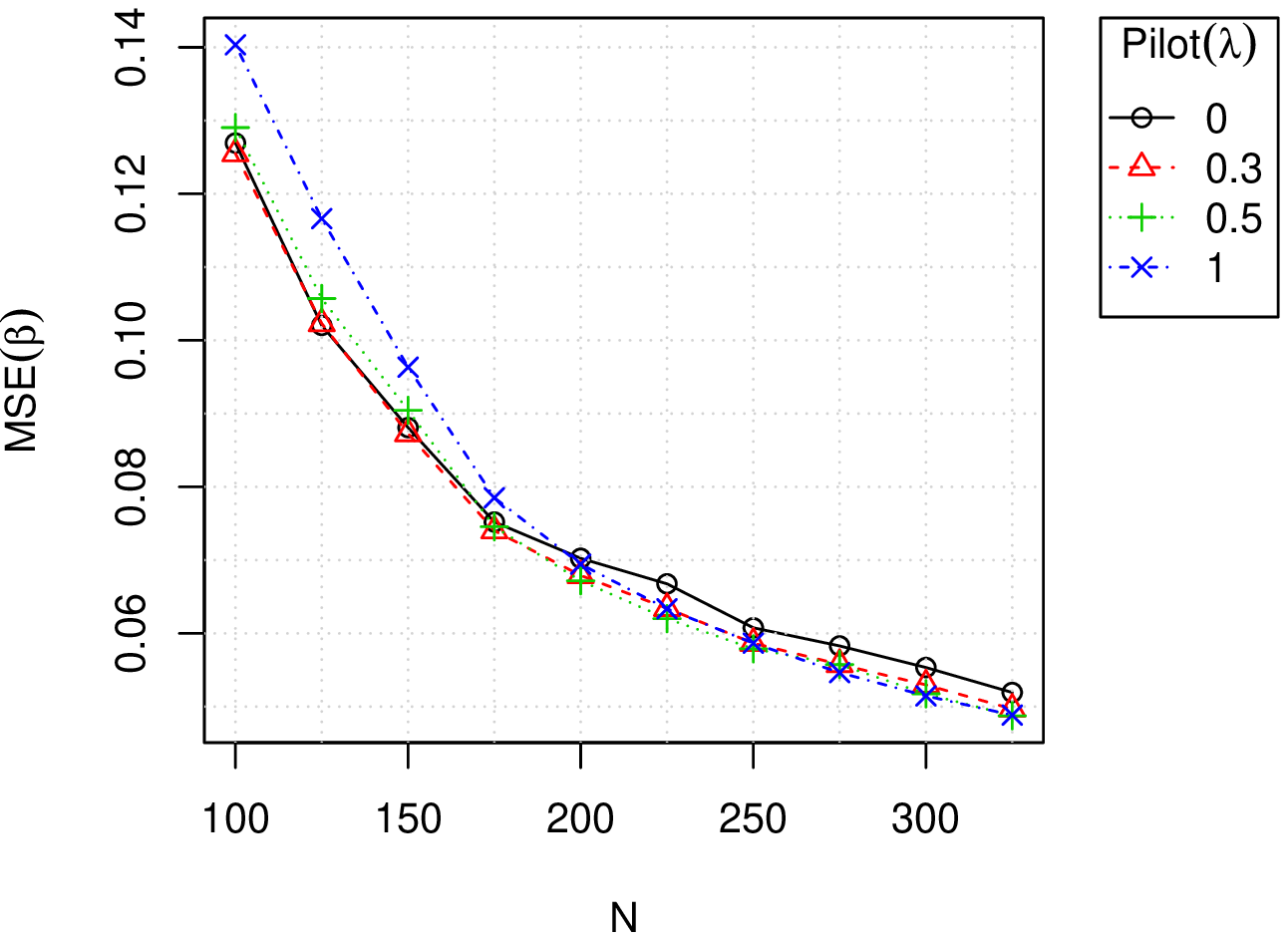} 
	\end{tabular}%
	\caption{Simulated MSEs of the MDPDEs at the optimally chosen $\protect\lambda$,
		starting from different pilot estimators under the pure data (left) and the contaminated data (right). }
	\label{fig:perf1}
\end{figure}

We now apply this proposed algorithm (with a grid of spacing 0.05) to our real datasets.
The optimum $\lambda$ turns out to be $1$ and $0.75$, respectively, 
for the Mammography experience data and the Liver enzyme data in presence of the outliers. 
After deleting the aforementioned outliers from the data, 
these optimum $\lambda$ values become $0.05$ and $0.35$, respectively. 
These optimum $\lambda$ values generate quite stable MDPDEs for both datasets, as we have seen in Section \ref{S5}. 
They indeed yield the automatic data-driven choices of $\lambda$ which are 
also consistent with the fact that we should use larger $\lambda>0$ for contaminated data 
and smaller $\lambda$ close to zero for clean data.
Note that, as discussed in Section \ref{S5}, 
the liver enzyme data might contain some masked outlying observations 
even after removing the aforementioned outliers from previous studies
which leads to the slightly larger optimum $\lambda$ value of 0.35. 
These evidences clearly justify the appropriateness and usefulness of our proposed algorithm 
of selecting optimum data-driven $\lambda$ for the MDPDEs in case of the PLRM.

\section{Concluding Remarks\label{S8}}

The PLRM is an extensively applied statistical tool 
which is widely used in many different areas, including health and life sciences.  
Although the classical inference procedures based on the MLE have asymptotic optimal properties, 
they are highly non-robust against outliers in data. 
So, there is a strong need for robust procedures in practical applications of the PLRM 
due to the presence of potential outlying observations in many real datasets. 
Here we derive a new family of estimators, MDPDEs, as a robust generalization of the MLE for the PLRM, 
exploring the \textquotedblleft nice\textquotedblright\ robustness properties of the DPD measure. 
A family of Wald-type test statistics based on the MDPDEs is also introduced 
for testing linear hypotheses under the PLRM. 
The study of two real data examples from medical sciences as well as the simulation results 
illustrate the advantages of our proposed inference procedures.

\section*{Acknowledgements}
We would like to thank the Editor, the Associate Editor and the Referees for
their helpful comments and suggestions. This research is supported by Grant
MTM2015-67057-P, from Ministerio de Economia y Competitividad (Spain) and by
the INSPIRE Faculty Research Grant from Department of Science \& Technology,
Govt. of India.

\newpage
\appendix
\section{Supplementary Materials: Appendices and Tables
	referenced in Sections \ref{S3}, \ref{S5}  \ref{S6} and \ref{S7}}
\subsection{\protect\normalsize How to obtain the MDPDEs? }

Consider the set-up and notations of Section 2 of the main paper. The
following theorem presents the estimating equation for the MDPDEs in the
PLRM, which can be solved numerically to obtain the estimates.

\begin{theorem}
	\label{THM:Est_Eqn} \label{Th0}The MDPDE, $\widehat{\boldsymbol{\beta }}%
	_{\lambda }$, of $\boldsymbol{\beta }$ can be obtained by solving the system
	of equations 
	\begin{equation*}
	\boldsymbol{u}_{\lambda }\left( \boldsymbol{\beta }\right) =\boldsymbol{0}%
	_{d(k+1)},
	\end{equation*}%
	where 
	\begin{equation*}
	\boldsymbol{u}_{\lambda }\left( \boldsymbol{\beta }\right) =	\sum_{i=1}^{N}\left[ \left( \boldsymbol{I}_{d},\boldsymbol{0}_{d}\right) 
	\boldsymbol{\Delta }(\boldsymbol{\pi }_{i}\left( \boldsymbol{\beta }\right) )%
	\mathrm{diag}^{\lambda -1}\{\boldsymbol{\pi }_{i}\left( \boldsymbol{\beta }%
	\right) \}\{\boldsymbol{y}_{i}-\boldsymbol{\pi }_{i}\left( \boldsymbol{\beta }%
	\right) \}\right] \otimes \boldsymbol{x}_{i}
	\end{equation*}
\end{theorem}

\begin{proof}
	The MDPDE of $\boldsymbol{\beta }$, is defined as
	\begin{align*}
	\widehat{\boldsymbol{\beta }}_{\lambda }& =\underset{\boldsymbol{\beta }\in 
		\mathbb{R}
		^{d(k+1)}}{\arg \min }\sum_{i=1}^{N}\left\{ \boldsymbol{\pi }_{i}\left( 
	\boldsymbol{\beta }\right) -\frac{\lambda +1}{\lambda }\boldsymbol{y}%
	_{i}\right\} ^{T}\boldsymbol{\pi }_{i}^{(\lambda )}\left( \boldsymbol{\beta }%
	\right) \\
	& =\underset{\boldsymbol{\beta }\in 
		\mathbb{R}
		^{d(k+1)}}{\arg \min }\sum_{i=1}^{N}\sum_{j=1}^{d+1}\left\{ \pi
	_{ij}^{\lambda +1}(\boldsymbol{\beta })-\frac{\lambda +1}{\lambda }y_{ij}\pi
	_{ij}^{\lambda }(\boldsymbol{\beta })\right\},
	\end{align*}%
	which can also be obtained by solving the system of equations $\boldsymbol{u}_{\lambda
	}\left( \boldsymbol{\beta }\right) =\boldsymbol{0}_{d(k+1)}$ where 
	\begin{align*}
	-\frac{1}{\lambda +1}\frac{\partial }{\partial \boldsymbol{\beta }}%
	\sum_{i=1}^{N}\left\{\boldsymbol{\pi }_{i}\left( \boldsymbol{\beta }\right) -%
	\frac{\lambda +1}{\lambda }\boldsymbol{y}_{i}\right\} ^{T}\boldsymbol{\pi }%
	_{i}^{(\lambda )}\left( \boldsymbol{\beta }\right) & =\boldsymbol{0}%
	_{d(k+1)}, \\
	\sum_{i=1}^{N}\sum_{j=1}^{d+1}\{y_{ij}-\pi _{ij}(\boldsymbol{\beta })\}\pi
	_{ij}^{\lambda +1}(\boldsymbol{\beta })\frac{\partial }{\partial \boldsymbol{%
			\beta }}\pi _{ij}(\boldsymbol{\beta })& =\boldsymbol{0}_{d(k+1)}, \\
	\sum_{i=1}^{N}\frac{\partial }{\partial \boldsymbol{\beta }}\boldsymbol{\pi }%
	_{i}\left( \boldsymbol{\beta }\right) ^{T}\mathrm{diag}^{\lambda -1}\{
	\boldsymbol{\pi }_{i}\left( \boldsymbol{\beta }\right) \}\{\boldsymbol{y}_{i}-%
	\boldsymbol{\pi }_{i}\left( \boldsymbol{\beta }\right)  \}& =\boldsymbol{0}%
	_{d(k+1)}.
	\end{align*}%
	Now,  taking into account that 
	\begin{equation*}
	\frac{\partial \pi _{ij}(\boldsymbol{\beta })}{\partial \beta _{uv}}%
	=x_{iu}\pi _{ij}(\boldsymbol{\beta })\left\{ \delta _{jv}-\pi _{iv}(%
	\boldsymbol{\beta })\right\} ,\qquad u=1,...,k,v=1,...,d,
	\end{equation*}%
	we get 
	\begin{equation}
	\frac{\partial \boldsymbol{\pi }_{i}(\boldsymbol{\beta })^{T}}{\partial 
		\boldsymbol{\beta }}=\left( \boldsymbol{I}_{d},\boldsymbol{0}_{d}\right) 
	\boldsymbol{\Delta }(\boldsymbol{\pi }_{i}\left( \boldsymbol{\beta }\right)
	)\otimes \boldsymbol{x}_{i},  \label{der}
	\end{equation}%
	and hence the system of equations becomes%
	\begin{equation*}
	\sum_{i=1}^{N}\left[ \left( \boldsymbol{I}_{d},\boldsymbol{0}_{d}\right) 
	\boldsymbol{\Delta }(\boldsymbol{\pi }_{i}\left( \boldsymbol{\beta }\right) )%
	\mathrm{diag}^{\lambda -1}\{\boldsymbol{\pi }_{i}\left( \boldsymbol{\beta }%
	\right) \}\{\boldsymbol{y}_{i}-\boldsymbol{\pi }_{i}\left( \boldsymbol{\beta }%
	\right) \}\right] \otimes \boldsymbol{x}_{i}=\boldsymbol{0}_{d(k+1)}.
	\end{equation*}%
\end{proof}

Note that, under the PLRM, the  tabulated response variables $%
\boldsymbol{Y}_{i}$, $i=1,...,N$, are independent but not identically
distributed since the  covariates  $\boldsymbol{x}_{i}$ are generally assumed
to be pre-fixed (and hence different over $i$). In particular, for each $%
i=1,\ldots ,N$, given $\boldsymbol{x}_{i}$, $\boldsymbol{Y}_{i}$ has a
multinomial distribution having joint probability mass function 
\begin{equation}
f_{i}(\boldsymbol{y}_{i},\boldsymbol{\beta })=\prod_{j=1}^{d+1}\pi
_{ij}\left( \boldsymbol{\beta }\right) ^{y_{ij}}=\sum_{j=1}^{d+1}{y_{ij}}\pi
_{ij}\left( \boldsymbol{\beta }\right) ,\quad y_{ij}\in \{0,1\}\text{ for }%
j=1,\ldots ,d+1\text{, with }\sum_{j=1}^{d+1}y_{ij}=1.  \label{EQ:pmf1}
\end{equation}%
This distribution indeed belongs to the family of $d$-dimensional
Generalized Linear Models (GLM), where the distribution of $\boldsymbol{Y}%
_{i}^{\ast }=(Y_{i1},\dots ,Y_{id})^{T}$ can be written as 
\begin{align*}
f_{i}(\boldsymbol{y}_{i}^{\ast },\boldsymbol{\beta })& =\exp \left\{ 
\boldsymbol{\eta }_{i}^{T}\boldsymbol{y}_{i}^{\ast }-b(\boldsymbol{\eta }%
_{i})\right\} , \\
\boldsymbol{\eta }_{i}& =(\eta _{i1},...,\eta _{id})^{T},\quad \eta _{ij}=%
\boldsymbol{x}_{i}^{T}\boldsymbol{\beta }_{j},\text{ }j=1,...,d, \\
b(\boldsymbol{\eta }_{i})& =\log \left\{ 1+\sum_{j=1}^{d}\exp (\eta
_{ij})\right\} .
\end{align*}%
The general estimating equation of Ghosh and Basu (2013), based on
these forms of the model densities, can be seen to coincide with the one given
in Theorem 1.

\newpage

\subsection{\protect\normalsize Power function of the Wald-type tests}

Consider the set-up and notations of Section 3 of the main paper. We
consider $\boldsymbol{\beta }_{\ast }$ $\in \Theta $ such that $\boldsymbol{L%
}^{T}\boldsymbol{\beta }_{\ast }\neq \boldsymbol{h,}$ i.e., $\boldsymbol{%
	\beta }_{\ast }$ does not satisfy the null hypothesis given in Equation (6) of the main paper. Let us denote
\begin{equation}
q_{\boldsymbol{\beta }_{1}}({\boldsymbol{\beta }_{2}})=\left( \boldsymbol{L}%
^{T}\boldsymbol{\beta }_{1}-\boldsymbol{h}\right) ^{T}\left\{\boldsymbol{L}%
^{T}\boldsymbol{\Sigma }_{\lambda }\left( \boldsymbol{\beta }_{2}\right) 
\boldsymbol{L}\right\} ^{-1}\left( \boldsymbol{L}^{T}\boldsymbol{\beta }_{1}-%
\boldsymbol{h}\right)  \label{3.03}
\end{equation}%
and derive an approximation to the power function for the MDPDE based Wald-type test
with the rejection rule given by 
\begin{equation}
W_{N}(\widehat{\boldsymbol{\beta }}_{\lambda })>\chi _{r,\alpha }^{2}.
\label{3.3}
\end{equation}

\begin{theorem}
	\label{THM:PowerApprox_Wald} Let $\boldsymbol{\beta }_{\ast }\in \Theta$, 
	with $\boldsymbol{L}^{T}\boldsymbol{\beta }_{\ast }\neq \boldsymbol{h}$, be the
	true value of the parameter such that $\widehat{\boldsymbol{\beta }}_{\lambda }%
	\underset{N\rightarrow \infty }{\overset{P}{\longrightarrow }}\boldsymbol{%
		\beta }_{\ast }$. The power function of the Wald-type test  given
	in (\ref{3.3}),  is given by 
	\begin{equation}
	\beta _{N,\lambda }\left( \boldsymbol{\beta }_{\ast }\right) =1-\Phi
	_{N}\left( \frac{1}{\sigma _{\lambda }\left( \boldsymbol{\beta }_{\ast
		}\right) }\left( \frac{\chi _{s,\alpha }^{2}}{\sqrt{N}}-\sqrt{N}q_{%
		\boldsymbol{\beta }_{\ast }}(\boldsymbol{\beta }_{\ast })\right) \right)
	\label{3.4}
	\end{equation}%
	where $\Phi _{N}\left( x\right) $ uniformly tends to the standard normal
	distribution function $\Phi(x)$  as $N\rightarrow \infty $  and
	\begin{equation*}
	\sigma _{\lambda }\left( \boldsymbol{\beta }_{\ast }\right)=2 \sqrt{\left( 
		\boldsymbol{L}^{T}\boldsymbol{\beta }_{\ast }-\boldsymbol{h}\right)
		^{T}\left\{ \boldsymbol{L}^{T}\boldsymbol{\Sigma }_{\lambda }\left( 
		\boldsymbol{\beta }_{\ast }\right) \boldsymbol{L}\right\} ^{-1}\left( 
		\boldsymbol{L}^{T}\boldsymbol{\beta }_{\ast }-\boldsymbol{h}\right)} .
	\end{equation*}
\end{theorem}

\begin{proof}
	We have 
	\begin{align*}
	\beta _{N,\lambda }\left( \boldsymbol{\beta }_{\ast }\right) & =\Pr \left(
	W_{N}(\widehat{\boldsymbol{\beta }}_{\lambda })>\chi _{r,\alpha }^{2}\right)
	=\Pr \left( N\left( q_{\widehat{\boldsymbol{\beta }}_{\lambda }}(\widehat{%
		\boldsymbol{\beta }}_{\lambda })-q_{\boldsymbol{\beta }_{\ast }}(\boldsymbol{%
		\beta }_{\ast })\right) >\chi _{r,\alpha }^{2}-Nq_{\boldsymbol{\beta }%
		_{\ast }}(\boldsymbol{\beta }_{\ast })\right) \\
	& =\Pr \left( \sqrt{N}\left( q_{\widehat{\boldsymbol{\beta }}_{\lambda }}(%
	\widehat{\boldsymbol{\beta }}_{\lambda })-q_{\boldsymbol{\beta }_{\ast }}(%
	\boldsymbol{\beta }_{\ast })\right) >\frac{\chi _{r,\alpha }^{2}}{\sqrt{%
			N}}-\sqrt{N}q_{\boldsymbol{\beta }_{\ast }}(\boldsymbol{\beta }_{\ast
	})\right) .
	\end{align*}%
	
	Note that, since $\widehat{\boldsymbol{\beta }}_{\lambda }\underset{N\rightarrow \infty 
	}{\overset{P}{\longrightarrow }}\boldsymbol{\beta }_{\ast }$, $q_{\widehat{%
			\boldsymbol{\beta }}_{\lambda }}(\widehat{\boldsymbol{\beta }}_{\lambda })$
	and $q_{\widehat{\boldsymbol{\beta }}_{\lambda }}(\boldsymbol{\beta }_{\ast
	})$ have the same asymptotic distribution. But, a first Taylor expansion of $q_{%
		\beta}(\boldsymbol{\beta }_{\ast })$ 
	around $\boldsymbol{\beta}^\ast$ 
	at $\widehat{\boldsymbol{\beta}}_\lambda$
	gives 
	\begin{equation*}
	q_{\widehat{\boldsymbol{\beta }}_{\lambda }}(\boldsymbol{\beta }_{\ast })-q_{%
		\boldsymbol{\beta }_{\ast }}(\boldsymbol{\beta }_{\ast })=\left. \frac{%
		\partial q_{\boldsymbol{\beta }}(\boldsymbol{\beta }_{\ast })}{\partial 
		\boldsymbol{\beta }^{T}}\right\vert _{\boldsymbol{\beta }=\boldsymbol{\beta }%
		_{\ast }}(\widehat{\boldsymbol{\beta }}_{\lambda }-\boldsymbol{\beta }_{\ast
	})+o_{p}\left( \left\Vert \widehat{\boldsymbol{\beta }}_{\lambda }-%
	\boldsymbol{\beta }_{\ast }\right\Vert \right) .
	\end{equation*}%
	Therefore, 
	\begin{equation*}
	\sqrt{N}(q_{\widehat{\boldsymbol{\beta }}_{\lambda }}(\widehat{\boldsymbol{\beta }}_{\ast })-q_{\boldsymbol{\beta }_{\ast }}(\boldsymbol{\beta }_{\ast }))\overset{\mathcal{L}}{\underset{N\rightarrow \infty }{\longrightarrow }}%
	\mathcal{N}\left( 0,\sigma _{\lambda }\left( \boldsymbol{\beta }_{\ast }\right)^{2}\right)
	\end{equation*}%
	where
	\begin{equation*}
	\sigma _{\lambda }\left( \boldsymbol{\beta }_{\ast }\right) ^{2}=\left. 
	\frac{\partial q_{\boldsymbol{\beta }}(\boldsymbol{\beta }_{\ast })}{%
		\partial \boldsymbol{\beta }^{T}}\right\vert _{\boldsymbol{\beta }=%
		\boldsymbol{\beta }_{\ast }}\boldsymbol{\Sigma }_{\lambda }\left( 
	\boldsymbol{\beta }_{\ast }\right) \left. \frac{\partial q_{\boldsymbol{%
				\beta }}(\boldsymbol{\beta }_{\ast })}{\partial \boldsymbol{\beta }}%
	\right\vert _{\boldsymbol{\beta }=\boldsymbol{\beta }_{\ast }}.
	\end{equation*}%
	But 
	\begin{equation*}
	\left. \frac{\partial q_{\boldsymbol{\beta }}(\boldsymbol{\beta }_{\ast })}{%
		\partial \boldsymbol{\beta } ^{T}}\right\vert _{\boldsymbol{\beta }=\boldsymbol{%
			\beta }_{\ast }}^{T}=2\left( \boldsymbol{L}^{T}\boldsymbol{\beta }_{\ast }-%
	\boldsymbol{h}\right)\left\{ \boldsymbol{L}^{T}\boldsymbol{\Sigma }%
	_{\lambda }\left( \boldsymbol{\beta }_{\ast }\right) \boldsymbol{L}\right\}
	^{-1}\boldsymbol{L}^{T},
	\end{equation*}%
	and hence the result follows.
\end{proof}

\begin{remark}
	Based on the previous theorem we can obtain the minimum sample size which is necessary to achieve
	a fix power, say $\beta _{N,\lambda }\left( \boldsymbol{\beta }_{\ast
	}\right) =\overline{\beta }$. Based in equation (\ref{3.4}), we must solve
	the equation 
	\begin{equation*}
	1-\overline{\beta }=\Phi \left( \frac{1}{\sigma _{\lambda }\left( 
		\boldsymbol{\beta }_{\ast }\right) }(\tfrac{\chi _{r,\alpha }^{2}}{\sqrt{N}}-%
	\sqrt{N}q_{\boldsymbol{\beta }_{\ast }}(\boldsymbol{\beta }_{\ast }))\right)
	\end{equation*}%
	and we get that $N=\left[ N^{\ast }\right] +1$ with 
	\begin{equation*}
	N^{\ast }=\frac{A+B+\sqrt{A(A+2B)}}{2q_{\boldsymbol{\beta }_{\ast }}({%
			\boldsymbol{\beta }_{\ast }})^{2}}
	\end{equation*}%
	being $A=\sigma _{\lambda }\left( \boldsymbol{\beta }_{\ast }\right)
	^{2}\left( \Phi ^{-1}\left( 1-\overline{\beta }\right) \right) ^{2}$ and $B=%
	\frac{1}{2}\chi _{r,\alpha }^{2}q_{\boldsymbol{\beta }_{\ast }}(${$%
		\boldsymbol{\beta }_{\ast }$}$)$ and $q_{\boldsymbol{\beta }_{\ast }}(${$%
		\boldsymbol{\beta }_{\ast }$}$)$ is given in (\ref{3.03})$.$ 
\end{remark}

\begin{remark}
	It also follows from Theorem \ref{THM:PowerApprox_Wald} that $\beta_{N,\lambda}(\beta^\ast) \rightarrow 1$, as $N\rightarrow \infty$, for all $\lambda\geq 0$. Therefore the proposed MDPDE based Wald-type tests are consistent at any fixed alternative.
\end{remark}

We now derive the asymptotic power function for the Wald-type test  with the rejection rule given in (\ref{3.3}) at an
alternative hypothesis close to the null hypothesis. Consider the parameter value $\boldsymbol{\beta }%
_{N} \in \Theta$ with $\boldsymbol{L}^{T}\boldsymbol{\beta }$ $\boldsymbol{_{N}}\neq 
\boldsymbol{h}$, and the alternative hypothesis   given by $\boldsymbol{\beta}=\boldsymbol{\beta}_N$. Suppose $\boldsymbol{\beta 
}_{0}$ be the closest  parameter value to $\boldsymbol{\beta}_N$ in the Euclidean distance such that $\boldsymbol{L}^{T}\boldsymbol{\beta }_{0}=%
\boldsymbol{h}$. A first possibility to introduce such contiguous alternative
hypotheses is to consider a fixed $\boldsymbol{d}\in 
\mathbb{R}
^{kd}$ and to permit $\boldsymbol{\beta }_{N}$ moving towards $\boldsymbol{%
	\beta }_{0}$ as $N$ increases in the following way 
\begin{equation}
H_{1,N}:\beta = \boldsymbol{\beta }_{N}=\boldsymbol{\beta }_{0}+N^{-1/2}\boldsymbol{d%
}.  \label{3.5}
\end{equation}%
A second approach is to relax the existence of a closest element $\boldsymbol{\beta}_0$ in the null parameter space and consider the sequence $\boldsymbol{\beta}_N$ such that for  $\boldsymbol{%
	\delta }\in \mathbb{R}^{kd}$, it satisfies
\begin{equation}
H_{1,N}^{\ast }:\boldsymbol{L}^{T}\boldsymbol{\beta }_{N}=\boldsymbol{h}%
-N^{-1/2}\boldsymbol{\delta }.  \label{3.6}
\end{equation}%
Note that, whenever the closest null parameter $\boldsymbol{\beta}_0$ exists, we have 
\begin{equation}
\boldsymbol{L}^{T}\boldsymbol{\beta }_{N}-\boldsymbol{h=L}^{T}\boldsymbol{%
	\boldsymbol{\beta }}_{0}\boldsymbol{+L}^{T}N^{-1/2}\boldsymbol{\boldsymbol{d}%
}-\boldsymbol{h}=N^{-1/2}\boldsymbol{L}^{T}\boldsymbol{\boldsymbol{d}}.
\label{3.7}
\end{equation}%
Then the equivalence between the two hypotheses (\ref{3.5}) and (\ref{3.6}) is given by 
\begin{equation}
\boldsymbol{L}^{T}\boldsymbol{\boldsymbol{d}}=\boldsymbol{\delta }.
\label{3.8}
\end{equation}

If we denote by $\chi _{r}^{2}(\Delta )$ the non central chi-square
distribution with $r$ degrees of freedom and non-centrality parameter $%
\Delta $, we can state the following theorem.

\begin{theorem}
	\label{THM:PowerCont_Wald} We have
	
	\begin{enumerate}
		\item[i)] $W_{N}(\widehat{\boldsymbol{\beta }}_{\lambda })\underset{%
			N\rightarrow \infty }{\overset{L}{\longrightarrow }}\chi _{r}^{2}\left(
		\Delta _{1}\right) $ under $H_{1,N}$ given in (\ref{3.5}), with 
		$\Delta _{1}=\boldsymbol{d}^{T}\boldsymbol{L}\left\{ \boldsymbol{L}^{T}%
		\boldsymbol{\Sigma }_{\lambda }\left( \boldsymbol{\beta }_{0}\right) 
		\boldsymbol{L}\right\} ^{-1}\boldsymbol{L}^{T}\boldsymbol{d}$. 
		
		\item[ii)] $W_{N}(\widehat{\boldsymbol{\beta }}_{\lambda })\underset{%
			N\rightarrow \infty }{\overset{L}{\longrightarrow }}\chi _{r}^{2}\left(
		\Delta _{2}\right) $ under $H_{1,N}^{\ast }$ given in (\ref{3.6}), with 
		$\Delta _{2}=\boldsymbol{\delta }^{T}\left\{\boldsymbol{L}^{T}\boldsymbol{%
			\Sigma }_{\lambda }\left( \boldsymbol{\beta }_{0}\right) \boldsymbol{L}%
		\right\} ^{-1}\boldsymbol{\delta }$. 
	\end{enumerate}
\end{theorem}

\begin{proof}
	We have 
	\begin{align*}
	\boldsymbol{L}^{T}\widehat{\boldsymbol{\beta }}_{\lambda }-\boldsymbol{h}& =%
	\boldsymbol{L}^{T}\boldsymbol{\beta }_{N}-\boldsymbol{h+L}^{T}(\widehat{%
		\boldsymbol{\beta }}_{\lambda }-\boldsymbol{\beta }_{N}) \\
	& =\boldsymbol{L}^{T}\boldsymbol{\beta }_{0}+\boldsymbol{L}^{T}N^{-1/2}%
	\boldsymbol{\boldsymbol{d}}-\boldsymbol{h+L}^{T}(\widehat{\boldsymbol{\beta }%
	}_{\lambda }-\boldsymbol{\beta }_{N}) \\
	& =\boldsymbol{L}^{T}N^{-1/2}\boldsymbol{\boldsymbol{d+L}}^{T}(\widehat{%
		\boldsymbol{\beta }}_{\lambda }-\boldsymbol{\beta }_{N}).
	\end{align*}%
	Therefore, 
	\begin{equation*}
	\boldsymbol{L}\widehat{\boldsymbol{\beta }}_{\lambda }-\boldsymbol{h}=%
	\boldsymbol{L}^{T}N^{-1/2}\boldsymbol{\boldsymbol{d+L}}^{T}(\widehat{%
		\boldsymbol{\beta }}_{\lambda }-\boldsymbol{\beta }_{N}).
	\end{equation*}%
	We know, under $H_{1,N}$, that 
	\begin{equation*}
	\sqrt{N}(\widehat{\boldsymbol{\beta }}_{\lambda }-\boldsymbol{\beta }_{N})%
	\overset{%
		\mathcal{L}%
	}{\underset{N\rightarrow \infty }{\longrightarrow }}%
	\mathcal{N}%
	\left( \mathbf{0}_{d(k+1)},\boldsymbol{\Sigma }_{\lambda }\left( \boldsymbol{%
		\beta }_{N}\right) \right)
	\end{equation*}%
	and $\boldsymbol{\beta }_{N}\underset{N\rightarrow \infty }{\longrightarrow }%
	\boldsymbol{\beta }_{0}$. Therefore
	\begin{equation*}
	\sqrt{N}(\widehat{\boldsymbol{\beta }}_{\lambda }-\boldsymbol{\beta }_{N})%
	\overset{%
		\mathcal{L}%
	}{\underset{N\rightarrow \infty }{\longrightarrow }}%
	\mathcal{N}%
	\left( \mathbf{0}_{d(k+1)},\boldsymbol{\Sigma }_{\lambda }(\boldsymbol{\beta 
	}_{0})\right)
	\end{equation*}%
	and 
	\begin{equation*}
	\sqrt{N}(\boldsymbol{L}^{T}\widehat{\boldsymbol{\beta }}_{\lambda }-%
	\boldsymbol{h})\overset{%
		\mathcal{L}%
	}{\underset{N\rightarrow \infty }{\longrightarrow }}%
	\mathcal{N}%
	\left( \boldsymbol{L}^{T}\boldsymbol{\boldsymbol{d}},\boldsymbol{L}^{T}%
	\boldsymbol{\Sigma }_{\lambda }(\boldsymbol{\beta }_{0})\boldsymbol{L}%
	\right) .
	\end{equation*}%
	But we know: \textquotedblleft If $\boldsymbol{Z}\in N\left( \boldsymbol{\mu
		,\Sigma }\right) $, $\boldsymbol{\Sigma }$ is a symmetric projection of rank 
	$k$ and $\boldsymbol{\Sigma \mu }=\boldsymbol{\mu }$, then $\boldsymbol{Z}%
	^{T}\boldsymbol{Z}$ is a chi-square distribution with $k$ degrees of freedom
	and non-centrality parameter $\boldsymbol{\mu }^{T}\boldsymbol{\mu }$%
	\textquotedblright .\ The quadratic form is 
	\begin{equation*}
	W_{N}(\widehat{\boldsymbol{\beta }}_{\lambda })=\boldsymbol{Z}^{T}%
	\boldsymbol{Z}
	\end{equation*}%
	with 
	\begin{equation*}
	\boldsymbol{Z}=\sqrt{N}\left\{ \boldsymbol{L}^{T}\boldsymbol{M}_{N,\lambda }(%
	\widehat{\beta }_{\lambda })\boldsymbol{L}\right\} ^{-1/2}(\boldsymbol{L}^{T}%
	\widehat{\boldsymbol{\beta }}_{\lambda }-\boldsymbol{h})
	\end{equation*}%
	and 
	\begin{equation*}
	\boldsymbol{Z}\overset{%
		\mathcal{L}%
	}{\underset{N\rightarrow \infty }{\longrightarrow }}%
	\mathcal{N}%
	\left( \left( \boldsymbol{L}^{T}\boldsymbol{\Sigma }_{\lambda }(\boldsymbol{%
		\beta }_{0})\boldsymbol{L}\right) ^{-1/2}\boldsymbol{L}^{T}\boldsymbol{d},%
	\boldsymbol{I}_{r}\right) .
	\end{equation*}%
	Hence, the result of i) is immediately verified  and the non-centrality
	parameter takes the form
	\begin{equation*}
	\boldsymbol{d}^{T}\boldsymbol{L}\left\{ \boldsymbol{L}^{T}\boldsymbol{\Sigma }%
	_{\lambda }(\boldsymbol{\beta }_{0})\boldsymbol{L}\right\} ^{-1}\boldsymbol{L}%
	^{T}\boldsymbol{d}.
	\end{equation*}%
	For ii), we can follow (\ref{3.8}).
\end{proof}

\newpage
\subsection{\protect\normalsize Proofs of the Theorems in the main paper}

\noindent \textit{Proof of Theorem 1 of the main paper: }

\textit{Let }$\boldsymbol{y}_{i}^{\ast }=(y_{i1},\ldots ,y_{id})^{T}$ be the
reduced version of $\boldsymbol{y}_{i}$\ and%
\begin{equation*}
\mathcal{Y}%
^{\ast }=\boldsymbol{0}_{d}\cup \{\boldsymbol{e}_{j,d}\}_{j=1}^{d},
\end{equation*}%
the sample space of the (reduced)  response vector, where $\boldsymbol{e}_{j,d}$ is
the $j$-th column of the identity matrix of order $d$. From Theorem 3.1 of
Ghosh and Basu (2013), we get the first part on consistency as well as 
\begin{equation*}
\boldsymbol{\Omega }_{N,\lambda }^{-1/2}\left( \boldsymbol{\beta }%
_{0}\right) \boldsymbol{\Psi }_{N,\lambda }\left( \boldsymbol{\beta }%
_{0}\right) \sqrt{N}(\widehat{\boldsymbol{\beta }}_{\lambda }-\boldsymbol{%
	\beta }_{0})\overset{%
	\mathcal{L}%
}{\underset{N\rightarrow \infty }{\longrightarrow }}%
\mathcal{N}%
\left( \boldsymbol{0}_{d(k+1)},\boldsymbol{I}_{d(k+1)}\right) ,
\end{equation*}%
where 
\begin{equation*}
\boldsymbol{\Psi }_{N,\lambda }\left( \boldsymbol{\beta }\right) =\frac{1}{N}%
\sum\limits_{i=1}^{N}\sum_{\boldsymbol{y}_{i}^{\ast }\in 
	\mathcal{Y}%
	^{\ast }}\boldsymbol{u}_{i}(\boldsymbol{y}_{i}^{\ast },\boldsymbol{\beta })%
\boldsymbol{u}_{i}(\boldsymbol{y}_{i}^{\ast },\boldsymbol{\beta }%
)^{T}f_{i}^{\lambda +1}(\boldsymbol{y}_{i}^{\ast },\boldsymbol{\beta }),
\end{equation*}
\begin{equation*}
\boldsymbol{\Omega }_{N,\lambda }\left( \boldsymbol{\beta }\right) =\frac{1}{%
	N}\sum\limits_{i=1}^{N}\left\{ \sum_{\boldsymbol{y}_{i}^{\ast }\in 
	\mathcal{Y}%
	^{\ast }}\boldsymbol{u}_{i}(\boldsymbol{y}_{i}^{\ast },\boldsymbol{\beta })%
\boldsymbol{u}_{i}(\boldsymbol{y}_{i}^{\ast },\boldsymbol{\beta })^{T}f_{i}(%
\boldsymbol{y}_{i}^{\ast },\boldsymbol{\beta })^{1+2\lambda }-\boldsymbol{%
	\xi }_{i,\lambda }\left( \boldsymbol{\beta }_{0}\right) \boldsymbol{\xi }%
_{i,\lambda }\left( \boldsymbol{\beta }_{0}\right) ^{T}\right\} ,
\end{equation*}%
and 
\begin{equation*}
\boldsymbol{\xi }_{i,\lambda }\left( \boldsymbol{\beta }\right) =\sum_{%
	\boldsymbol{y}_{i}^{\ast }\in 
	\mathcal{Y}%
	^{\ast }}\boldsymbol{u}_{i}(\boldsymbol{y}_{i}^{\ast },\boldsymbol{\beta }%
)f_{i}(\boldsymbol{y}_{i}^{\ast },\boldsymbol{\beta })^{\lambda +1},
\end{equation*}%
with $\boldsymbol{u}_{i}(\boldsymbol{y}_{i}^{\ast },\boldsymbol{\beta })=%
\frac{\partial }{\partial \boldsymbol{\beta }}\ln f_{i}(\boldsymbol{y}%
_{i}^{\ast },\boldsymbol{\beta })=\left( \boldsymbol{y}_{i}^{\ast }-%
\boldsymbol{\pi }_{i}^{\ast }\left( \boldsymbol{\beta }\right) \right)
\otimes \boldsymbol{x}_{i}$. The calculations of the matrix $\boldsymbol{\Psi }_{N,\lambda }\left( \boldsymbol{\beta }\right)$ are as follows%
\allowdisplaybreaks
\begin{align*}
\boldsymbol{\Psi }_{N,\lambda }\left( \boldsymbol{\beta }\right) & =\frac{1}{%
	N}\sum\limits_{i=1}^{N}\left\{ \sum_{j=1}^{d}\boldsymbol{u}_{i}(\boldsymbol{e}%
_{j,d},\boldsymbol{\beta })\boldsymbol{u}_{i}(\boldsymbol{e}_{j,d},%
\boldsymbol{\beta })^{T}f_{i}^{\lambda +1}(\boldsymbol{e}_{j,d},\boldsymbol{%
	\beta })+\boldsymbol{u}_{i}(\boldsymbol{0}_{d},\boldsymbol{\beta })%
\boldsymbol{u}_{i}(\boldsymbol{0}_{d},\boldsymbol{\beta })^{T}f_{i}^{\lambda
	+1}(\boldsymbol{0}_{d},\boldsymbol{\beta })\right\} \\
& =\frac{1}{N}\sum\limits_{i=1}^{N}\left[ \sum_{j=1}^{d}\left\{ (\boldsymbol{%
	e}_{j,d}-\boldsymbol{\pi }_{i}^{\ast }\left( \boldsymbol{\beta }\right)
)\otimes \boldsymbol{x}_{i}\right\} \left\{ (\boldsymbol{e}_{j,d}-\boldsymbol{%
	\pi }_{i}^{\ast }\left( \boldsymbol{\beta }\right) )\otimes \boldsymbol{x}%
_{i}\right\} ^{T}\pi _{ij}^{\lambda +1}\left( \boldsymbol{\beta }\right)
\right. \\
& \left. +\left\{ \boldsymbol{\pi }_{i}^{\ast }\left( \boldsymbol{\beta }%
\right) \otimes \boldsymbol{x}_{i}\right\} \left\{ \boldsymbol{\pi }_{i}^{\ast
}\left( \boldsymbol{\beta }\right) \otimes \boldsymbol{x}_{i}\right\} ^{T}\pi
_{i,d+1}^{\lambda +1}\left( \boldsymbol{\beta }\right) \right] \\
& =\frac{1}{N}\sum\limits_{i=1}^{N}\left[ \sum_{j=1}^{d}\left\{ \pi
_{ij}^{\lambda +1}\left( \boldsymbol{\beta }\right) (\boldsymbol{e}_{j,d}-%
\boldsymbol{\pi }_{i}^{\ast }\left( \boldsymbol{\beta }\right) )(\boldsymbol{%
	e}_{j,d}-\boldsymbol{\pi }_{i}^{\ast }\left( \boldsymbol{\beta }\right) )^{T}%
\right\} \otimes \boldsymbol{x}_{i}\boldsymbol{x}_{i}^{T}\right. \\
& +\left. \left\{ \pi _{i,d+1}^{\lambda +1}\left( \boldsymbol{\beta }\right) 
\boldsymbol{\pi }_{i}^{\ast }\left( \boldsymbol{\beta }\right) \boldsymbol{%
	\pi }_{i}^{\ast }\left( \boldsymbol{\beta }\right) ^{T}\right\} \otimes 
\boldsymbol{x}_{i}\boldsymbol{x}_{i}^{T}\right]\\
& =\frac{1}{N}\sum\limits_{i=1}^{N}\left[ \sum_{j=1}^{d}\left[ \pi
_{ij}^{\lambda +1}\left( \boldsymbol{\beta }\right) \left\{\boldsymbol{e}_{j,d}-%
\boldsymbol{\pi }_{i}^{\ast }\left( \boldsymbol{\beta }\right) \right\}\left\{\boldsymbol{%
	e}_{j,d}-\boldsymbol{\pi }_{i}^{\ast }\left( \boldsymbol{\beta }\right) \right\}^{T}%
\right] \right. \\
& +\left. \pi _{i,d+1}^{\lambda +1}\left( \boldsymbol{\beta }\right) 
\boldsymbol{\pi }_{i}^{\ast }\left( \boldsymbol{\beta }\right) \boldsymbol{%
	\pi }_{i}^{\ast }\left( \boldsymbol{\beta }\right) ^{T}\right] \otimes 
\boldsymbol{x}_{i}\boldsymbol{x}_{i}^{T} \\
& =\frac{1}{N}\sum\limits_{i=1}^{N}\left\{ \sum_{j=1}^{d}\pi _{ij}^{\lambda
	+1}\left( \boldsymbol{\beta }\right) \boldsymbol{e}_{j,d}\boldsymbol{e}%
_{j,d}^{T}-\left( \sum_{j=1}^{d}\pi _{ij}^{\lambda +1}\left( \boldsymbol{%
	\beta }\right) \boldsymbol{e}_{j,d}\right) \boldsymbol{\pi }_{i}^{\ast
}\left( \boldsymbol{\beta }\right) ^{T}\right. \\
& -\boldsymbol{\pi }_{i}^{\ast }\left( \boldsymbol{\beta }\right) \left(
\sum_{j=1}^{d}\pi _{ij}^{\lambda +1}\left( \boldsymbol{\beta }\right) 
\boldsymbol{e}_{j,d}\right) ^{T}+\left( \sum_{j=1}^{d}\pi _{ij}^{\lambda
	+1}\left( \boldsymbol{\beta }\right) \right) \boldsymbol{\pi }_{i}^{\ast
}\left( \boldsymbol{\beta }\right) \boldsymbol{\pi }_{i}^{\ast }\left( 
\boldsymbol{\beta }\right) ^{T} \\
& \left. +\pi _{i,d+1}^{\lambda +1}\left( \boldsymbol{\beta }\right) 
\boldsymbol{\pi }_{i}^{\ast }\left( \boldsymbol{\beta }\right) \boldsymbol{%
	\pi }_{i}^{\ast }\left( \boldsymbol{\beta }\right) ^{T}\right\} \otimes 
\boldsymbol{x}_{i}\boldsymbol{x}_{i}^{T} \\
& =\frac{1}{N}\sum\limits_{i=1}^{N}\left[ \mathrm{diag}^{\lambda +1}\{
\boldsymbol{\pi }_{i}^{\ast }\left( \boldsymbol{\beta }\right) \}-\boldsymbol{%
	\pi }_{i}^{\ast (\lambda +1)}\left( \boldsymbol{\beta }\right) \boldsymbol{%
	\pi }_{i}^{\ast }\left( \boldsymbol{\beta }\right) ^{T}-\boldsymbol{\pi }%
_{i}^{\ast }\left( \boldsymbol{\beta }\right) \boldsymbol{\pi }_{i}^{\ast
	(\lambda +1)}\left( \boldsymbol{\beta }\right) ^{T}\right. \\
& \left. +\left( \sum_{j=1}^{d+1}\pi _{ij}^{\lambda +1}\left( \boldsymbol{%
	\beta }\right) \right) \boldsymbol{\pi }_{i}^{\ast }\left( \boldsymbol{\beta 
}\right) \boldsymbol{\pi }_{i}^{\ast }\left( \boldsymbol{\beta }\right) ^{T}%
\right] \otimes \boldsymbol{x}_{i}\boldsymbol{x}_{i}^{T}.
\end{align*}%
For the matrix $\Omega_{N, \lambda}$, we note that
\begin{align*}
\boldsymbol{\xi }_{i,\lambda }\left( \boldsymbol{\beta }\right) & =\sum_{%
	\boldsymbol{y}_{i}^{\ast }\in 
	\mathcal{Y}%
	^{\ast }}\boldsymbol{u}_{i}(\boldsymbol{y}_{i}^{\ast },\boldsymbol{\beta }%
)f_{i}^{\lambda +1}(\boldsymbol{y}_{i}^{\ast },\boldsymbol{\beta }) \\
& =\sum_{j=1}^{d}\boldsymbol{u}_{i}(\boldsymbol{e}_{j,d},\boldsymbol{\beta }%
)f_{i}^{\lambda +1}(\boldsymbol{e}_{j,d},\boldsymbol{\beta })+\boldsymbol{u}%
_{i}(\boldsymbol{0}_{d},\boldsymbol{\beta })f_{i}^{\lambda +1}(\boldsymbol{0}%
_{d},\boldsymbol{\beta }) \\
& =\sum_{j=1}^{d}\pi _{ij}^{\lambda +1}\left( \boldsymbol{\beta }\right) \{
\boldsymbol{e}_{j,d}-\boldsymbol{\pi }_{i}^{\ast }\left( \boldsymbol{\beta }%
\right) \}\otimes \boldsymbol{x}_{i}-\pi _{i,d+1}^{\lambda +1}\left( 
\boldsymbol{\beta }\right) \boldsymbol{\pi }_{i}^{\ast }\left( \boldsymbol{%
	\beta }\right) \otimes \boldsymbol{x}_{i} \\
& =\left[ \boldsymbol{\pi }_{i}^{\ast (\lambda +1)}\left( \boldsymbol{\beta }%
\right) -\left\{ \sum_{j=1}^{d+1}\pi _{ij}^{\lambda +1}\left( \boldsymbol{%
	\beta }\right) \right\} \boldsymbol{\pi }_{i}^{\ast }\left( \boldsymbol{\beta 
}\right) \right] \otimes \boldsymbol{x}_{i}.
\end{align*}%
So in the expression of $\boldsymbol{\Omega }_{N,\lambda }\left( \boldsymbol{%
	\beta }\right)$, we obtain%
\begin{align*}
& \boldsymbol{\xi }_{i,\lambda }\left( \boldsymbol{\beta }\right) 
\boldsymbol{\xi }_{i,\lambda }\left( \boldsymbol{\beta }\right) ^{T} \\
& =\left[ \left\{ \boldsymbol{\pi }_{i}^{\ast (\lambda +1)}\left( 
\boldsymbol{\beta }\right) -\left( \sum_{j=1}^{d+1}\pi _{ij}^{\lambda
	+1}\left( \boldsymbol{\beta }\right) \right) \boldsymbol{\pi }_{i}^{\ast
}\left( \boldsymbol{\beta }\right) \right\} \left\{ \boldsymbol{\pi }%
_{i}^{\ast (\lambda +1)}\left( \boldsymbol{\beta }\right) -\left(
\sum_{j=1}^{d+1}\pi _{ij}^{\lambda +1}\left( \boldsymbol{\beta }\right)
\right) \boldsymbol{\pi }_{i}^{\ast }\left( \boldsymbol{\beta }\right) %
\right\} ^{T}\right] \\
& \otimes \boldsymbol{x}_{i}\boldsymbol{x}_{i}^{T} \\
& =\left\{ \boldsymbol{\pi }_{i}^{\ast (\lambda +1)}\left( \boldsymbol{\beta 
}\right) \boldsymbol{\pi }_{i}^{\ast (\lambda +1)}\left( \boldsymbol{\beta }%
\right) ^{T}\mathrm{-}\left( \sum_{j=1}^{d+1}\pi _{ij}^{\lambda +1}\left( 
\boldsymbol{\beta }\right) \right) \boldsymbol{\pi }_{i}^{\ast }\left( 
\boldsymbol{\beta }\right) \boldsymbol{\pi }_{i}^{\ast (\lambda +1)}\left( 
\boldsymbol{\beta }\right) ^{T}\right. \\
& \left. -\left( \sum_{j=1}^{d+1}\pi _{ij}^{\lambda +1}\left( \boldsymbol{%
	\beta }\right) \right) \boldsymbol{\pi }_{i}^{\ast (\lambda +1)}\left( 
\boldsymbol{\beta }\right) \boldsymbol{\pi }_{i}^{\ast }\left( \boldsymbol{%
	\beta }\right) ^{T}+\left( \sum_{j=1}^{d+1}\pi _{ij}^{\lambda +1}\left( 
\boldsymbol{\beta }\right) \right) ^{2}\boldsymbol{\pi }_{i}^{\ast }\left( 
\boldsymbol{\beta }\right) \boldsymbol{\pi }_{i}^{\ast }\left( \boldsymbol{%
	\beta }\right) ^{T}\right\} \otimes \boldsymbol{x}_{i}\boldsymbol{x}_{i}^{T},
\end{align*}%
from which the desired expression  for $\boldsymbol{\Omega }_{N,\lambda }\left( 
\boldsymbol{\beta }\right) $  can be obtained in a straightforward manner. 

\newpage

\noindent \textit{Proof of Theorem 2 of the main paper:}

We have $\boldsymbol{L}^{T}\widehat{\boldsymbol{\beta }}_{\lambda }-%
\boldsymbol{h}=\boldsymbol{L}^{T}(\widehat{\boldsymbol{\beta }}_{\lambda }-%
\boldsymbol{\beta }_{0})$ and 
\begin{equation*}
\sqrt{N}(\widehat{\boldsymbol{\beta }}_{\lambda }-\boldsymbol{\beta }_{0})%
\overset{%
	\mathcal{L}%
}{\underset{N\rightarrow \infty }{\longrightarrow }}%
\mathcal{N}%
\left( \boldsymbol{0}_{d(k+1)},\boldsymbol{\Sigma }_{\lambda }\left( 
\boldsymbol{\beta }_{0}\right) \right) ,
\end{equation*}%
with $\boldsymbol{\Sigma }_{\lambda }\left( \boldsymbol{\beta }_{0}\right)
=\lim_{N\rightarrow \infty }\boldsymbol{M}_{N,\lambda }(\widehat{\beta }%
_{\lambda })=\boldsymbol{J}_{\lambda }^{-1}\left( \boldsymbol{\beta }%
_{0}\right) \boldsymbol{V}_{\lambda }\left( \boldsymbol{\beta }_{0}\right) 
\boldsymbol{J}_{\lambda }^{-1}\left( \boldsymbol{\beta }_{0}\right) $.
Therefore 
\begin{equation*}
\sqrt{N}\left( \boldsymbol{L}^{T}\widehat{\boldsymbol{\beta }}_{\lambda }-%
\boldsymbol{h}\right) \overset{%
	\mathcal{L}%
}{\underset{N\rightarrow \infty }{\longrightarrow }}%
\mathcal{N}%
\left( \boldsymbol{0}_{r},\boldsymbol{L^{T}\Sigma }_{\lambda }\left( 
\boldsymbol{\beta }_{0}\right) \boldsymbol{L}\right)
\end{equation*}%
and the asymptotic distribution of $W_{N}(\widehat{\boldsymbol{\beta }}%
_{\lambda })$ is a chi-square distribution with $r$ degrees of freedom
because 
\begin{equation*}
\boldsymbol{L}^{T}\boldsymbol{\Sigma }_{\lambda }\left( \boldsymbol{\beta }_{0}\right) 
\boldsymbol{L}\left\{ \boldsymbol{L}^{T}\boldsymbol{J}_{\lambda }^{-1}\left( 
\boldsymbol{\beta }_{0}\right) \boldsymbol{V}_{\lambda }\left( \boldsymbol{%
	\beta }_{0}\right) \boldsymbol{J}_{\lambda }^{-1}\left( \boldsymbol{\beta }%
_{0}\right) \boldsymbol{L}\right\} ^{-1}=\boldsymbol{I}_{r}.
\end{equation*}%

\linespread{1.2} \clearpage

\newpage
\subsection{Tables}\label{tables}
\setlength\tabcolsep{5pt}
\renewcommand{\arraystretch}{0.9}

\begin{table}[th]
	\caption{Minimum density power divergence estimators for the mammography data}
	\label{table:mamo_1}\centering\vspace{0.5cm} 
	\begin{tabular}{lrrrrrr}
		\hline
		$\lambda $ & $\widehat{\beta }_{01,\lambda }$ & $\widehat{\beta }%
		_{SYMPT1,1,\lambda }$ & $\widehat{\beta }_{SYMPT2,1,\lambda }$ & $\widehat{%
			\beta }_{SYMPT3,1,\lambda }$ & $\widehat{\beta }_{HIST,1,\lambda }$ & $%
		\widehat{\beta }_{BSE,1,\lambda }$ \\ \hline
		0 & -1.5787 & 1.1096 & 1.4157 & 0.3151 & 1.0685 & 1.0620 \\ 
		0.1 & -1.6511 & 1.1699 & 1.3808 & 0.3249 & 1.0978 & 1.0706 \\ 
		0.2 & -1.7270 & 1.2352 & 1.3607 & 0.3482 & 1.1193 & 1.1048 \\ 
		0.3 & -1.7944 & 1.3256 & 1.3284 & 0.3585 & 1.1653 & 1.1745 \\ 
		0.4 & -1.9113 & 1.3119 & 1.2913 & 0.3458 & 1.2427 & 1.2872 \\ 
		0.5 & -2.0629 & 1.3763 & 1.2666 & 0.3371 & 1.3288 & 1.3704 \\ 
		0.6 & -2.1549 & 1.4995 & 1.2402 & 0.3293 & 1.3743 & 1.3928 \\ 
		0.7 & -2.2114 & 1.6137 & 1.2164 & 0.3188 & 1.4001 & 1.3812 \\ 
		0.8 & -2.2511 & 1.7101 & 1.1955 & 0.3044 & 1.4186 & 1.3456 \\ 
		0.9 & -2.2846 & 1.786 & 1.1754 & 0.2858 & 1.4294 & 1.301 \\ 
		1 & -2.3139 & 1.8401 & 1.156 & 0.2646 & 1.4343 & 1.2574 \\ \hline
		$\lambda $ & $\widehat{\beta }_{DETC1,1,\lambda }$ & $\widehat{\beta }%
		_{DETC2,1,,\lambda }$ & $\widehat{\beta }_{PB,1,\lambda }$ & $\widehat{\beta 
		}_{02,\lambda }$ & $\widehat{\beta }_{SYMPT1,2,\lambda }$ & $\widehat{\beta }%
		_{SYMPT2,2,\lambda }$ \\ \hline
		0 & -0.6689 & 0.236 & 0.1485 & 1.4446 & -1.3661 & -0.9276 \\ 
		0.1 & -0.5903 & 0.2141 & 0.1548 & 1.3131 & -1.4693 & -0.9747 \\ 
		0.2 & -0.5265 & 0.1889 & 0.1616 & 1.1912 & -1.6064 & -1.0177 \\ 
		0.3 & -0.4563 & 0.1729 & 0.1650 & 1.0880 & -1.7353 & -1.0571 \\ 
		0.4 & -0.3068 & 0.1688 & 0.1724 & 1.0336 & -1.7867 & -1.0618 \\ 
		0.5 & 0.0080 & 0.1472 & 0.1832 & 0.9214 & -1.7641 & -1.0453 \\ 
		0.6 & 0.2614 & 0.1151 & 0.1919 & 0.8251 & -1.7182 & -1.0412 \\ 
		0.7 & 0.4167 & 0.0842 & 0.1986 & 0.7629 & -1.6787 & -1.0418 \\ 
		0.8 & 0.5100 & 0.0523 & 0.2041 & 0.7219 & -1.6453 & -1.0422 \\ 
		0.9 & 0.5619 & 0.0186 & 0.2101 & 0.6907 & -1.6174 & -1.0393 \\ 
		1 & 0.5878 & -0.0137 & 0.2164 & 0.6675 & -1.5848 & -1.0309 \\ \hline
		$\lambda $ & $\widehat{\beta }_{SYMPT3,2,\lambda }$ & $\widehat{\beta }%
		_{HIST,2,\lambda }$ & $\widehat{\beta }_{BSE,2,\lambda }$ & $\widehat{\beta }%
		_{DETC1,2,\lambda }$ & $\widehat{\beta }_{DETC2,2,\lambda }$ & $\widehat{%
			\beta }_{PB,2,\lambda }$ \\ \hline
		0 & -0.2155 & -0.3005 & -0.2298 & -1.5812 & -0.6534 & -0.0710 \\ 
		0.1 & -0.1978 & -0.2761 & -0.2517 & -1.7736 & -0.7238 & -0.0540 \\ 
		0.2 & -0.1724 & -0.2540 & -0.2386 & -1.8798 & -0.7922 & -0.0392 \\ 
		0.3 & -0.1610 & -0.1980 & -0.1797 & -1.9429 & -0.8414 & -0.0310 \\ 
		0.4 & -0.1711 & -0.1776 & -0.0967 & -1.9037 & -0.8815 & -0.0243 \\ 
		0.5 & -0.1819 & -0.1364 & -0.0562 & -1.6614 & -0.9420 & -0.0119 \\ 
		0.6 & -0.1949 & -0.0954 & -0.0709 & -1.4194 & -1.0084 & -0.0012 \\ 
		0.7 & -0.2120 & -0.0677 & -0.0914 & -1.2394 & -1.0703 & 0.0065 \\ 
		0.8 & -0.2337 & -0.0484 & -0.1127 & -1.0941 & -1.1362 & 0.0125 \\ 
		0.9 & -0.2594 & -0.0387 & -0.1358 & -0.9758 & -1.2097 & 0.0186 \\ 
		1 & -0.2876 & -0.0362 & -0.1559 & -0.8796 & -1.2855 & 0.0245 \\ \hline
	\end{tabular}
\end{table}

\clearpage
\begin{table}[th]
	\caption{Minimum density power divergence estimators for the mammography data in
		absence of outliers (distinct covariate indices 1, 3, 17, 35, 75, 81, and 102)}
	\label{table:mamo_2}\centering\vspace{0.5cm} 
	\begin{tabular}{lrrrrrr}
		\hline
		$\lambda $ & $\widehat{\beta }_{01,\lambda }$ & $\widehat{\beta }%
		_{SYMPT1,1,\lambda }$ & $\widehat{\beta }_{SYMPT2,1,\lambda }$ & $\widehat{%
			\beta }_{SYMPT3,1,\lambda }$ & $\widehat{\beta }_{HIST,1,\lambda }$ & $%
		\widehat{\beta }_{BSE,1,\lambda }$ \\ \hline
		0 & -1.8925 & 1.4266 & 1.3669 & 0.3256 & 1.1116 & 1.6248 \\ 
		0.1 & -1.9181 & 1.4846 & 1.3406 & 0.3303 & 1.1473 & 1.6814 \\ 
		0.2 & -1.9530 & 1.4341 & 1.3124 & 0.3184 & 1.2141 & 1.7322 \\ 
		0.3 & -2.0763 & 1.3997 & 1.3008 & 0.3252 & 1.299 & 1.8002 \\ 
		0.4 & -2.1994 & 1.4703 & 1.2857 & 0.334 & 1.3638 & 1.8316 \\ 
		0.5 & -2.2740 & 1.5918 & 1.2690 & 0.3386 & 1.4018 & 1.8104 \\ 
		0.6 & -2.3082 & 1.7108 & 1.2544 & 0.3381 & 1.4271 & 1.7601 \\ 
		0.7 & -2.3281 & 1.7978 & 1.2461 & 0.3338 & 1.4402 & 1.6934 \\ 
		0.8 & -2.2511 & 1.7101 & 1.1955 & 0.3044 & 1.4186 & 1.3456 \\ 
		0.9 & -2.3616 & 1.8963 & 1.2391 & 0.3148 & 1.5180 & 1.5919 \\ 
		1 & -2.3914 & 1.9706 & 1.2271 & 0.2943 & 1.5283 & 1.4638 \\ \hline
		$\lambda $ & $\widehat{\beta }_{DETC1,1,\lambda }$ & $\widehat{\beta }%
		_{DETC2,1,,\lambda }$ & $\widehat{\beta }_{PB,1,\lambda }$ & $\widehat{\beta 
		}_{02,\lambda }$ & $\widehat{\beta }_{SYMPT1,2,\lambda }$ & $\widehat{\beta }%
		_{SYMPT2,2,\lambda }$ \\ \hline
		0 & -0.3568 & 0.2915 & 0.1793 & 1.5959 & -5.4674 & -1.3728 \\ 
		0.1 & -0.3155 & 0.2298 & 0.1806 & 1.5181 & -4.5081 & -1.2968 \\ 
		0.2 & -0.2328 & 0.1896 & 0.1801 & 1.4976 & -3.7857 & -1.2253 \\ 
		0.3 & -0.0204 & 0.1807 & 0.1862 & 1.4111 & -3.1953 & -1.1531 \\ 
		0.4 & 0.2309 & 0.1591 & 0.195 & 1.2880 & -2.7678 & -1.0862 \\ 
		0.5 & 0.4338 & 0.1286 & 0.2015 & 1.1915 & -2.4453 & -1.0338 \\ 
		0.6 & 0.5676 & 0.0986 & 0.2045 & 1.1360 & -2.2082 & -0.995 \\ 
		0.7 & 0.6432 & 0.0721 & 0.207 & 1.1174 & -2.0422 & -0.9631 \\ 
		0.8 & 0.6927 & 0.0440 & 0.2094 & 1.1162 & -1.9315 & -0.9355 \\ 
		0.9 & 0.7257 & 0.0441 & 0.2038 & 1.1455 & -1.8532 & -0.8839 \\ 
		1 & 0.7227 & 0.0063 & 0.2089 & 1.1462 & -1.7657 & -0.8674 \\ \hline
		$\lambda $ & $\widehat{\beta }_{SYMPT3,2,\lambda }$ & $\widehat{\beta }%
		_{HIST,2,\lambda }$ & $\widehat{\beta }_{BSE,2,\lambda }$ & $\widehat{\beta }%
		_{DETC1,2,\lambda }$ & $\widehat{\beta }_{DETC2,2,\lambda }$ & $\widehat{%
			\beta }_{PB,2,\lambda }$ \\ \hline
		0 & -0.3246 & -0.4810 & 0.1545 & -4.4961 & -0.5605 & -0.0756 \\ 
		0.1 & -0.3105 & -0.4198 & 0.2659 & -3.6716 & -0.6706 & -0.0711 \\ 
		0.2 & -0.3158 & -0.3607 & 0.3383 & -3.0346 & -0.7600 & -0.0736 \\ 
		0.3 & -0.3074 & -0.3102 & 0.4025 & -2.4822 & -0.8003 & -0.0681 \\ 
		0.4 & -0.3004 & -0.2556 & 0.4036 & -2.0631 & -0.8515 & -0.0573 \\ 
		0.5 & -0.2997 & -0.2048 & 0.3459 & -1.7300 & -0.9100 & -0.049 \\ 
		0.6 & -0.3054 & -0.1604 & 0.2685 & -1.4688 & -0.9677 & -0.0454 \\ 
		0.7 & -0.3175 & -0.1420 & 0.2092 & -1.2732 & -1.0205 & -0.0437 \\ 
		0.8 & -0.3369 & -0.1310 & 0.1663 & -1.1216 & -1.0774 & -0.0434 \\ 
		0.9 & -0.3511 & -0.0924 & 0.1685 & -1.0401 & -1.093 & -0.0527 \\ 
		1 & -0.3806 & -0.1013 & 0.1065 & -0.8946 & -1.1813 & -0.0492 \\ \hline
	\end{tabular}
\end{table}

\clearpage
\begin{table}[h]
	\caption{Estimated mean deviations of the predicted probabilities with
		respect to the relative frequencies for the mammography data in presence and
		the absence of outliers}
	\label{tab:mamo_ef}\centering\vspace{0.5cm} 
	\begin{tabular}{lllllllll}\hline
		& \multicolumn{4}{c}{With Outliers} & \multicolumn{4}{c}{Without Outliers}
		\\ \hline
		$\lambda $ & $\widehat{md}_{1}(\widehat{\beta }_{\lambda }^{pres})$ & $%
		\widehat{md}_{2}(\widehat{\beta }_{\lambda }^{pres})$ & $\widehat{md}_{3}(%
		\widehat{\beta }_{\lambda }^{pres})$ & $\overline{\widehat{md}}(\widehat{%
			\beta }_{\lambda }^{pres})$ & $\widehat{md}_{1}(\widehat{\beta }_{\lambda
		}^{abs})$ & $\widehat{md}_{2}(\widehat{\beta }_{\lambda }^{abs})$ & $%
		\widehat{md}_{3}(\widehat{\beta }_{\lambda }^{abs})$ & $\overline{\widehat{md%
		}}(\widehat{\beta }_{\lambda }^{abs})$ \\ \hline
		0 & 0.2497 & 0.1627 & 0.2244 & 0.2123 & 0.2333 & 0.1424 & 0.2130 & 0.1962 \\ 
		0.1 & 0.2481 & 0.1611 & 0.2241 & 0.2111 & 0.2331 & 0.1443 & 0.2126 & 0.1967
		\\ 
		0.2 & 0.2465 & 0.1600 & 0.2234 & 0.2100 & 0.2333 & 0.1459 & 0.2128 & 0.1973
		\\ 
		0.3 & 0.2446 & 0.1592 & 0.2221 & 0.2086 & 0.2325 & 0.1479 & 0.212 & 0.1975
		\\ 
		0.4 & 0.2424 & 0.1590 & 0.2202 & 0.2072 & 0.2311 & 0.1497 & 0.2105 & 0.1971
		\\ 
		0.5 & 0.2393 & 0.1596 & 0.2174 & 0.2054 & 0.2299 & 0.1511 & 0.2097 & 0.1969
		\\ 
		0.6 & 0.2368 & 0.1600 & 0.2156 & 0.2041 & 0.2293 & 0.1521 & 0.2096 & 0.1970
		\\ 
		0.7 & 0.2354 & 0.1604 & 0.2147 & 0.2035 & 0.2293 & 0.1530 & 0.2097 & 0.1973
		\\ 
		0.8 & 0.2346 & 0.1609 & 0.2144 & 0.2033 & 0.2293 & 0.1539 & 0.2102 & 0.1978
		\\ 
		0.9 & 0.2342 & 0.1614 & 0.2145 & 0.2034 & 0.2296 & 0.1546 & 0.211 & 0.1984
		\\ 
		1 & 0.2341 & 0.1621 & 0.2147 & 0.2036 & 0.2302 & 0.1559 & 0.2121 & 0.1994 \\ 
		\hline
	\end{tabular}%
\end{table}
\clearpage
\begin{table}[h]
	\caption{Mean deviations between predicted probabilities in the presence and
		the absence of outliers, mammography data}
	\label{tab:mamo_rob}\centering\vspace{0.5cm} 
	\begin{tabular}{lllll}
		\hline
		$\lambda $ & $\widehat{md}_{1}(\widehat{\beta }_{\lambda }^{pres},\widehat{%
			\beta }_{\lambda }^{abs})$ & $\widehat{md}_{2}(\widehat{\beta }_{\lambda
		}^{pres},\widehat{\beta }_{\lambda }^{abs})$ & $\widehat{md}_{3}(\widehat{%
			\beta }_{\lambda }^{pres},\widehat{\beta }_{\lambda }^{abs})$ & $\overline{%
			\widehat{md}}(\widehat{\beta }_{\lambda }^{pres},\widehat{\beta }_{\lambda
		}^{abs})$ \\ \hline
		0 & 0.0423 & 0.0295 & 0.0244 & 0.0321 \\ 
		0.1 & 0.0382 & 0.0254 & 0.0231 & 0.0289 \\ 
		0.2 & 0.0344 & 0.0231 & 0.0218 & 0.0264 \\ 
		0.3 & 0.0329 & 0.0205 & 0.0216 & 0.0250 \\ 
		0.4 & 0.0313 & 0.0177 & 0.0209 & 0.0233 \\ 
		0.5 & 0.0265 & 0.0165 & 0.0164 & 0.0198 \\ 
		0.6 & 0.0224 & 0.0159 & 0.0129 & 0.0171 \\ 
		0.7 & 0.0194 & 0.0157 & 0.0108 & 0.0153 \\ 
		0.8 & 0.0177 & 0.0159 & 0.0101 & 0.0146 \\ 
		0.9 & 0.0167 & 0.0170 & 0.011 & 0.0149 \\ 
		1 & 0.0151 & 0.0168 & 0.0100 & 0.0139 \\ \hline
	\end{tabular}
	\vspace{-0.5cm}
\end{table}

\clearpage
\begin{table}[h]
	\caption{ Minimum density power divergence estimators for the mammography data
		in the presence and absence of outliers (observations $3,101,108,116,131$ and $%
		136$) for the liver enzyme data}\centering\vspace{0.5cm} 
	\begin{tabular}{lrrrrrrrr}\hline
		& \multicolumn{4}{c}{With Outliers} & \multicolumn{3}{c}{Without Outliers}
		&  \\ \hline
		$\lambda $ & $\widehat{\beta }_{01,\lambda }$ & $\widehat{\beta }%
		_{11,\lambda }$ & $\widehat{\beta }_{21,\lambda }$ & $\widehat{\beta }%
		_{31,\lambda }$ & $\widehat{\beta }_{01,\lambda }$ & $\widehat{\beta }%
		_{11,\lambda }$ & $\widehat{\beta }_{21,\lambda }$ & $\widehat{\beta }%
		_{31,\lambda }$ \\ \hline
		0 & -11.3562 & -5.6253 & 8.7928 & -2.6974 & -13.7400 & -6.0662 & 9.8171 & 
		-2.9305 \\ 
		0.1 & -11.9074 & -8.2968 & 12.0207 & -3.7737 & -13.7704 & -11.9607 & 17.4944
		& -6.8940 \\ 
		0.2 & -13.9013 & -8.8064 & 13.0695 & -4.1168 & -14.3529 & -12.2444 & 17.8452
		& -6.8712 \\ 
		0.3 & -12.2642 & -9.7634 & 13.6570 & -3.9701 & -15.3101 & -11.7853 & 16.7262
		& -5.0802 \\ 
		0.5 & -19.7314 & -8.4364 & 13.9747 & -4.3414 & -32.3112 & -10.4116 & 19.7954
		& -7.8408 \\ 
		0.7 & -20.2738 & -8.2416 & 13.8016 & -4.1064 & -21.8045 & -9.2977 & 16.1602
		& -6.3050 \\ 
		0.9 & -15.9960 & -9.8613 & 14.085 & -3.1146 & -25.5161 & -9.3104 & 16.2868 & 
		-5.1840 \\ \hline
		$\lambda $ & $\widehat{\beta }_{02,\lambda }$ & $\widehat{\beta }%
		_{12,\lambda }$ & $\widehat{\beta }_{22,\lambda }$ & $\widehat{\beta }%
		_{32,\lambda }$ & $\widehat{\beta }_{02,\lambda }$ & $\widehat{\beta }%
		_{12,\lambda }$ & $\widehat{\beta }_{22,\lambda }$ & $\widehat{\beta }%
		_{32,\lambda }$ \\ \hline
		0 & 6.0838 & -6.6832 & 6.2269 & -2.3655 & 5.7237 & -7.0629 & 6.6722 & -2.3064
		\\ 
		0.1 & 6.7921 & -8.9577 & 8.7550 & -3.1893 & 7.9148 & -12.4001 & 13.1958 & 
		-5.6809 \\ 
		0.2 & 7.7503 & -9.7336 & 9.4087 & -3.4241 & 8.0888 & -12.9810 & 13.2378 & 
		-4.8349 \\ 
		0.3 & 7.5506 & -10.4368 & 10.1915 & -3.4620 & 8.1183 & -12.1101 & 11.7976 & 
		-3.5897 \\ 
		0.5 & 10.1578 & -10.1466 & 9.1995 & -3.3291 & 12.4949 & -11.5529 & 10.2186 & 
		-3.6970 \\ 
		0.7 & 10.3416 & -9.9920 & 8.9252 & -3.145 & 9.4703 & -9.9344 & 9.3999 & 
		-3.6731 \\ 
		0.9 & 8.8516 & -10.7851 & 9.8726 & -2.7130 & 11.0658 & -10.3511 & 8.8947 & 
		-2.6976 \\ \hline
		$\lambda $ & $\widehat{\beta }_{03,\lambda }$ & $\widehat{\beta }%
		_{13,\lambda }$ & $\widehat{\beta }_{23,\lambda }$ & $\widehat{\beta }%
		_{33,\lambda }$ & $\widehat{\beta }_{03,\lambda }$ & $\widehat{\beta }%
		_{13,\lambda }$ & $\widehat{\beta }_{23,\lambda }$ & $\widehat{\beta }%
		_{33,\lambda }$ \\ \hline
		0 & -5.8146 & -1.4866 & 2.0449 & 1.1873 & -6.7381 & -2.1027 & 2.6535 & 1.5256
		\\ 
		0.1 & -5.9428 & -2.1120 & 2.7574 & 1.1689 & -6.8763 & -2.8613 & 3.6163 & 
		1.3250 \\ 
		0.2 & -6.0419 & -2.2494 & 2.9226 & 1.1696 & -6.9709 & -2.8159 & 3.5948 & 
		1.3484 \\ 
		0.3 & -6.0837 & -2.7485 & 3.3916 & 1.2724 & -6.4149 & -3.5845 & 4.1865 & 
		1.5139 \\ 
		0.5 & -7.2092 & -2.7256 & 3.5795 & 1.3007 & -8.5194 & -3.1961 & 4.2202 & 
		1.4999 \\ 
		0.7 & -7.9578 & -2.9992 & 3.9547 & 1.3981 & -6.8510 & -2.9633 & 3.6638 & 
		1.4231 \\ 
		0.9 & -6.7793 & -3.4111 & 4.1136 & 1.4402 & -8.4754 & -3.6218 & 4.5947 & 
		1.5834 \\ \hline
	\end{tabular}%
	\label{table:enzyme_estimator}
\end{table}

\clearpage
\begin{table}[h]
	\caption{ Estimated mean deviations of the predicted probabilities with
		respect to the relative frequencies in the presence and the absence of outliers, for
		the liver enzyme data}\centering\vspace{0.5cm} 
	\begin{tabular}{llllll}\hline
		& \multicolumn{5}{c}{With Outliers} \\ \hline
		$\lambda $ & $\widehat{md}_{1}(\widehat{\beta }_{\lambda }^{pres})$ & $%
		\widehat{md}_{2}(\widehat{\beta }_{\lambda }^{pres})$ & $\widehat{md}_{3}(%
		\widehat{\beta }_{\lambda }^{pres})$ & $\widehat{md}_{4}(\widehat{\beta }%
		_{\lambda }^{pres})$ & $\overline{\widehat{md}}(\widehat{\beta }_{\lambda
		}^{pres})$ \\ \hline
		0 & 0.0810 & 0.1041 & 0.1922 & 0.1867 & 0.1410 \\ 
		0.1 & 0.0732 & 0.0899 & 0.1831 & 0.1694 & 0.1289 \\ 
		0.2 & 0.0677 & 0.0834 & 0.1791 & 0.1656 & 0.1240 \\ 
		0.3 & 0.0702 & 0.0831 & 0.1753 & 0.1586 & 0.1218 \\ 
		0.5 & 0.0596 & 0.0750 & 0.1653 & 0.1548 & 0.1137 \\ 
		0.7 & 0.0597 & 0.0748 & 0.1603 & 0.1497 & 0.1111 \\ 
		0.9 & 0.0635 & 0.0761 & 0.1657 & 0.1484 & 0.1135 \\ \hline
		& \multicolumn{5}{c}{Without Outliers} \\ \hline
		$\lambda $ & $\widehat{md}_{1}(\widehat{\beta }_{\lambda }^{abs})$ & $%
		\widehat{md}_{2}(\widehat{\beta }_{\lambda }^{abs})$ & $\widehat{md}_{3}(%
		\widehat{\beta }_{\lambda }^{abs})$ & $\widehat{md}_{4}(\widehat{\beta }%
		_{\lambda }^{abs})$ & $\overline{\widehat{md}}(\widehat{\beta }_{\lambda
		}^{abs})$ \\ \hline
		0 & 0.0739 & 0.0970 & 0.1759 & 0.1704 & 0.1293 \\ 
		0.1 & 0.0656 & 0.0796 & 0.1643 & 0.1512 & 0.1152 \\ 
		0.2 & 0.0651 & 0.0772 & 0.1663 & 0.151 & 0.1149 \\ 
		0.3 & 0.0622 & 0.0748 & 0.1657 & 0.1479 & 0.1126 \\ 
		0.5 & 0.0501 & 0.0667 & 0.153 & 0.1447 & 0.1036 \\ 
		0.7 & 0.0557 & 0.0758 & 0.1609 & 0.154 & 0.1116 \\ 
		0.9 & 0.0540 & 0.0707 & 0.1540 & 0.1431 & 0.1054 \\ \hline
	\end{tabular}%
	\label{table:enzyme_efficiency}
\end{table}

\clearpage
\begin{table}[h]
	\caption{ Mean deviations between predicted probabilities in the presence and
		the absence of outliers, for the liver enzyme data}
	\label{table:enzyme_robust}\centering\vspace{0.5cm} 
	\begin{tabular}{llllll}
		\hline
		$\lambda $ & $\widehat{md}_{1}(\widehat{\beta }_{\lambda }^{pres},\widehat{%
			\beta }_{\lambda }^{abs})$ & $\widehat{md}_{2}(\widehat{\beta }_{\lambda
		}^{pres},\widehat{\beta }_{\lambda }^{abs})$ & $\widehat{md}3(\widehat{\beta 
		}_{\lambda }^{pres},\widehat{\beta }_{\lambda }^{abs})$ & $\widehat{md}_{4}(%
		\widehat{\beta }_{\lambda }^{pres},\widehat{\beta }_{\lambda }^{abs})$ & $%
		\overline{\widehat{md}}(\widehat{\beta }_{\lambda }^{pres},\widehat{\beta }%
		_{\lambda }^{abs})$ \\ \hline
		0 & 0.0115 & 0.0103 & 0.0256 & 0.0242 & 0.0179 \\ 
		0.1 & 0.0130 & 0.0223 & 0.0316 & 0.0306 & 0.0244 \\ 
		0.2 & 0.0146 & 0.0187 & 0.0222 & 0.0226 & 0.0195 \\ 
		0.3 & 0.0126 & 0.0139 & 0.017 & 0.0177 & 0.0153 \\ 
		0.5 & 0.0138 & 0.0121 & 0.0193 & 0.0152 & 0.0151 \\ 
		0.7 & 0.0120 & 0.0134 & 0.0129 & 0.0113 & 0.0124 \\ 
		0.9 & 0.0167 & 0.0170 & 0.0228 & 0.0128 & 0.0173 \\ \hline
	\end{tabular}%
\end{table}

\newpage
\section{Supplementary Materials: R codes used for the computations in Section \ref{S6.2}}
The following R codes are is provided to help reader to implementation the proposed MDPDE and the corresponding   Wald-type tests for any practical application. These codes were used for our simulation studies presented in Section 6.2 of the main paper.

\begin{lstlisting}


############################################### 

#DEFINITION OF FUNCTIONS

PLRM<-function(x,beta){
###############################################
#INPUTS
#x: vector of explanatory variables, dim--1xk
#beta: vector of unknown parameters: dxk
#OUTPUT
#pisol: vector of probabilities of all categories, dim--1x(d+1)
###############################################
k=length(x)
d=length(beta)/k
betat=matrix(beta,k,d)
pisol=exp(x@%*%@betat)/(1+sum(exp(x@%*%@betat)))
pisol=c(pisol,1-sum(pisol))
return(pisol)
}


DISTANCE<-function(beta,lambda,X,Y,n){
###############################################
#INPUTS
#beta: vector of unknown parameters: dxk
#lambda: tuning parameter of power divergences
#X: matrix of explanatory variables, dim--Nxk
#Y: matrix of responses, dim--kx(d+1)
#n:vector of sizes n(x), dim--1xN
#OUTPUT
#divsol: power divergence distance
############################################### 
N=dim(X)[1]
d=length(beta)/(dim(X)[2])
#initialize
divsol=0 
#DPD (lambda!=0)
if (lambda !=0){
for (i in 1:N){
pis=PLRM(X[i,],beta)
for (l in 1:(d+1)){
divsol=divsol+ (n[i]*pis[l])^(1+lambda)-(1+1/lambda)*Y[i,l]*(n[i]*pis[l])^(lambda)
}
} 
}
#MLE (lambda==0)
else{
for (i in 1:N){
pis=PLRM(X[i,],beta)
for (l in 1:(d+1)){
if (Y[i,l]==0){divsol=divsol+0}
else{divsol=divsol+ Y[i,l]*log(Y[i,l]/(n[i]*pis[l]))}
}
}
}
return(divsol)
}


TEST<-function(beta,lambda,Y,X, PIS,L){
###############################################
#INPUTS
#beta: vector of unknown parameters: dxk
#lambda: tuning parameter of power divergences
#Y: matrix of responses, dim--kx(d+1)
#X: matrix of explanatory variables, dim--Nxk
#PIS: matrix of probabilities under the model
#L: matrix of the contrast
#OUTPUT
#M: matrix of the Wald-type test
############################################### 
d=dim(Y)[2]-1
N=dim(X)[1]
tol= 1e-12 #to discard non-invertible matrices
#initialize
PHIM=matrix(0,length(beta),length(beta))
OMEGAM=matrix(0,length(beta),length(beta))
for (i in 1:N){
#definition of auxiliar variables
clambda=sum(PIS[i,]^(1+lambda))
cdoslambda=sum(PIS[i,]^(1+2*lambda))
etai=kronecker(PIS[i,1:d]^(1+lambda)-clambda*PIS[i,1:d],X[i,])
#computation of matrices Phi and Omega
PHIM=PHIM+kronecker(diag(PIS[i,1:d]^(1+lambda))-PIS[i,1:d]@%*%@t(PIS[i,1:d]^(1+lambda))-PIS[i,1:d]^(1+lambda)@%*%@t(PIS[i,1:d])+ clambda*PIS[i,1:d]@%*%@t(PIS[i,1:d]),X[i,]@%*%@t(X[i,]))
OMEGAM=OMEGAM+kronecker(diag(PIS[i,1:d]^(1+2*lambda))-PIS[i,1:d]@%*%@t(PIS[i,1:d]^(1+2*lambda))-PIS[i,1:d]^(1+2*lambda)@%*%@t(PIS[i,1:d])+ cdoslambda*PIS[i,1:d]@%*%@t(PIS[i,1:d]),X[i,]@%*%@t(X[i,]))-etai@%*%@t(etai)
}
if (det(PHIM)<tol){M=NULL}
else{
aux=t(L)@%*%@solve(PHIM)@%*%@OMEGAM@%*%@solve(PHIM)@%*%@L
if(det(aux)<tol){M=NULL}
else{M=solve(aux)}
}
return(M)
}


############################################### 

#SIMULATION

set.seed(2509) #set the seed
Samples=1000 #number of samples in the simulation

beta0=c(0,-0.9,0.1,0.6,-1.2,0.8) #true value of the variable vector
N=c(100,125,150,175,200,225,250,275,300,325) #samples sizes considered
lambda=c(0,0.1,0.2,0.3,0.5,0.7,0.9) #lambdas considered for the MDPDEs
out=0.95 #percentage of pure data

# We will contrast the null hypothesis H0:beta02=0.6 with 1 degree of freedom
L=matrix(0,6,1)
L[4,1]=1
ll=matrix(0.6,1,1)

valid_samples=matrix(Samples,length(lambda),length(N)) #initialize the number of valid samples

level_model_without_outliers=matrix(0,length(lambda),length(N)) #matrix of levels for the pure case
level_model_with_outliers=matrix(0,length(lambda),length(N)) #matrix of levels for the contaminated case
initial=c(0,0,0,0,0,0) #initialize the value of the parameter vector

#LOOP FOR THE SIMULATION
for (l in 1:length(N)){
#definition of parameters which depend on the size of N
n=vector(,N[l])+1
PIS=matrix(0,N[l],3)
PIS1=PIS
PIS2=PIS
Y=matrix(0,N[l],3)  
for (s in 1:Samples){
#definition of parameters which vary for each sample
#generation fo the matrix X
x0=vector(,N[l])+1
x1=rnorm(N[l],0,1)
x2=rnorm(N[l],0,1)
X=cbind(x0,x1,x2)
for (i in 1:N[l]){
PIS[i,]=PLRM(X[i,],beta0)
Y[i,]=rmultinom( 1,n[i],PIS[i,]) 
}
#change in the response variable for outliers
Y2=Y
for (i in floor(out*N[l]):N[l]){
aux=c(Y[i,3],Y[i,1],Y[i,2])
Y2[i,]=aux 
}
#calculation for each lambda
for (r in 1:length(lambda)){
#estimation of the vector of parameters beta
aux_model_without_outliers=tryCatch(nlm(DISTANCE,initial,lambda[r],X,Y,n ),error=function(sol){sol$code=3;return(sol)})
aux_model_with_outliers=tryCatch(nlm(DISTANCE,initial,lambda[r],X,Y2,n ),error=function(sol){sol$code=3;return(sol)})
beta_model_without_outliers=aux_model_without_outliers$estimate
beta_model_with_outliers=aux_model_with_outliers$estimate
#if the minimization does not converge: we delete this sample
if(length(beta_model_without_outliers)==0){valid_samples[r,l]=valid_samples[r,l]-1}
else if(length(beta_model_with_outliers)==0){valid_samples[r,l]=valid_samples[r,l]-1}
else{
#estimated probabilities, depend on lambda
for (i in 1:N[l]){
PIS1[i,]=PLRM(X[i,],beta_model_without_outliers)
PIS2[i,]=PLRM(X[i,],beta_model_with_outliers)
}
#finally, the estimation of the test matrices
test_without_outliers=TEST(beta_model_without_outliers,lambda[r],Y,X, PIS1,L)
test_with_outliers=TEST(beta_model_with_outliers,lambda[r],Y2,X, PIS2,L)
#what happens if it is not invertible?: we delete this sample
if (length(test_without_outliers)==0){valid_samples[r,l]=valid_samples[r,l]-1}
else if (length( test_with_outliers)==0){valid_samples[r,l]=valid_samples[r,l]-1}
#if everything's ok, let's compute the levels
else{
result_model_without_outliers=(t(L) @%*%@ beta_model_without_outliers -ll)@%*%@test_without_outliers @%*%@t((t(L) @%*%@ beta_model_without_outliers -ll))
result_model_with_outliers=(t(L) @%*%@ beta_model_with_outliers -ll)@%*%@test_with_outliers @%*%@t((t(L) @%*%@ beta_model_with_outliers -ll))
#the asymptotic distribution under the null hypothesis is a chi-square distribution with 1 degree of freedom
level_model_without_outliers[r,l]=level_model_without_outliers[r,l]+(result_model_without_outliers>qchisq(0.95, df=1))
level_model_with_outliers[r,l]=level_model_with_outliers[r,l]+(result_model_with_outliers>qchisq(0.95, df=1))
}
}
}
}  
}
#taking into account the valid samples, we obtain the final values
levels_without_outliers=level_model_without_outliers/valid_samples
levels_with_outliers=level_model_with_outliers/valid_samples









############################################### 

#RESULTS

#names of rows and cols
row.names(levels_without_outliers)=c(lambda)
row.names(levels_with_outliers)=c(lambda)
colnames(levels_without_outliers)=c(N)
colnames(levels_with_outliers)=c(N)

#print on screen
(levels_without_outliers)
(levels_with_outliers)

#plotting the results
x_axis=N

#plotting the pure case
level=levels_without_outliers 
par(mar=c(5.1, 4.1,5.1, 6.1), xpd=FALSE)
#print the levels of the MLE
plot(x_axis,level[1,],xlim=c(min(x_axis),max(x_axis)),ylim=c(0.03,0.07),xlab=expression(N),ylab=c("Level"),col=1,"o",lty=1, pch=1)
abline(a=NULL,b=NULL,v=x_axis,col="lightgray",  lty = 3)
abline(a=NULL,b=NULL,h=c(0:100)/(100),col="lightgray",  lty = 3)
abline(a=NULL,b=NULL,h=c(0.05),  lty = 3) #for the significance level alpha=0.05
#print the levels of the alternative MDPDEs
for (j in 2:length(lambda)){
lines(x_axis,level[j,],col=j,"o",lty=j, pch=j)
}
#print the legend
par(mar=c(5.1, 3.1,5.1, 6.1), xpd=TRUE)
legend("topright", title= expression(lambda),horiz=F, inset=c(-0.28,0), as.character(lambda), col = c(1:length(lambda)), lty = c(1:length(lambda)), ,pch = c(1:length(lambda)))

#plotting the contaminated case
level=levels_with_outliers 
par(mar=c(5.1, 4.1,5.1, 6.1), xpd=FALSE)
#print the levels of the MLE
plot(x_axis,level[1,],xlim=c(min(x_axis),max(x_axis)),ylim=c(0.03,max(level)),xlab=expression(N),ylab=c("Level"),col=1,"o",lty=1, pch=1)
abline(a=NULL,b=NULL,v=x_axis,col="lightgray",  lty = 3)
abline(a=NULL,b=NULL,h=c(0:10)/(10),col="lightgray",  lty = 3)
abline(a=NULL,b=NULL,h=c(0.05),  lty = 3) #for the significance level alpha=0.05
#print the levels of the alternative MDPDEs
for (j in 2:length(lambda)){
lines(x_axis,level[j,],col=j,"o",lty=j, pch=j)
}
#print the legend
par(mar=c(5.1, 3.1,5.1, 6.1), xpd=TRUE)
legend("topright", title= expression(lambda),horiz=F, inset=c(-0.28,0), as.character(lambda), col = c(1:length(lambda)), lty = c(1:length(lambda)), ,pch = c(1:length(lambda)))



\end{lstlisting}


\begin{thebibliography}{99}
	\bibitem{aerts} Aerts, S. and Haesbroeck, G. (2017). Robust asymptotic tests for the equality of multivariate coefficients of variation.  \textit{TEST} \textbf{26}(1), 163--187. 
	
	\bibitem{albert} Albert, A. and Harris, E.K. (1987). Multivariate
	interpretation of clinical laboratory data, New York: Marcel Dekker.
	
	
	\bibitem{b1} Basu, A., Shioya, H. and Park, C. (2011). \textit{The minimum
		distance approach. Monographs on Statistics and Applied Probability}. CRC
	Press, Boca Raton.
	
	
	
	\bibitem{ba} Basu, A., Mandal, A., Mart\'{\i}n, N. and Pardo, L. (2016).
	Generalized Wald-type tests based on minimum density power divergence
	estimators. \textit{Statistics} \textbf{50}, 1--26.
	
	
	\bibitem{aux2} Basu, A. Ghosh, A. Mandal,  Martin, N. and  Pardo, L. (2017a). A Wald-type test statistic for testing linear hypothesis in logistic regression models based on minimum density power divergence estimator. \textit{Electonic Journal of Statistics} \textbf{11}, 2741-2772.
	
	
	\bibitem{aux3}  Basu, A., Ghosh, A.,  Martin, N. and  Pardo, L. (2017b). Robust Wald-type tests for non-homogeneous observations based on minimum density power divergence estimator. \textit{ArXiv pre-print, arXiv:1707.02333 [stat.ME].}
	
	\bibitem{bu} Begg, C.B. and Gray, R. (1984). Calculations of polychotomous
	logistic regression estimates using individualized regressions. \textit{Biometrika} \textbf{\ 71}, 1-18.
	
	\bibitem{bert}{Bertens, L. C. M., Moons,  K. G. M., Rutten, F. H. van Mourik, I., Hoes, A. W. and Reitsma, J. B. (2015). A nomogram was developed
		to enhance the use of multinomial logistic regression modeling in diagnostic
		research. \textit{ Journal of Clinical Epidemiology} \textbf{\ 71}, 51--57.}
	
	\bibitem{biesel}{Biesheuvel, C.J., Vergouwe, Y., Steyerberg, E.W., Grobbee,
		D.E. and Moons, K.G. (2008). Polytomous logistic regression analysis could be
		applied more often in diagnostic research.  \textit{ Journal of Clinical Epidemiology} \textbf{61}, 125-134.}
	
	\bibitem{blizar}{Blizzard, L. and Hosmer, D. W. (2007). The Log Multinomial Regression Model for Nominal Outcomes with More than Two Attributes. \textit{Biometrical Journal} \textbf{49}, 889--902.}
	
	
	
	\bibitem{bulll}{Bull, S. B., Lewinger, J. P. and Lee, S. S. F. (2007).
		Confidence intervals for multinomial logistic regression in sparse data. 
		\textit{Statistics in Medicine} \textbf{26}, 903--918.}
	
	\bibitem{daniels}{Daniels, M. J. and Gatsonis, C. (1997). Hierarchical polytomous regression models with applications to health services research.  \textit{Statistics in Medicine} \textbf{16}, 2311--2325.}
	
	\bibitem{dey}{Dey, S., Raheem, E. and Lu, Z. (2016). Multilevel multinomial
		logistic regression model for identifying factors associated with anemia in
		children 6--59 months in northeastern states of India. \textit{Cogent
			Mathematics} \textbf{3}, 1-12.}
	
	\bibitem{drea}{Dreassi, D. (2007). Polytomous Disease Mapping to Detect
		Uncommon Risk Factors for Related Diseases. \textit{Biometrical Journal}  \textbf{49}, 520--529.}
	
	\bibitem{g3} Ghosh, A. and Basu, A. (2013). Robust Estimation for Independent
	but Non-Homogeneous Observations using Density Power Divergence with
	application to Linear Regression. \textit{Electonic Journal of Statistics} \textbf{7}, 2420--2456.
	
	\bibitem{gos4} Ghosh, A. and Basu, A. (2015). Robust estimation for non-homogeneous data and the selection of the optimal tuning parameter: the density power divergence approach. \textit{Journal of Applied Statitsics} \textbf{42}(9), 2056--2072.
	
	\bibitem{gos5} Ghosh, A. and Basu, A. (2016). Robust Estimation in Generalized
	Linear Models: The density power divergence. \textit{TEST}\textbf{\ 25},
	269-290.
	
	
	\bibitem{aux1}  Ghosh, A., and  Basu, A. (2017). Robust Bounded Influence Tests for Independent but Non-Homogeneous Observations. \textit{Statistica Sinica} doi:10.5705/ss.202015.0320.
	
	
	
	\bibitem{g16} Ghosh, A., Mandal, A., Mart\'{\i}n, N. and Pardo, L. (2016).
	Influence Analysis of Robust Wald-type Tests. \textit{Journal Multivariate Analysis} \textbf{147}, 102--126.
	
	\bibitem{gu5} Gupta, A. K., Kasturiratna, D., Nguyen, T. and Pardo, L.
	(2006). A new family of BAN estimators for polytomous logistic regression
	models based on $\phi$-divergence measures. \textit{Statistical Methods \& Applications} \textbf{15}, 159--176.
	
	
	
	\bibitem{mmm} Hayashi, F. (2000). \textit{Econometrics}. New Jersey:
	Princeton University Press.
	
	\bibitem{hosmer} Hosmer, D.W. and  Lemeshow, S. (2000). {\em Applied logistic regression}. Wiley, New York.
	
	\bibitem{huber} Huber, P. J. (1983). Minimax aspects of bounded-influence regression (with discussion). \textit{Journal of the American Statistical Association} \textbf{78}, 66-80.
	
	\bibitem{ke}{Ke, Y., Fu, B. and Zhang, W. (2016). Semi-varying coefficient
		multinomial logistic regression for disease progression risk prediction. 
		\textit{Statistics in Medicine} \textbf{35}, 4764--4778.}
	
	
	
	\bibitem{martin}  Mart\'{\i}n, N. (2015). Using Cook's distance in polytomous
	logistic regression. \textit{British Journal of Mathematical and Statistical Psychology} \textbf{68}, 84--115.
	
	
	\bibitem{munoz} Mu\~{n}oz-Pichardo, J. M. ; Enguix-Gonzalez, A.; Mu\~{n}%
	oz-Garc\'{\i}a, J. and Moreno-Rebollo, J. L. (2011). Infuence analysis on
	discriminant coordinates. \textit{Communications in Statistics - Simulation and Computation} \textbf{40, }793-807.
	
	\bibitem{nel} Nelder, J. A. and Wedderburn, R. W. M. (1972). \textit{%
		Generalized Linear Models. }London, Chapman \& Hall.
	
	\bibitem{bb} Pardo, L. (2005). \textit{Statistical Inference Based on
		Divergence Measures. Statistics: Texbooks and Monographs. }Chapman \&
	Hall/CRC, New York.
	
	\bibitem{pl} Plomteux, G. (1980). Multivariate analysis of an enzyme profile for the differential diagnosis of viral hepatitis. \textit{Clinical Chemistry} \textbf{26}, 1897–1899.
	
	\bibitem{rom} Rom, M. and Cohen, A. (1995). Estimation in the polytomous
	logistic regression models, \textit{Journal of Statistical Planning and Inference} \textbf{43}, 341-353.
	
	\bibitem{ro} Ronchetti, E. and Trojani, F. (2001). Robust inference with gmm
	estimators. \textit{Journal of Econometrics} \textbf{101}, 37--69.
	
	
	\bibitem{ro2 }Rousseeuw, P. J. and Ronchetti, E. (1979). The influence curve for tests. \textit{Research Report} \textbf{21}, Fachgruppe fur Statistik, ETH Zurich.
	
	\bibitem{v} Victoria-Feser, M. and Ronchetti, E. (1997). Robust estimation
	for grouped data. \textit{Journal of the American Statistical Association} 
	\textbf{92, }333--340.
	
	\bibitem{w} Wang, X. (2014). Modified generalized method of moments for a
	robust estimation of polytomous logistic model. \textit{PeerJ 2:e467 https://doi.org/10.7717/peerj.467}
	
	\bibitem{Warwick/Jones:2005} Warwick, J., and Jones, M.~C. (2005). Choosing a robustness tuning parameter. \emph{Journal of Statistical Computation and Simulation}, \textbf{75}, 581--588.
	
	
\end{thebibliography}
\end{document}